\documentclass[12pt,a4paper]{article}
\usepackage[utf8]{inputenc}
\usepackage[T1]{fontenc}
\usepackage{textcomp}

\usepackage{amsmath,amssymb,amsthm,amsfonts}
\usepackage[dvips]{graphicx}
\usepackage{epsfig}
\usepackage{caption}
\usepackage{subcaption}
\usepackage{float}
\usepackage{verbatim}
\usepackage[font=small]{caption}
\usepackage{soul}
\usepackage[normalem]{ulem}
\usepackage{physics}
\usepackage{multirow}
\usepackage{url}
\usepackage{enumerate}
\usepackage{physics}
\usepackage{mathtools}
\usepackage{lmodern}

\usepackage{color}
\usepackage{xcolor}

\def\a{\alpha}

\def\e{\epsilon}

\def\p{\partial}

\def\({\text{\huge (}}
\def\){\text{\huge )}}

\def\]{\text{\huge ]}}
\def\[{\text{\huge [}}

\newcommand\C{\hat{c}}
\newcommand\q{\hat{q}}
\newcommand\x{\hat{x}}
\newcommand\T{\hat{t}}

\newcommand{\bi}{\begin{itemize}}
\newcommand{\ei}{\end{itemize}}
\newcommand{\be}{\begin{equation}}
\newcommand{\ee}{\end{equation}}
\newcommand{\ba}{\begin{align}}
\newcommand{\ea}{\end{align}}

\newcommand\nc{\newcommand}
\nc\pad[2]{\frac{\p #1}{\p #2}} \nc\padd[2]{\frac{\p^2 #1}{\p
{#2}^2}} \nc\nd[2]{\frac{d #1}{d #2}} \nc\pat[2]{\frac{D #1}{D
#2}} \nc\ov{\overline} \nc\degree{^{\circ}} \nc\ord[1]{{\cal
O}(#1)} \nc\ra{\rightarrow} \nc\Ra{\Rightarrow} \nc\dint{{\mbox ~
d}}

\DeclareMathOperator{\Pe}{Pe^{-1}}       
\DeclareMathOperator{\Da}{Da}		

\newcommand{\bea}{\begin{eqnarray}}
\newcommand{\eea}{\end{eqnarray}}
\newcommand{\beas}{\begin{eqnarray*}}
\newcommand{\eeas}{\end{eqnarray*}}

\newtheorem{prop}{Proposition}

\title{Analysis of travelling-wave equations in sorption processes}

\author{M. Aguareles, J. Anglada-Lloveras,  E. Barrab\'es}

\begin{document}

\maketitle
\begin{abstract}
This work presents a mathematical model of an adsorption column to study the evolution of contaminant concentration and adsorbed quantity along the longitudinal axis of the filter. The model is formulated as a system of partial differential equations (PDEs) and analysed using a travelling-wave approach, which reduces the system to a second-order ordinary differential equation depending on the inverse P\'eclet number, typically a small parameter. 
By neglecting this parameter, the model is simplified via a singular perturbation to a leading-order approximation, which can be interpreted as a slow-fast system.
We rigorously justify this reduction by proving the persistence of the heteroclinic connection associated with the travelling wave. Using analytical continuation, we conclude that, at least for small values of the inverse P\'eclet number, the concentration profile transitions from a clean downstream state of the adsorbent matrix to fully upstream saturation. 
Numerical simulations are presented to validate the analytical results and to assess the accuracy of the reduced model. A sensitivity analysis demonstrates that the travelling-wave approximation remains remarkably robust for moderate values of the inverse P\'eclet number.

\end{abstract}
\section{Introduction}

The increasing concern over climate change and the necessity to reduce harmful gas emissions into the atmosphere call for the development of effective filtration systems to prevent environmental pollution by actively removing pollutants from gases and liquids. Adsorption columns have emerged as a crucial tool in capturing pollutants from various liquid or gaseous emissions. These systems consist of a column filled with a porous adsorbent material through which the contaminated fluid flows. In experimental settings, small tubes approximately 3 cm in height and 5 cm in diameter are filled with adsorbent material, typically a fine powder made of small grains. On the laboratory scale, one can assume that the pollutant molecules attach directly to the surfaces of these grains. This assumption is appropriate at the laboratory scale considered here, where the adsorbent particles are sufficiently small and surface effects dominate the adsorption process. Figure \ref{fig:column} illustrates this setup: polluted air is injected (at a fixed velocity, $\vb{u}$) at the inlet of the column (at $x=0$) and, as it flows through, the pollutant molecules adhere to the grains, thus releasing a clean fluid at the outlet (at $x=L$).
\begin{figure}
\begin{center}
	\includegraphics[height=6.5cm]{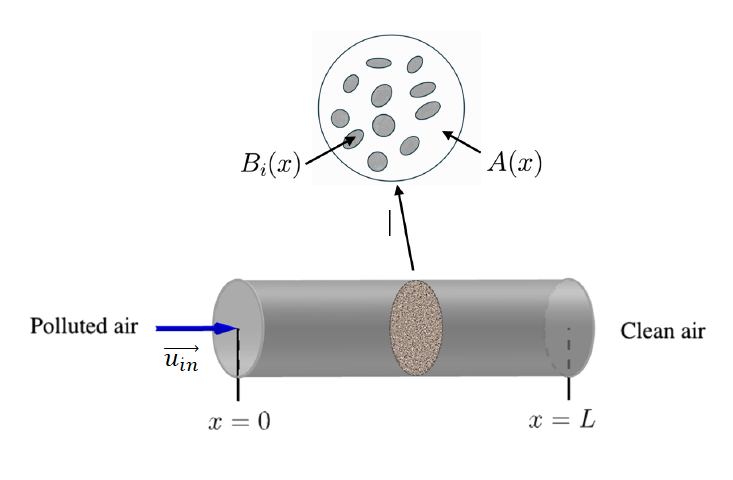}
\caption{Sketch of an adsorption column.} \label{fig:column}
\end{center}
\end{figure}

The development and analysis of mathematical models for adsorption columns are essential for enhancing and optimizing adsorption equipment. In the literature, one can find many semi-empirical models with a limited scope of applicability. 
For instance, in \cite{rok24}, the authors conduct a series of batch experiments to determine the adsorption rate. They use a set of ordinary differential equations to model the evolution of the concentration of palladium, nickel, and ruthenium in a liquid solvent. Batch experiments, also known as static systems, involve adding adsorbent material in the form of a fine powder to a solution containing components expected to be adsorbed by the grains of the adsorbent material. This mixture is continuously stirred, and the concentration evolution is periodically monitored. However, as noted in \cite{Tsing09}, the adsorption rates depend on the type of experiment, so the values obtained from batch experiments are generally not comparable to those obtained in a column experiment. Similarly, in \cite{Bai23}, the authors employ machine learning tools to infer the adsorption capacity of biochar to capture a specific type of antibiotic (tetracycline). Again, the models are based on specific experiments, and the validity of the results is restricted to this particular combination of molecule and adsorbent. In \cite{Bos19}, the authors highlight that, in certain cases, the Linear Driving Force model, which assumes that the adsorption rate is a linear function of the molecular concentration, is insufficient to describe the adsorption kinetics. Consequently, they derive a semi-empirical model to demonstrate this point. These examples highlight the importance of developing accurate and robust mathematical model for adsorption processes. In particular, such models are essential to infer parameters that cannot be directly measured, without relying exclusively on experimental fitting procedures.

Recently, many researchers have conducted comprehensive reviews of the field (see, for instance, \cite{Li18}, \cite{Patel2019} or \cite{Shafeeyan14}), leading to the creation of various mathematical models based on the physicochemical principles underlying adsorption processes. In \cite{Mondal19}, \cite{Myers20a}, \cite{Myers20b},\cite{AGUARELES2023}, and \cite{AUTON2024827}, for the first time, the authors rigorously derived a set of partial differential equations describing adsorption in a packed column, taking into account the fluid dynamics around and inside the adsorbent grains as well as providing the equations to include any possible heat transfer effects. These models are general enough to describe a wide range of applications, and they describe the correct way to upscale the results from the laboratory to the industrial scale. Also, in \cite{myers23}, the authors point out a set of errors that are commonly found in the literature due to a lack of consistency and mathematical rigour.

 Building on the models and simplifications developed in \cite{Myers20a}, \cite{Myers20b}, and \cite{AGUARELES2023}, other researchers have applied these models to address environmental challenges, such as mitigating pollution and toxicity in natural ecosystems (see, for instance, \cite{SOBOLEV}). Additionally, these models have been instrumental in the development of natural adsorbent materials, often derived from industrial waste, and in enhancing the efficiency of operating conditions in adsorption processes (see, for example, \cite{CHAJRI}, \cite{SADAIYAN}, and \cite{Valverde24}).

Two types of curves are typically obtained from experimental data to compare and assess the performance of adsorbents: isotherm curves and breakthrough curves. Isotherms are generated by injecting a constant flow of a contaminated fluid into the column until saturation is reached. At saturation, the concentration of the contaminant at the inlet matches the concentration at the outlet, as the adsorbent can no longer retain any more contaminant. At this point, the experiment is stopped, and the total mass of contaminant captured is measured. By conducting this experiment with different inlet concentrations while keeping the rest of the experimental conditions constant, a curve relating the inlet concentration to the total captured mass is produced. These curves, known as isotherms, are typically used to determine the adsorption-desorption rate and the maximum capacity of the adsorbent material. Moreover, it should be noted that isotherms only provide equilibrium information and do not describe the transient dynamics.

In contrast, the breakthrough curve is obtained by plotting the concentration at the outlet (breakthrough) as a function of time. This curve is generally used to compare numerical models with experimental data. In \cite{AGUARELES2023}, the authors derive a system of partial differential equations (PDEs) for the contaminant concentration in the fluid and for the amount of contaminant adsorbed at each point in the column over time, considering the physicochemical processes involved in adsorption and the fluid dynamics around the adsorbent grains. Using a travelling-wave approach, they reduce the system of PDEs to a set of ordinary differential equations (ODEs) for the contaminant concentration, which they further simplify into implicit algebraic equations for the breakthrough curves. These equations can be directly used to evaluate experiments. The approximate solutions derived are successfully validated by comparison with various experimental datasets. These approximate solutions rely on neglecting a small parameter (the inverse P\'eclet number), which represents a singular perturbation of the original system. However, in \cite{AGUARELES2023}, the original system of PDEs is not used to validate the approximate solution, which is instead used to infer the system parameters from two different data sets. In this paper, we aim to extend this work by rigorously proving that this approximation is indeed correct, i.e. the approximate system is, in fact, the limit of the original system of ODEs. Furthermore, we numerically solve the system of PDEs to show that travelling-wave profiles appear for a column with a typical laboratory  . We carry out these simulations for several different reaction exponents. Additionally, we perform a sensitivity analysis that demonstrates the robustness of the system. Specifically, we shall show that even for not so small values of the singular parameter, the agreement between the curves obtained by neglecting this parameter and those from the original system remains very good (in some cases we find relative errors of less than 3\% even for order one values of the inverse P\'eclet number). This is particularly relevant as it provides the necessary mathematical rigour, extending beyond mere comparison with specific numerical datasets and establishing the limits of validity for the valuable explicit expressions presented in \cite{AGUARELES2023}.

The paper is structured as follows: Section~\ref{sec:model} presents a detailed description of the model, which leads to a system of partial differential equations governing the evolution of the contaminant concentration. Additionally, numerical simulations are provided to demonstrate the validity of a travelling-wave approximation for a specific experiment and dataset. Section~\ref{sec:trav} derives a second-order ordinary differential equation for the travelling wave and rigorously proves the existence of solutions satisfying the prescribed boundary conditions. Section~\ref{sec:sens} conducts a systematic sensitivity analysis to determine the range of validity of the approximation obtained when the inverse P\'eclet number is neglected. Section~\ref{sec:pde} explores the existence of solutions of the system of partial differential equations in the form of travelling-wave fronts, providing substantial evidence supporting the accuracy of the travelling-wave approximation for a wide range of parameter values. Additionally, we briefly explore the effect of the initial conditions on the evolution of the concentration and adsorbed fraction. Finally, Section~\ref{sec:conc} concludes the findings of this research. 

All numerical computations were performed in MATLAB using double precision and standard built-in routines.

\section{Model description and numerical solutions}
\label{sec:model}
To mathematically describe adsorption columns, models are typically formulated as a system of partial differential equations (PDEs). These equations describe the contaminant concentration at any point in the fluid region of the column, $C$, and the amount of contaminant that has already been captured and adsorbed onto the surfaces of the solid grains, $C^\textrm{ad}$. Adsorption reactions can occur through simple Van der Waals forces or through chemical bonds between the contaminant molecule and the adsorbent material. The first type, usually known as physisorption, results in weaker adsorption forces, making the molecules more likely to detach and return to the fluid. The second type of adsorption mechanism is known as chemisorption, where adsorption occurs through stronger chemical bonds between the contaminant molecule and the adsorbent material. Adsorption reactions are usually of the form
\beas
\label{chemreac}
    \mathrm{m\, \mathcal{C} + n\, \mathcal{A}} \xrightleftharpoons[\kappa_\textrm{ad}]{\kappa_\textrm{de}}\mathrm{ Products } \, ,
\eeas
where m and n are the stoichiometric coefficients of the adsorption reaction, $\mathcal{C}$ is the contaminant molecule, and $\mathcal{A}$ is the adsorbent component. In the case of physisorption, m$=$n$=$1. The law of mass action states that the rate at which contaminant molecules attach to the adsorbent surface is the result of a balance between the adsorption and the desorption rates:
\begin{equation}
\label{eq:Cads}
    \pdv{C^\textrm{ad}}{t} =  \kappa_\textrm{ad}C^m(C^\textrm{ad}_{\textrm{sat}} - C^\textrm{ad})^n - \kappa_\textrm{de} (C^\textrm{sat}-C)^m(C_\textrm{ad})^n \, ,
	\end{equation}
where $C^\textrm{sat}, C^\textrm{ad}_\textrm{sat}$ are the saturation concentration of the contaminant in the fluid and in the adsorbent material, respectively. The exponents $m$ and $n$ are the global orders of the reaction, which in some cases coincide with the stoichiometric coefficients $\text{m}$, $\text{n}$ (see, for example, \cite{arnaut2006chemical} for a review of chemical kinetics laws). In general, the reaction rate coefficients, $\kappa_\textrm{ad}$ and $\kappa_\textrm{de}$, cannot be directly measured and they are fitted from experimental data. Also, they are known to differ depending on whether the experiment is conducted in batch or column mode (see \cite{Tsing09}). In adsorption columns, this is done by running experiments for a long time keeping a constant flux and concentration at the inlet, $c_\text{in}$, until the filter is saturated and an equilibrium is reached, i.e. the concentration at the outlet is the same as the one at the inlet and therefore the concentration is $c_\text{in}$ everywhere in the column. The final weight of the column, $M_f$, is compared with its initial weight when the filter is fresh, $M_i$, and the following equilibrium adsorbed fraction is obtained: 
\begin{equation*}
	q_e=\frac{M_f-M_i}{M_i}=\frac{(1-\phi)C^\textrm{ad}_e}{\rho_{b}}\, ,
\end{equation*}
where $\phi$ is the filter's void fraction (so $1-\phi$ is the filter's solid fraction) and $\rho_b$ is the density of the column measured as the total mass of adsorbent over the column volume.
The set of points $(c_\textrm{in},q_e)$ obtained from these experiments generates the isotherm curve of the adsorbent. Sips in \cite{Sips_1948} proposed isotherms of the form
\begin{equation}
	\label{eq:Sips}
	q_e = \frac{q_\textrm{max}(k_L c_\textrm{in}^m)^{1/n}}{1+(k_Lc_\text{in}^m)^{1/n}}\, ,
\end{equation}
where $q_\textrm{max}$ is the maximum adsorbed fraction of the column. By fitting the set of points $(c_\textrm{in},q_e)$ with the isotherm, $\kappa_L$ and $q_\textrm{max}$ are determined. We now show that this isotherm, in fact, corresponds to the equilibrium solutions of \eqref{eq:Cads} when the contaminant concentration is low enough. Indeed, if the concentration is far enough from its saturation value in the fluid, $(C^\textrm{sat}-C)\sim  C^\textrm{sat}$, equation \eqref{eq:Cads}, written in terms of the adsorbed fraction $q = (1-\phi)C^\textrm{ad}/\rho_b$, reads
\begin{alignat}{2}
\label{eq:q0}
   \pdv{q}{t} &=  k_\textrm{ad}C^m(q_\textrm{max} - q)^n - k_\textrm{de} C^m q^n \, ,\end{alignat}
where 
$$k_\text{de}=(C^\text{sat})^m\kappa_\text{de}\left(\frac{\rho_{b}}{1-\phi}\right)^{n-1},\quad k_\text{ad}=\kappa_\text{ad}\left(\frac{\rho_{b}}{1-\phi}\right)^{n-1}, $$ 
and whose equilibrium solutions, $q=q_e$, when the concentration is constant and fixed to its initial value, $C=c_\textrm{in}$, are given by \eqref{eq:Sips}, being $k_L=k_\textrm{ad}/k_\textrm{de}$. 

It is commonly assumed that, when dealing with trace amounts of a pollutant, the contaminant concentration in the fluid is low enough that its mass loss does not significantly impact the fluid flow. Consequently, the fluid velocity is assumed to remain constant and equal to that at the inlet. Under this assumption and by means of a volume average, the authors in \cite{AGUARELES2023} derive a system of partial differential equations for the main pollutant concentration and adsorbed fraction at each section of the filter, $c(x,t)$ and $q(x,t)$: 
\begin{subequations}
\label{eq:Sist_dim}
\begin{alignat}{2}
    \pdv{c}{t}+u_\text{in}\pdv{c}{x} &= D \pdv[2]{c}{x}-\frac{\rho_{b}}{\phi}\pdv{q}{t}\,, \quad && x\in(0,L),\quad t>0,\\
\label{eq:q}
    \pdv{q}{t} &=  k_\textrm{ad}c^m(q_{\textrm{max}} - q)^n - k_\textrm{de} q^n \, ,\quad&& x\in(0,L),\quad t>0 ,
\end{alignat}
where $L$ is the column length and $u_\text{in}$ is the constant inlet velocity. We note that, by definition, the adsorbed fraction satisfies $q(x,t)<q_e<1$. We also note that this model is only valid if the incompressibility assumption remains valid, which is the case in most air purification contexts where pressure drops are sufficiently low. As for the boundary conditions, a Dankwerts' condition is assumed at the inlet (see \cite{Danc53}), and a reduction of the Dankwerts' condition at the outlet (see \cite{Pear59}):
\bea
\label{eq:BC_dim}
u_\text{in} c_\text{in}=\left(u c-D_x\pad{c}{x}\right)\Bigg\vert_{x=0^+}  ,\qquad \pdv{c}{x}\Bigg\vert_{x=L^-} = 0\,,\quad \textrm{for $t>0$} .
\eea
Finally, to close the system, assuming that the filter is initially clean and that the air inside is free from contaminant, the initial conditions read 
\bea
c(x,0)=q(x,0)=0\, ,\quad \textrm{for $x\in(0,L)$}.
\eea
\end{subequations}

To reveal the effect of the different parameters of the system it is common practice to write equations \eqref{eq:Sist_dim} in non-dimensional form. Setting $c= c_\textrm{in}\C$, $q = q_\textrm{max} \q$, $x= \mathcal{L} \x$ and $t = \tau \T$, and defining
\begin{equation*}
   \mathcal{L}=\frac{\phi \tau u_\text{in} c_\textrm{in}}{\rho_b q_\textrm{max}}\,, \quad \tau = q_\textrm{max}^{1-n}(k_\text{ad} c_\textrm{in}^m + k_\text{de})^{-1}\,,\quad \Da=\frac{\mathcal{L}}{\tau u_\text{in}}\, ,\quad \Pe=\frac{D}{u_\text{in} \mathcal{L}}\,,\quad \alpha = \frac{k_\text{ad}c_\text{in}^m}{k_\text{ad}c_\text{in}^m + k_\text{de}}\,,
\end{equation*}
dropping the hats, system \eqref{eq:Sist_dim} becomes
\begin{subequations}
\label{eq:Sist_ND}
\begin{alignat}{2}
    \Da\pdv{c}{t}+\pdv{c}{x} &= \Pe\pdv[2]{c}{x}-\pdv{q}{t} \,, \quad && x\in(0,L),\quad t>0,\label{eq:c_dim}\\
    \pdv{q}{t} & =  \a c^m(1 - q)^n - (1-\a) q^n \, ,\quad&& x\in(0,L),\quad t>0 ,\label{eq:q_dim}\\
    \left(c-\Pe\pad{c}{x}\right)\Bigg\vert_{x=0^+}&=1\,, \qquad  \pdv{c}{x}\Bigg\vert_{x=L^-} =0 \quad && t>0 \, ,\label{eq:Dank}\\
       q(x,0)&=c(x,0)=0, && x\in(0,L)\, ,\label{eq:ic}
\end{alignat}    
\end{subequations}
where now $L$ represents the non-dimensional value of $L$, i.e., $L/\mathcal{L}$.
 The parameter $\Da>0$ is the Damk\"ohler number, which may be interpreted as the ratio of advection to reaction time-scales,  $\Pe>0$ is the inverse P\'eclet number, representing the relative importance of diffusion over advection and $\a\in[0,1]$ is an indicator of the relative importance of adsorption over desorption, so that $\a=1$ corresponds to a purely adsorbing system while $\a=0$ represents a purely desorbing one. In the physical systems we are considering, $\alpha \in (0, 1)$. The Damk\"ohler number ($\Da$) is typically extremely small, as the rate at which the contaminant is being adsorbed is significantly higher than the velocity at which the fluid flows along the column. Additionally, the diffusion of the pollutant in the fluid, in most applications, occurs at a significantly lower rate compared to the advection velocity. Consequently, the P\'eclet number becomes relatively large, so $\Pe$ is actually a small parameter. In fact, in most applications $\Da\ll \Pe\ll 1$ and this separation of scales reflects the fact that adsorption reactions typically occur much faster than advective transport, while molecular diffusion remains slow compared to advection in packed columns (refer to \cite{SULAYMON2014325} or \cite{Valverde24} and the references therein for further details).
 
In what follows, we show the results of a simulation to explore solutions of system \eqref{eq:Sist_ND}. The main goal is to show that the concentration and adsorbed fraction evolve like a travelling wave, using data from an actual laboratory experiment. In particular, we  use data for the adsorption of toluene extracted from \cite{CABRERACODONY2018565}, where: 
\begin{equation}
	\label{values}
			\phi=0.3357,\,\, u_\text{in}=0.13\, \text{m/s},\,\, c_\text{in}=2.835\, \text{kg/m}^3,\,\,\rho_b=377.25\, \text{kg/m}^3,\,\,\text{and}\,\, L=5.4 \cdot 10^{-3}\, \text{m},
\end{equation}
and where the adsorption exponents are known to be, in this case, $m=n=1$ (see \cite{CABRERACODONY2018565} and \cite{MYERS2023123660}). The remaining parameters, $k_\text{ad}$, $k_\text{de}$, and $q_\text{max}$, are fitted in \cite{MYERS2023123660}, where the authors obtain the following values: 
\begin{equation}
	\label{values2}
k_\text{ad}=1.13\, \text{m}^3\text{s}^{-1}\text{kg}^{-1},\,\, k_\text{de}=2.614\cdot 10^{-4} \,\text{s}^{-1},\,\,q_\text{max}=0.358. 
\end{equation}
Also, in \cite{MYERS2023123660}, the authors, based on the arguments presented in \cite{Lev99}, claim that diffusion in this experiment (which was not directly measured in \cite{CABRERACODONY2018565}) is of the order of $10^{-5}$. This value of the diffusion would lead to an inverse P\'eclet number of around $\Pe =0.1$. 

Figure \ref{Fig:figPDE1} illustrates the evolution of the concentration, $c(x,t)$, and adsorbed fraction, $q(x,t)$, obtained by solving system \eqref{eq:Sist_ND} implementing a Scharfetter-Gummel Discretization scheme, using the data provided in \eqref{values} and \eqref{values2} and with $\Pe =0.1$. Figure \ref{Fig:figPDE1} illustrates the evolution of the concentration $c(x,t)$ and adsorbed fraction $q(x,t)$ obtained by solving system \eqref{eq:Sist_ND} using a Scharfetter-Gummel Discretisation scheme. The data provided in \eqref{values} and \eqref{values2} and with $\Pe =0.1$ are used. The non-dimensional length of the column is 18.88. However, a shorter spatial span has been chosen to properly show the formation of travelling waves. This demonstrates that the actual column length is sufficient for the waves to fully develop. The solutions are plotted in dimensional form. The two figures in the first row depict the evolution of the profiles of $c(x,t)$ and $q(x,t)$ over time. It is evident that, after an initial transient, the solutions maintain their shape as they approach the end of the $x$-domain. To ascertain the validity of this observation, the curves in the second row are shifted to ensure that they intersect at the point $(x-x_{1/2},c)=(0,1/2)$, for the concentration, and $(x-x_{1/2},q)=(0,q_\text{e}/2)$, for the adsorbed fraction, where $q_\text{e}$ is given in \eqref{eq:Sips} and it corresponds to the equilibrium value of $q$ when $c=c_\text{in}$. To eliminate the initial transient, the curves before $0.4$ times the final time are excluded from this second row. We note that the adsorbed fraction, $q(x,t)$, seems to have a longer transient. Finally, on the third row, we present evidence demonstrating that the profiles of the second row evolve at a constant velocity. To achieve this, we track points $(x,t)$ with a fixed value for either the concentration or the adsorbed fraction, i.e., $c(x,t)=c^*$ (left figure) and $q(x,t)=q^*$ (right figure). Specifically, we select three distinct values, $c^*=\{1/4, 1/2, 3/4\}$, and $q^*=\{1/4q_e, 1/2q_e, 3/4q_e\}$. If the profiles evolve at a constant speed, they should correspond to functions such as $F(\eta)= c(x,t)$, $G(\eta)=q(x,t)$, where $\eta =x-vt$, where $v$ represents the (constant) speed of the front. Furthermore, we observe that both $F(\eta)$ and $G(\eta)$ are monotonic functions. Therefore, given a value $c^*\in(0,1)$ and $q^*\in(0,q_e)$, it is possible to invert the functions to obtain $\eta = F^{-1}(c^*)$ and $\eta = G^{-1}(q^*)$, which is equivalent to $x-vt =F^{-1}(c^*)$ and $x-vt = G^{-1}(q^*)$. Consequently, the level sets $c(x,t)=c^*$ and $q(x,t)=q^*$ are straight lines whose slope corresponds to the front velocity, as the figures on the third row show.

\begin{figure}
	\includegraphics[width=\textwidth]{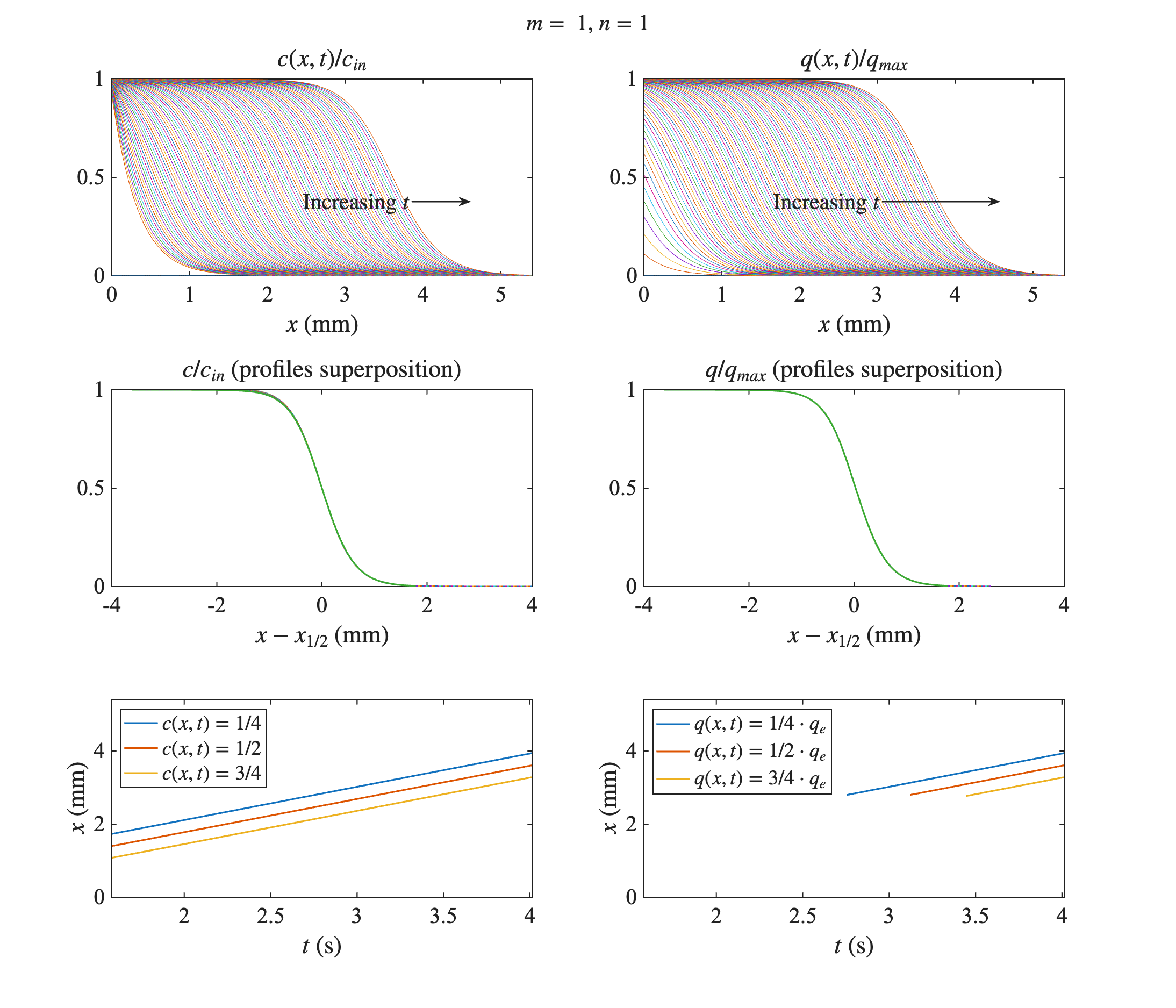}
		\caption{Solution, in dimensional form, $c(x,t)$ and $q(x,t)$ of system \eqref{eq:Sist_ND} using the data provided in \eqref{values} and \eqref{values2} with reaction exponents $(m,n)  = (1,1)$, and $\Pe = 0.1$.  First row: evolution in time of the profiles for $c(x,t)$ (left) and $q(x,t)$ (right) as a function of $x$. Second row: shifted curves to force the intersections at $(x-x_{1/2},c)=(0,1/2)$ (left), and $(x-x_{1/2},q)=(0, q_\text{e}/2)$ (right), where $q_\text{e}$ is given in \eqref{eq:Sips}. Initial transient curves have been excluded from the second row. Third row: level sets for $c(x,t)=c^*$ and $q(x,t)=q^*$, where $c^*=\{1/4, 1/2,3/4\}$, $q^*=\{q_e/4, q_e/2,3q_e/2\}$.}
	\label{Fig:figPDE1}
\end{figure}

It should also be noted that Figure \ref{Fig:figPDE1} shows an example based on values of a real experimental setting. However to explore the presence of travelling waves in a wide range of situations, in Section~\ref{sec:pde}, we present the solution to the partial differential equation \eqref{eq:Sist_ND} for various combinations of the reaction exponents, $(m,n)$ and for different values of the non-dimensional parameters $\Da$, $\Pe$ and $\alpha$. The results demonstrate that the travelling-wave approximation appears to be applicable to a broad spectrum of adsorption reactions.

In \cite{AGUARELES2023}, the authors obtain a first-order system of ordinary differential equations for the profile's functions, $F(\eta)$ and $G(\eta)$. By neglecting the inverse P\'eclet number, which is usually small, the authors in \cite{AGUARELES2023} derive a reduced first-order equation from which explicit solutions for certain values of the global reaction orders, $m$ and $n$, can be obtained. These solutions are extremely useful to experimentalists as they provide insight into the roles and dependencies of the parameters, and they are the main tool to infer unknown experimental parameters. However, these solutions are derived as solutions of an equation that has been obtained by performing two approximations: first, approximating the solutions of the original system of partial differential equations by a system of ordinary differential equations for the travelling-wave profile, and second, by further reducing this system to an even simpler one that allows for analytic solutions in an implicit form. The key contributions of the current work are twofold: it provides conditions for the existence and uniqueness of travelling-wave solutions for the full system of ordinary differential equations, and it explores the range of validity of the analytic solutions by comparing them with the numerical solutions of the full system of ordinary differential equations.

In the following section, we will perform an analytical analysis to establish the conditions that the parameters $(m,n)$ must satisfy to for the travelling-wave approximation to yield a valid solution.

\section{Travelling-wave solutions}
\label{sec:trav}

\label{sec:TW}

As explained in \cite{AGUARELES2023}, to construct a travelling-wave approximation, one must first extend the physical finite domain to an infinite one. In what follows, we shall consider that upstream and far enough from the wave front, the concentration is close to that at the inlet, i.e. $c\sim 1$ while downstream the filter is fresh, so $q \sim 0$.

To obtain the travelling-wave equations, one starts by assuming that $c$ and $q$ can be written as functions that only depend on the similarity variable, $\eta$,
\begin{equation}
    \eta = x - x_f - v(t-t_f)\,,\quad c(x,t)=F(\eta)\,,\quad q(x,t)=G(\eta)\, ,
\end{equation}
where $v$ is the (constant) wave velocity, and $x_f\in\mathbb{R}$, $ t_f>0$ are arbitrary values. Since travelling-wave solutions are defined in unbounded domains, the similarity variable is now defined in the whole real line, $\eta\in\mathbb{R}$. Then, equations (\ref{eq:c_dim},\ref{eq:q_dim}) become
\begin{subequations}
\label{eq:TW}
\begin{alignat}{2}
    (1-v\Da) F'&=\Pe F''+vG', 
    \label{eq:F}\\
 -vG'&=\alpha F^m (1-G)^n-(1-\alpha)G^n, 
 \label{eq:G}
\end{alignat}
where $'=\dv*{}{\eta}$.
Moreover, in the limits $\eta\to\pm\infty$, the solutions are expected to attain constant values
\begin{equation}
\label{eq:TW_BC}
\lim_{\eta\to-\infty} F(\eta)=F_0, \quad \lim_{\eta\to-\infty} G(\eta) = G_0,  \quad\lim_{\eta\to +\infty} F(\eta)=F_\infty, \quad\lim_{\eta\to +\infty} G(\eta)=G_\infty\, ,    
\end{equation}    
\end{subequations}
and 
\begin{equation*}
\lim_{\eta\to\pm\infty} F'(\eta)=0.    
\end{equation*}
Equation \eqref{eq:G} at plus and minus infinity provides the relations
\begin{equation}
\label{eq:F0Finf}
	 \frac{\a}{1-\a}F_0^m=\left(\frac{G_0}{1-G_0}\right)^n\, ,\quad\frac{\a}{1-\a}F_\infty^m=\left(\frac{G_\infty}{1-G_\infty}\right)^n,
\end{equation}
while integrating \eqref{eq:F} and evaluating at plus and minus infinity determines the velocity:
\begin{equation}
\label{eq:vF0Finf}
    v = \frac{F_0 - F_\infty}{G_0 - G_\infty + \Da (F_0 - F_\infty)}\, .
\end{equation}

We note that the second expression in \eqref{eq:F0Finf} states a relation between the downstream concentration and the adsorbed fraction. In particular, we note that, if one assumes that the fluid inside the filter is initially clean, then the travelling-wave assumption only works if the adsorbent is fresh. These are, in fact, the initial conditions that we were considering in \eqref{eq:ic}. So, in what follows,
$$F_\infty=G_\infty=0.$$

Upstream, the fluid concentration corresponds to the inlet value, given as $F_0=1$
in non-dimensional form. At this upstream limit, the filter reaches its adsorption capacity, so $G_0=q_e$, also in non-dimensional form. Consequently, the first expression in \eqref{eq:F0Finf} represents the isotherm in its non-dimensional form, which provides a relation between the maximum adsorption capacity, $q_e$, and the parameter $\alpha$:
\begin{equation}
\label{eq:G_inf}
    \frac{\a}{1-\a}=\left(\frac{q_e}{1-q_e}\right)^n.
\end{equation}
Introducing all these conditions in \eqref{eq:vF0Finf}, the velocity is found to be
\begin{equation}
\label{eq:v}
    v = \frac{1}{q_e+ \Da}\, .
\end{equation}
In Section~\ref{sec:pde} we compare expression \eqref{eq:v} with the velocity obtained from the numerical solutions of the system of partial differential equations \eqref{eq:Sist_ND}.

Notice that equations \eqref{eq:F}, \eqref{eq:G} are decoupled. Integrating equation \eqref{eq:F} and applying the boundary conditions \eqref{eq:TW_BC} with $F_0=1, G_0=q_e$ and $F_\infty=G_\infty=0$, the function $G$ can be written as a function of $F$ and $F'$ as
\begin{equation}
	\label{eq:GF}
	G=\frac{1}{v}\left((1-v\Da)F -\Pe F'\right)=q_e F-\Pe(q_e+\Da)F'.
\end{equation}
Then, combining expressions \eqref{eq:GF} and \eqref{eq:F} and using \eqref{eq:v}, one obtains a second-order ordinary differential equation for the contaminant front concentration, $F(\eta)$
\begin{subequations}
\label{eq:fullODE}
	\begin{equation}
	\label{eq:ODE2n}
	\begin{split}
		\Pe F''=\frac{q_e}{q_e+\Da}  F'&+\alpha F^m \left(1-q_e F +\Pe( q_e+\Da) F'\right)^n\\&-(1-\alpha)\left(q_e F -\Pe( q_e+\Da)F'\right)^n.
	\end{split}
\end{equation}
We note that equation \eqref{eq:ODE2n} has two equilibria: $F=1$ and $F=0$. In what follows, we will show that there exist heteroclinic solutions of equation \eqref{eq:ODE2n} connecting these two equilibria. From a physical perspective, this heteroclinic connection represents the transition between a fully saturated upstream state and a fresh downstream state, corresponding to a sharp adsorption front propagating along the column. In particular, these solutions satisfy 
\begin{equation}
\label{eq:BC_F}
	\lim_{\eta\to-\infty} F(\eta) = 1,\quad \lim_{\eta\to\infty} F(\eta) = \lim_{\eta\to\pm\infty}F'(\eta) =0.
\end{equation} 
\end{subequations}
The P\'eclet number is usually very large and so $\Pe\ll 1$ (see for instance \cite{SULAYMON2014325} or \cite{Valverde24} and the references therein). Therefore,  we first show that in the limit $\Pe=0$ there exist solutions of equation \eqref{eq:ODE2n} satisfying the conditions provided in \eqref{eq:BC_F} and, afterwards, we will provide the proof that these heteroclinic connections persist for positive values of $\Pe$. 
\subsection{Leading-order approximation}
By neglecting $\Pe$ in \eqref{eq:ODE2n} we obtain the first-order ordinary differential equation
\begin{subequations}
\label{eq:leadProb}
\begin{equation}
	\label{eq:lead}
	\frac{q_e^{1-n}}{q_e+\Da}  F'=(1-\alpha)F^n-\alpha F^m \left(\frac{1}{q_e}-F\right)^n, 
\end{equation}
where $\a$ and $q_e$ satisfy \eqref{eq:G_inf}, and we are looking for solutions satisfying
\begin{equation}
\label{eq:lead_BC_ad}
    \lim_{\eta\to-\infty} F(\eta) = 1\, ,\quad \lim_{\eta\to +\infty} F(\eta) = \lim_{\eta\to\pm\infty}F'(\eta) = 0.
\end{equation}	
\end{subequations}
The following proposition provides the conditions that the parameters and initial condition must satisfy for \eqref{eq:lead} to have a solution satisfying \eqref{eq:lead_BC_ad}.

\begin{prop}\label{prop1}
Consider fixed values $n,m\in\mathbb{N}$ and $\a,q_e\in(0,1)$ satisfying \eqref{eq:G_inf}, 
and let $F(\eta)$ be the solution of the initial value problem given by \eqref{eq:lead}
with the initial condition $F(0) = c_0\in(0,1)$.
Then $F(\eta)$ is strictly decreasing and satisfies \eqref{eq:lead_BC_ad} if and only if $m\leq n$.
\end{prop}

\begin{proof}
First, we will see that $F=0$ and $F=1$ are equilibrium points of the dynamical system given by \eqref{eq:lead}. 
Next, for $m\leq n$, we will see that there are no other equilibrium points for $F\in(0,1)$, and any solution with initial condition $c_0\in(0,1)$ is a decreasing function. Therefore, in this case, any solution is, in fact, a heteroclinic connection between the two points and so the conditions \eqref{eq:lead_BC_ad} are satisfied. 
For $m>n$, we will see that \eqref{eq:lead_BC_ad} cannot be satisfied.

The equilibrium points of the system are given by the zeros of the polynomial
\begin{eqnarray}
    p(u) &=& (1-\alpha)u^n-\alpha u^m \left(a-u\right)^n= \a \left((a-1)^n u^n - u^m(a-u)^n\right),\nonumber
\\
    &=& \a (a-1)^n u^m\left( u^{n-m}-\left(\frac{a-u}{a-1}\right)^n\right)
    \label{eq:p}
\end{eqnarray}
where $a=1/q_e$ and we have used that \eqref{eq:G_inf} is equivalent to
$1-\a=\a(a-1)^n.$  
Clearly, $u=0$ and $u=1$ are zeros of $p(u)$. 

We now prove that, if $m\leq n$, $p(u)$ does not have any other zero for $u\in(0,1)$. By \eqref{eq:p}, 
the non-zero roots of $p(u)$ are also solutions of 
\begin{equation*}
	g(u):=u^{(n-m)/n}=\dfrac{a-u}{a-1}:=h(u). 
\end{equation*}
Clearly, $u=1$ is a solution, and it is the unique solution when $n=m$.  If $m<n$, for $u\in[0,1]$, $g(u)$ is a strictly increasing function, whereas $h(u)$ is strictly decreasing and $g(1)=h(1)=1$. Therefore, $g(u)$ and $h(u)$ cannot attain the same value at any other point in the interval $u\in(0,1)$. Also, if $m=n$, $g(u):=1$ and $h(u)>1$ for all $u\in[0,1)$, and again $u=1$ is the only solution in the interval $u\in[0,1]$.

Moreover, for $m \leq  n$, $g(u)<h(u)$ in $u\in(0,1)$, so $p(u)<0$ for $u\in(0,1)$.  Therefore, any solution $F$ of the initial value problem defined by \eqref{eq:lead} with initial condition $c_0\in(0,1)$ is a decreasing solution that connects $F=1$ with $F=0$ and satisfies the boundary conditions \eqref{eq:lead_BC_ad}.

Finally, if $m>n$, there always exists $0<\e<1$ small enough such that $p(\e)>0$. Therefore, the boundary conditions~\eqref{eq:lead_BC_ad} cannot be satisfied. 
\end{proof}

This proposition proves that there is one heteroclinic connection between the two equilibrium points $F=0$ and $F=1$ when $m\leq n$.
 We note that it is not difficult to see that in the case $m>n$, there exists one equilibrium point $F=c^*\in(0,1)$ provided $a< m/(m-n)$ so, in this case, any solution of~\eqref{eq:lead} with initial condition $c_0\in (c^*,1)$ will be a heteroclinic connection between $c^*$ and $1$, so that $\lim_{\eta\to +\infty} F(\eta)=c^*\neq 0$, which contradicts the initial assumption of $F_{\infty}=0$. Furthermore, when $a\geq  m/(m-n)$, any solution in the interval $(0,1)$ is an increasing function, which has no physical sense in this context. 

Proposition \ref{prop1} provides the scenarios where travelling-wave solutions can be used to describe the evolution of the contaminant in the column. In fact, using \eqref{eq:GF}, we observe that, to leading order (that is, neglecting $\Pe$),
$$G(\eta) = q_e F(\eta),$$
and so $G(\eta)$ has the same profile than $F(\eta)$. Moreover, this proposition offers insight into why travelling waves do not exist when $m>n$. In these scenarios, the contaminant concentration in the fluid does not evolve as a well-defined front. Instead, the fluid flows through the porous media, leaving many adsorption sites unoccupied leading the process to a premature breakthrough. This means that, by the time the pollutant begins to escape from the filter, the adsorbent remains largely unsaturated. In principle, this appears to be a bad scenario, as the filter's lifespan is not at its maximum. However, many other considerations must be taken into account, like manufacturing and regeneration costs.

\subsection{Travelling-wave solutions for the full system}
\label{sec:TWComp}
We now consider the full system provided by equation \eqref{eq:ODE2n}, which, for clarity, we rewrite here:
\begin{equation*}
	\begin{split}
		\Pe F''=\frac{q_e}{q_e+\Da}  F'&+\alpha F^m \left(1-q_e F +\Pe( q_e+\Da) F'\right)^n\\&-(1-\alpha)\left(q_e F -\Pe( q_e+\Da)F'\right)^n,
	\end{split}
\end{equation*}
and we aim to investigate the existence of heteroclinic connections   satisfying
\begin{equation*}
	\lim_{\eta\to-\infty} F(\eta) = 1,\quad \lim_{\eta\to\infty} F(\eta) = \lim_{\eta\to\pm\infty}F'(\eta) =0,
\end{equation*} 
when $\Pe>0$. As we have observed in the preceding section, the leading order satisfies these conditions only if $m \leq n$; therefore, we restrict our analysis to these cases.

In what follows, we prove the existence of travelling-wave solutions for the full system \eqref{eq:ODE2n} and for small (positive) values of $\Pe$. We consider fixed values of $m$ and $n$, such that $m\leq n$, for which we already know, due to Proposition \ref{prop1}, that there exists a travelling-wave solution for the zero-order approximation \eqref{eq:leadProb}.

First, naming $\epsilon=\Pe$, we note that \eqref{eq:ODE2n} can be written as
\begin{equation}
    \label{eq:slow-fast1}
    \left\{
    \begin{array}{rcl}
    y' &=& z,\\
    \epsilon z' &=& \dfrac{q_e}{q_e+\Da} \, z - P(y,\epsilon z),
    \end{array}
    \right.
\end{equation}
where $y=F$, $z=F'$, and 
$$
P(y,\epsilon z)= 
(1-\alpha)\left(q_e y -\epsilon(q_e+\Da) z\right)^n-
\alpha y^m \left(1-q_e y +\epsilon(q_e+\Da) z\right)^n.
$$
System \eqref{eq:slow-fast1}
is a slow-fast system, being $'=\text{d}/\text{d}\eta$ and $\eta$ the slow time. 

Slow-fast systems, a subclass of dynamical systems, are characterized by the presence of variables that evolve on distinctly different time scales: ``slow'' variables, which change gradually, and ``fast'' variables, which evolve more rapidly. The interplay between these time scales can lead to rich and often counter-intuitive dynamics, including canards, relaxation oscillations, and bifurcations.  The works of Dumortier provide foundational insights into the qualitative behaviour of these complex systems, in particular his publication in the \emph{Memoirs of the American Mathematical Society}, \cite{Dumortier1996}, has profoundly influenced the study and application of slow-fast dynamics in theoretical and applied contexts. There is an enormous
literature on this topic, see for example \cite{Dumortier2021} and the references therein. 

Considering $\epsilon=0$ in system \eqref{eq:slow-fast1}, we recover the leading-order approximation equation \eqref{eq:lead}. In particular, from \eqref{eq:slow-fast1}, we get the set
\begin{equation}
    \label{eq:slowset}
 z= \dfrac{q_e+\Da}{q_e}P(y,0)=q_e^{n-1}({q_e+\Da})p(y),
\end{equation}
where $p(y)$ is given in \eqref{eq:p} (see Proposition \ref{prop1}),
which is known as the critical ``slow set''.

The dynamics along the slow set \eqref{eq:slowset} are given by 
$$y' = \dfrac{{q_e+\Da}}{q_e}P(y,0).$$
Notice that the critical slow set \eqref{eq:slowset} implicitly gives the solution of \eqref{eq:lead} in the phase space satisfying \eqref{eq:lead_BC_ad}.
As stated in Proposition \ref{prop1}, for $m\leq n$
the critical slow set is a heteroclinic connection between the equilibrium points $(1,0)$ and $(0,0)$. In \cite{AGUARELES2023}, the authors provide solutions in implicit form of \eqref{eq:leadProb} for a set of values of the parameters $m$ and $n$.

Introducing the fast time $\xi=\eta/\epsilon$ in \eqref{eq:slow-fast1}, the system can be rewritten as 
\begin{equation}
    \label{eq:slow-fast2}
    \left\{
    \begin{array}{rcl}
    \dot{y} &=& \epsilon z,\\
    \dot{z} &=& \dfrac{q_e}{q_e+\Da}\,z  - P(y,\epsilon z),
    \end{array}
    \right.
\end{equation}
where $^. = \text{d}/\text{d}\xi$. Clearly, the phase space $(y,z)$ of both systems is the same; only the velocity along the solution curves changes. Moreover, the points $y=0,z=0$ and $y=1,z=0$ are also equilibrium points of the full system for any value of $\epsilon$.

Again, considering the limit system for $\epsilon=0$, we obtain the layer equation
\begin{equation}
    \label{eq:slow-fast3}
    \left\{
    \begin{array}{rcl}
    \dot{y} &=& 0,\\
    \dot{z} &=& \dfrac{q_e}{q_e+\Da}\, z  - P(y,0),
    \end{array}
    \right.
\end{equation}
where the ``layers'' $y=y_0$ are invariant sets of the system. In fact, system \eqref{eq:slow-fast3} is a uniparametric family of one-dimensional subsystems ($y_0$ being the parameter) representing the fast motion. Again, the slow motion takes place along the curve \eqref{eq:slowset}. Moreover, any point on the slow set is a normally hyperbolic point (that is, the linear part of the differential field when $\epsilon=0$ has at least one non-zero eigenvalue). 
In Figure~\ref{fig:slow-fast-layers}, the dynamics of system \eqref{eq:slow-fast3} are shown. Clearly, along the layers where $y=y_0$, on the points over the slow-set, $\dot{z}>0$, whereas in those below, $\dot{z}<0$.

\begin{figure}[!ht]
\centering	
\includegraphics[width=0.6\linewidth]{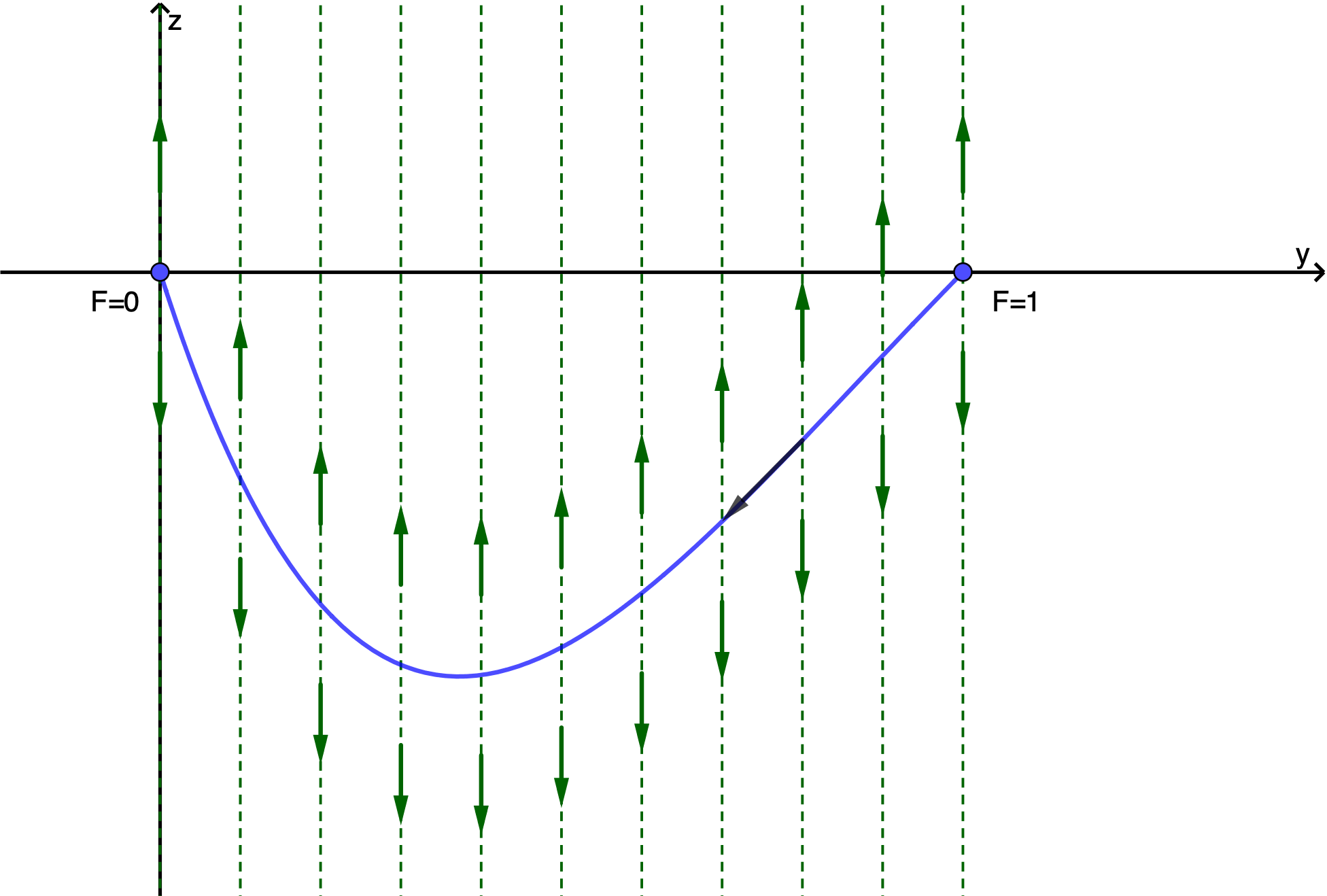}
\caption{Slow-fast dynamics in the configuration space system \eqref{eq:slow-fast2} for $\epsilon=0$. The fast dynamics take place along the ``layers'' $y=y_0$. The blue curve corresponds to the slow-set \eqref{eq:slowset}.}
   \label{fig:slow-fast-layers}
\end{figure}

Viewing the problem as a slow-fast system allows us to ensure the persistence of the connection for $\epsilon\neq 0$, for example, applying Theorem 9.10 in \cite{Dumortier2021}. That is, there exists a heteroclinic connection between the two equilibrium points for the full system \eqref{eq:slow-fast2}, or equivalently, there exists a solution for the boundary value problem \eqref{eq:fullODE}, which corresponds to the travelling wave.  The proof is based on an analytical continuation with respect to the small parameter, so that, provided that $\Pe$ is small enough, the solution that connects the two equilibrium points is close to the solution for the leading-order equation \eqref{eq:lead}.


\section{Sensitivity analysis}
\label{sec:sens}
The main goal of modelling pollutant capture through adsorption is to replicate the operational dynamics of a column during its experimental phase, and to assess the scalability to actual industrial settings. To achieve this, a simplified model is required to derive explicit solutions that facilitate the estimation of experimental parameters, such as $\a$. In the previous sections, we have proved that the heteroclinic connections representing travelling-wave solutions persist for non-zero values of the inverse P\'eclet number. Although the inverse P\'eclet number is usually very small (see, for instance, \cite{myers23} or \cite{CABRERACODONY2018565}), it is necessary to perform an analysis to evaluate the accuracy of the leading-order approximation \eqref{eq:leadProb}, identifying the range of values for $\Pe$ where the deviation between the full and the reduced models remains acceptable. 

We solve the system of equations \eqref{eq:TW} using the MATLAB built-in function \texttt{ode15s}, which is particularly useful to solve problems with a singular mass matrix. We integrate backwards from initial conditions close to $F=0$, identifying heteroclinic connections between $F=0$ and $F=1$.  We then shift the curves so that $F(0)=1/2$ for all of them. The results are presented in Figures~\ref{Fig:solMaria_perNM_EspaiFases_FrontOna_A} and~\ref{Fig:solMaria_perNM_EspaiFases_FrontOna_B}. A qualitative comparison of the curves in Figures~\ref{Fig:solMaria_perNM_EspaiFases_FrontOna_A} and~\ref{Fig:solMaria_perNM_EspaiFases_FrontOna_B} shows the robustness of the leading-order approximation, which even for values of $\Pe$ of order one remains close to the solution of the full problem \eqref{eq:TW}. In this section, we go further and measure the difference between the solutions of \eqref{eq:ODE2n} and the leading-order approximation ($\Pe=0$) for values of the inverse P\'eclet number up to 1.5, well inside the order-one regime.   

\begin{figure}
    \centering
\includegraphics[width=0.45\linewidth]{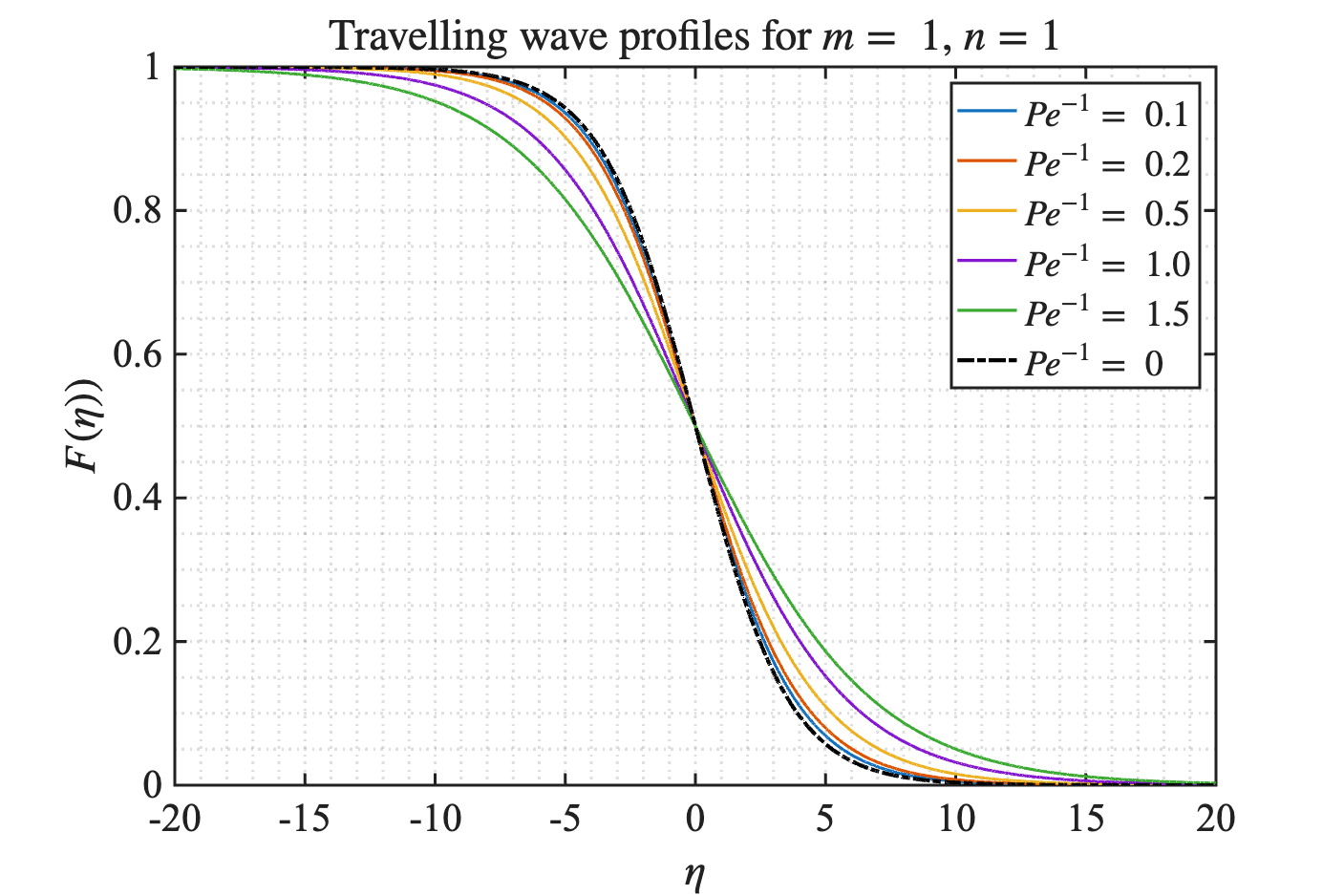}
\includegraphics[width=0.45\linewidth]{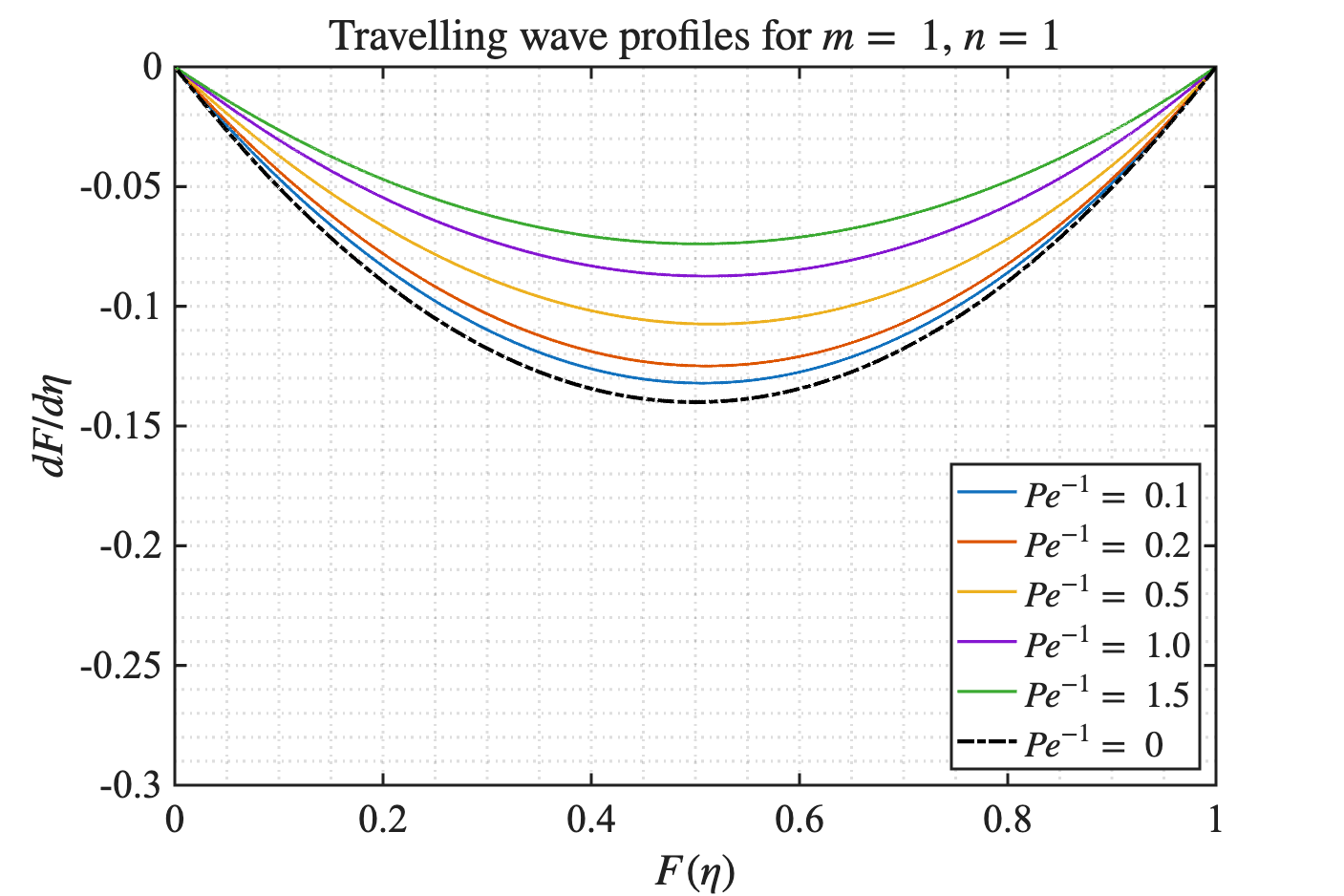}
\includegraphics[width=0.45\linewidth]{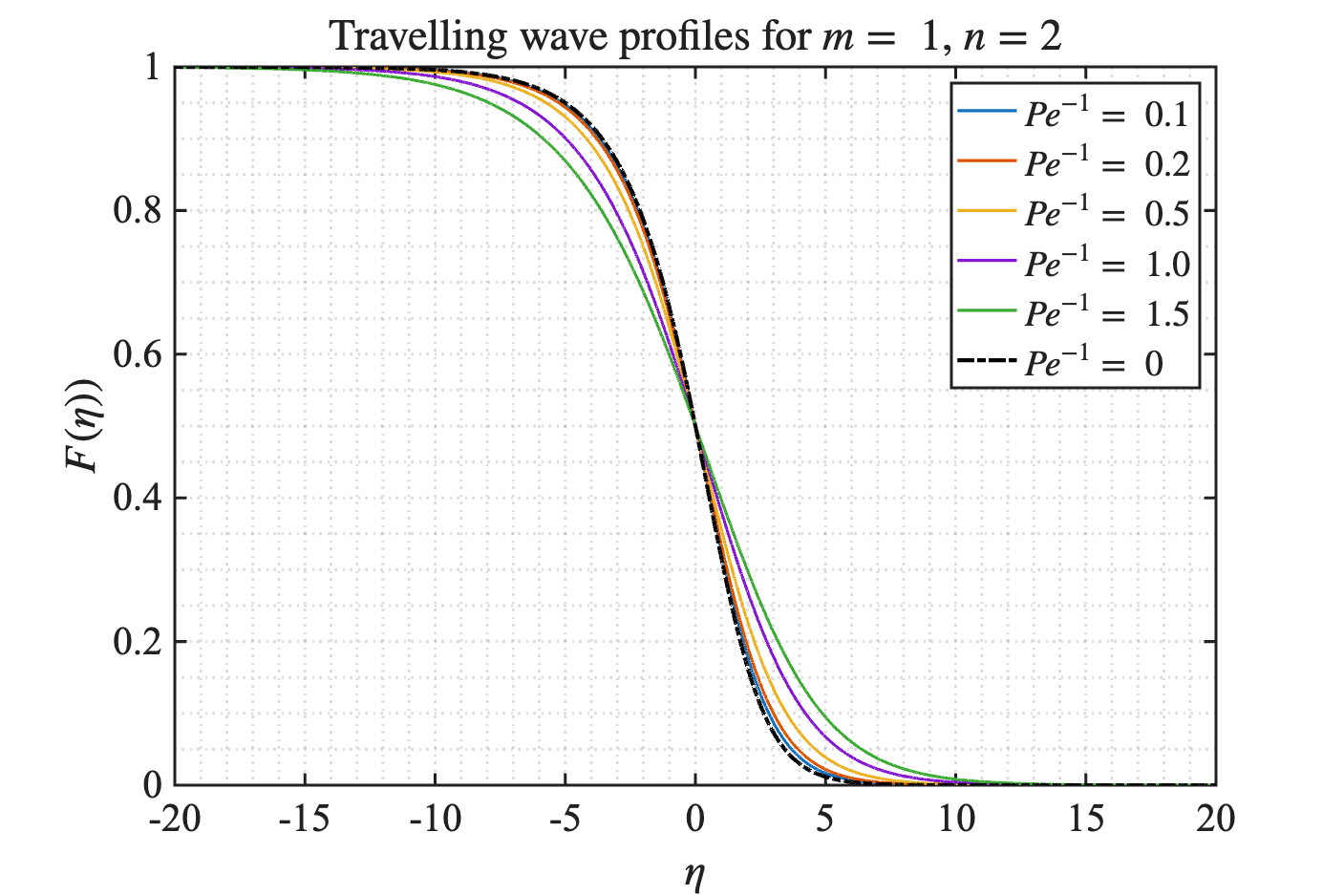}
\includegraphics[width=0.45\linewidth]{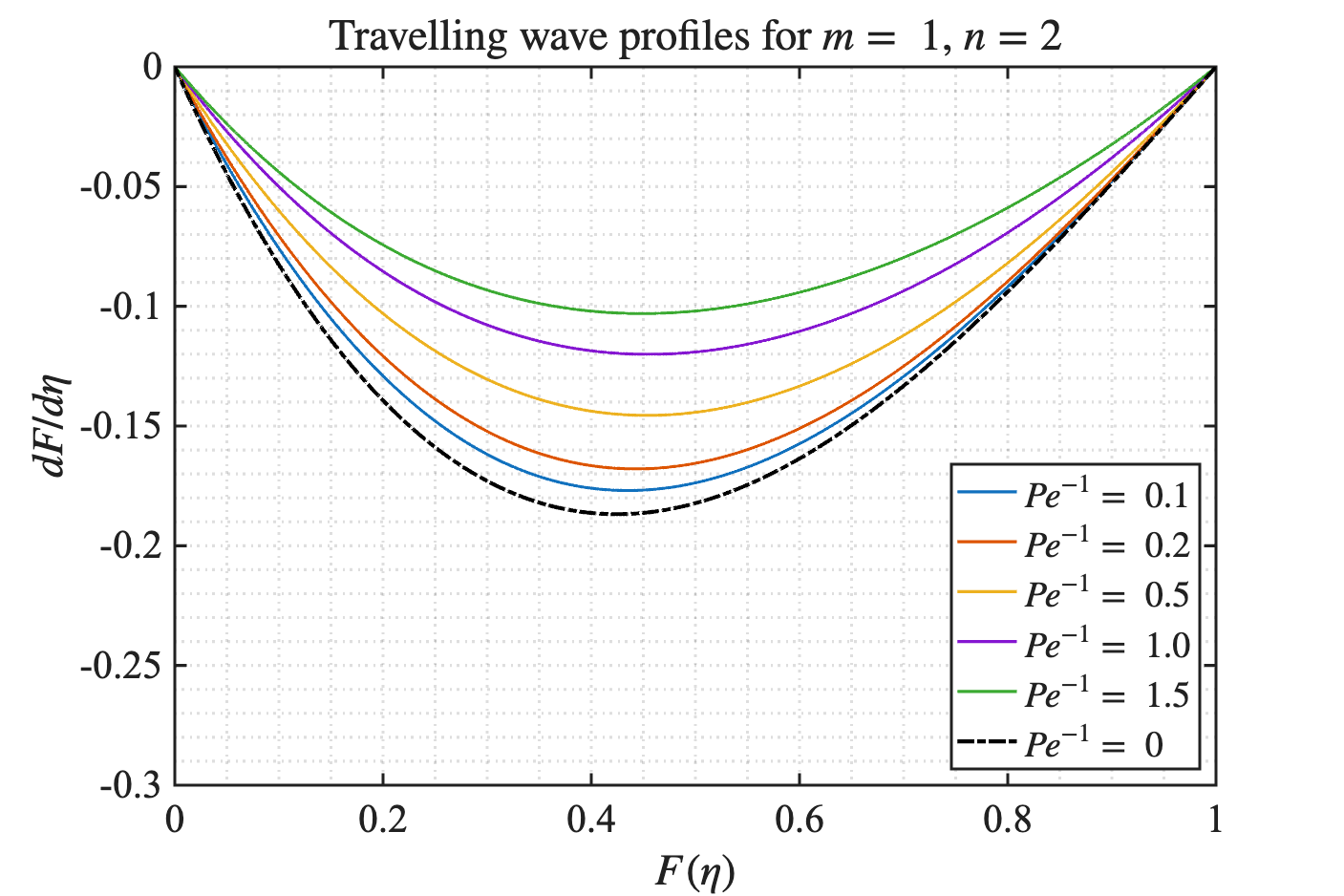}
\includegraphics[width=0.45\linewidth]{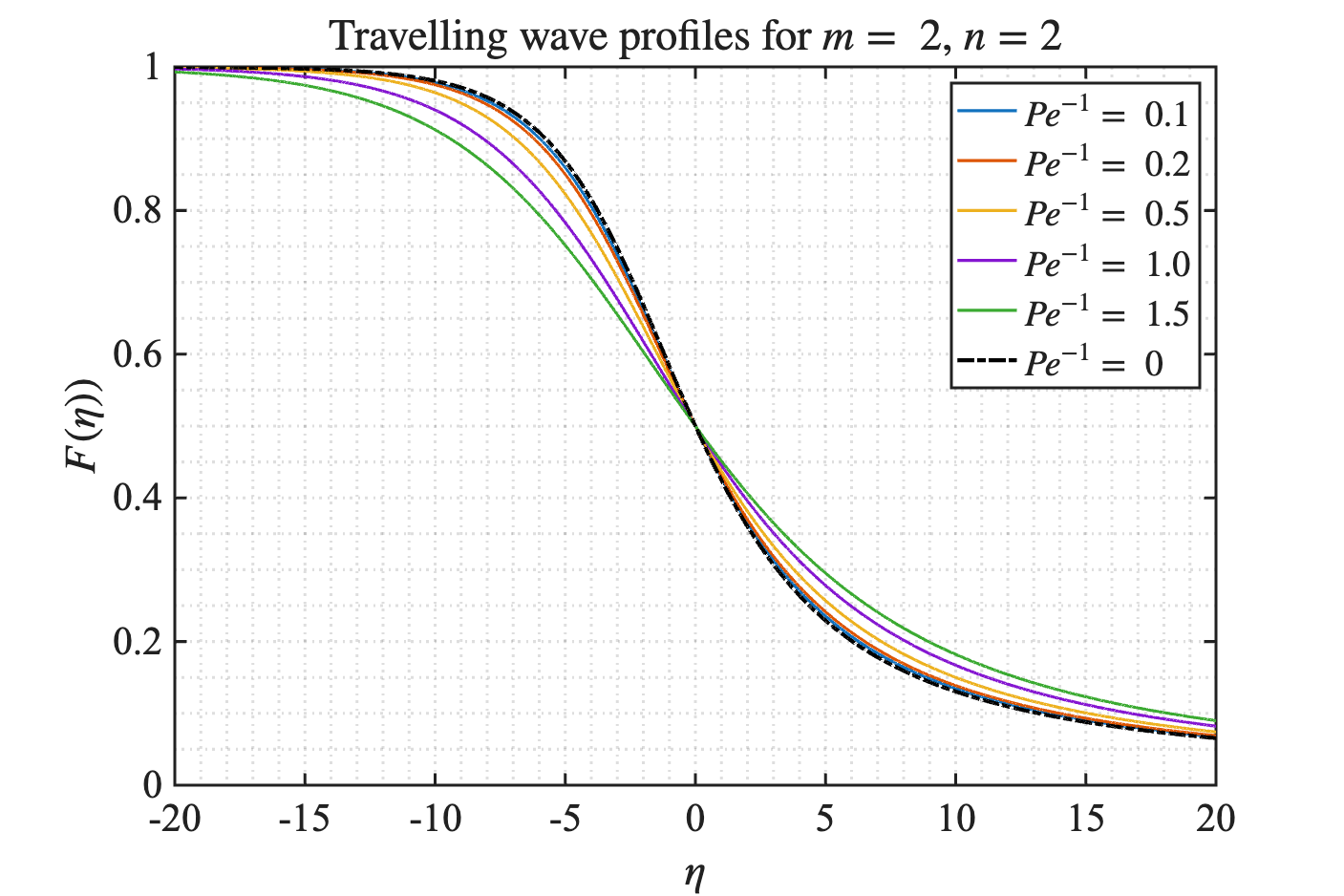}
\includegraphics[width=0.45\linewidth]{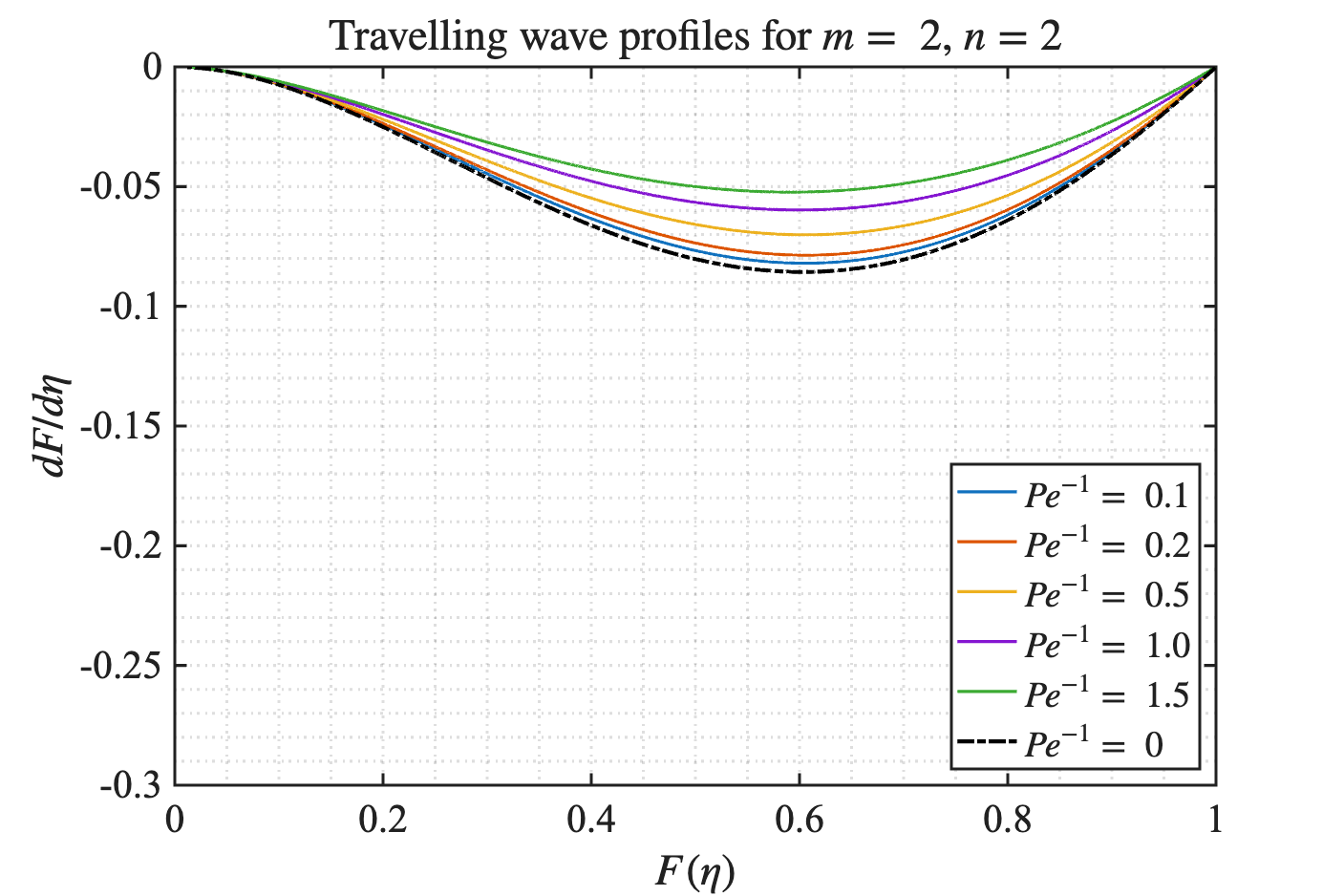}
    \caption{Travelling-wave solutions as a function of $\eta$ (left column), and heteroclinic connections of \eqref{eq:ODE2n} in the phase plane (right column) for $(m,n)=(1,1),(1,2)$ and $(2,2)$ for different values of the inverse P\'eclet number. The simulation values are $q_e=0.7$, $\Da=0.1$.}
 \label{Fig:solMaria_perNM_EspaiFases_FrontOna_A}
\end{figure}

\begin{figure}
\centering
\includegraphics[width=0.45\linewidth]{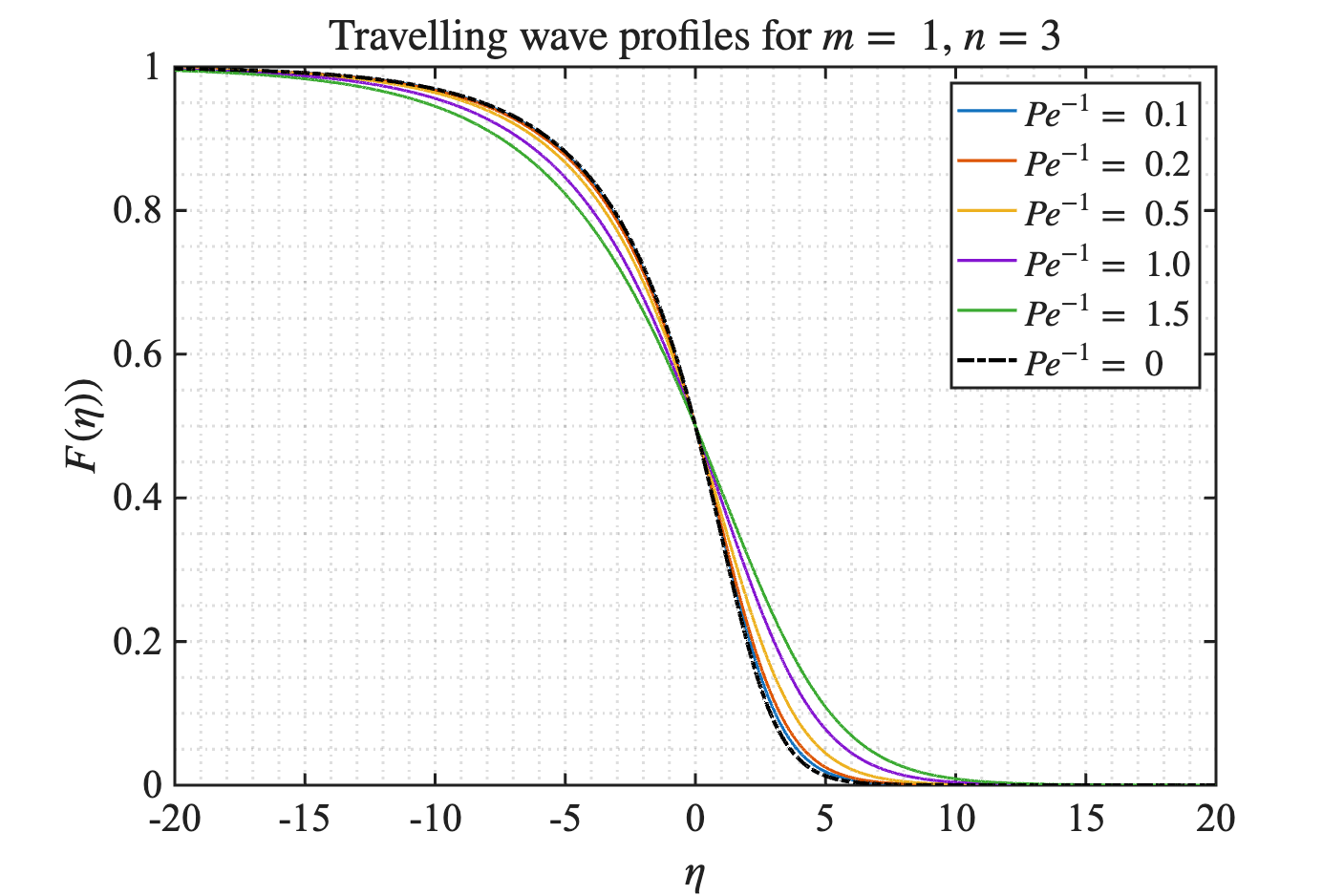}
\includegraphics[width=0.45\linewidth]{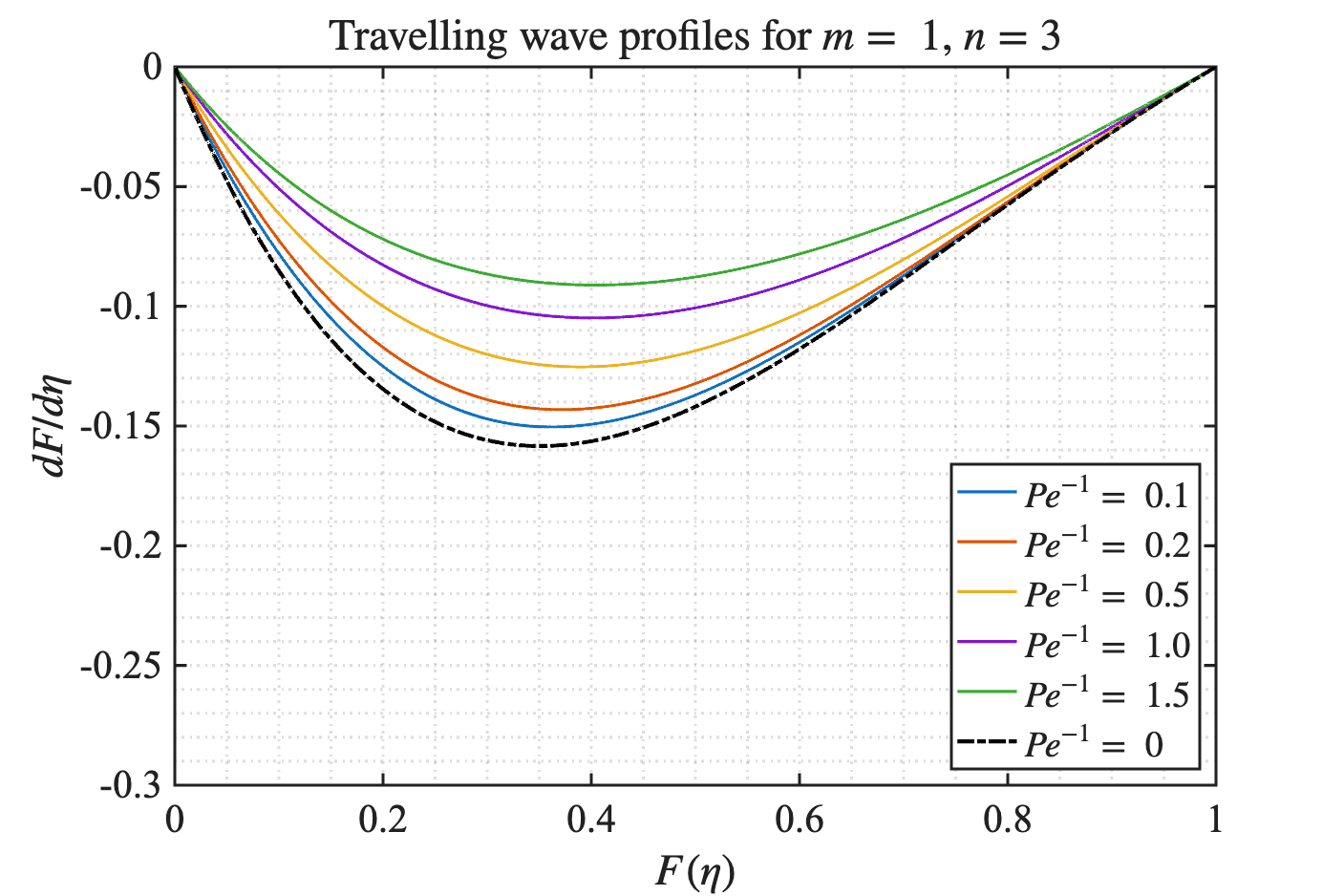}
\includegraphics[width=0.45\linewidth]{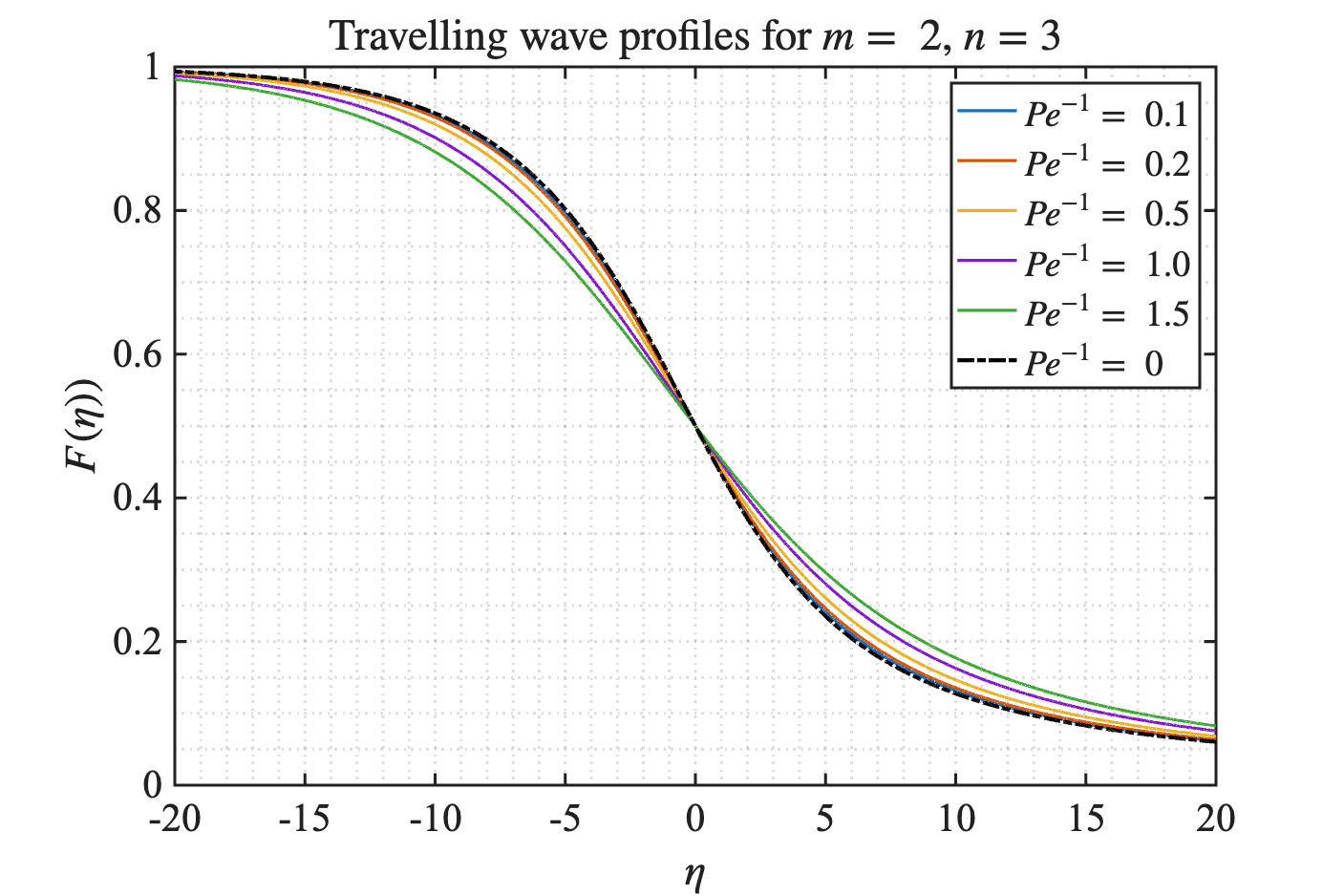}
\includegraphics[width=0.45\linewidth]{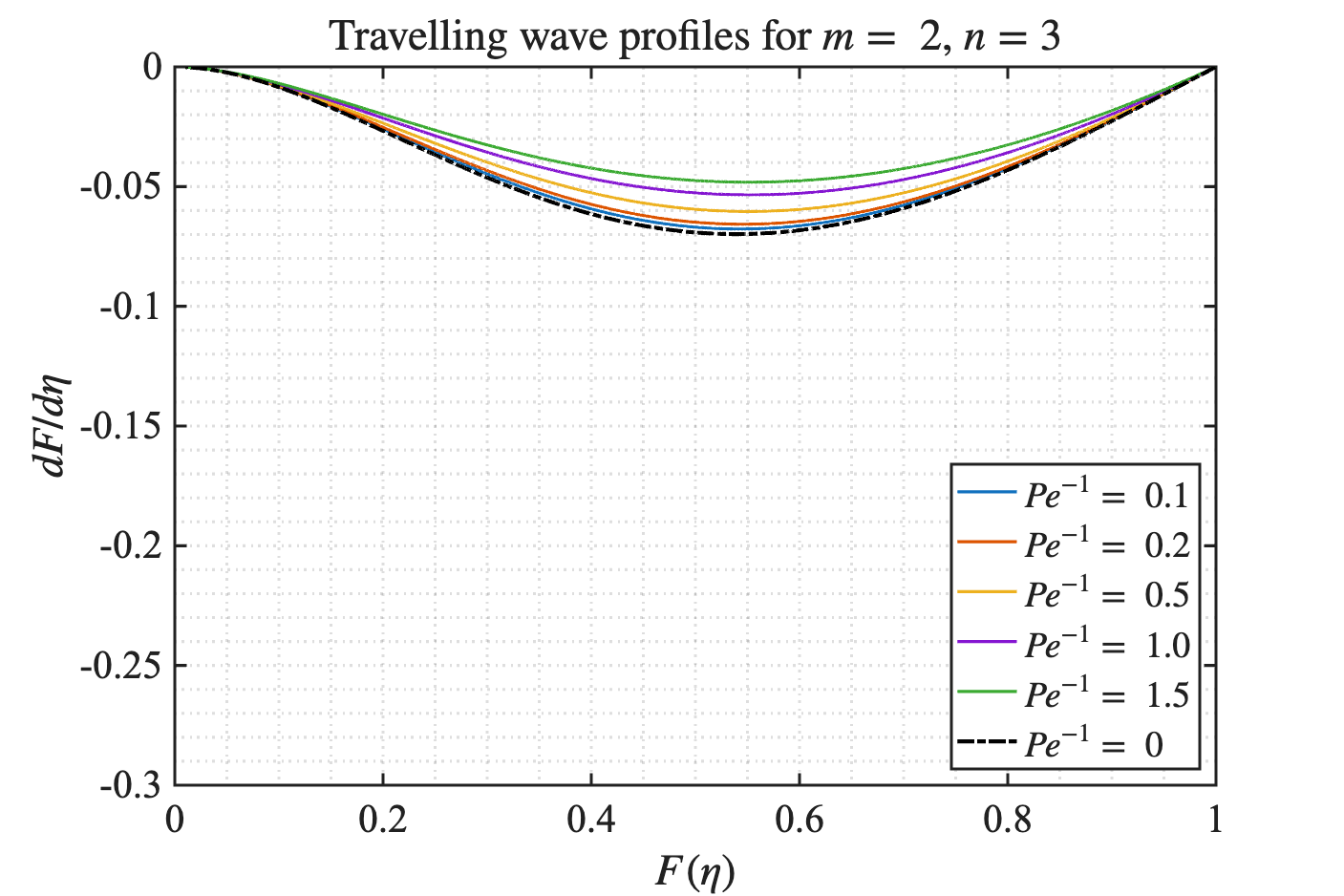}
\includegraphics[width=0.45\linewidth]{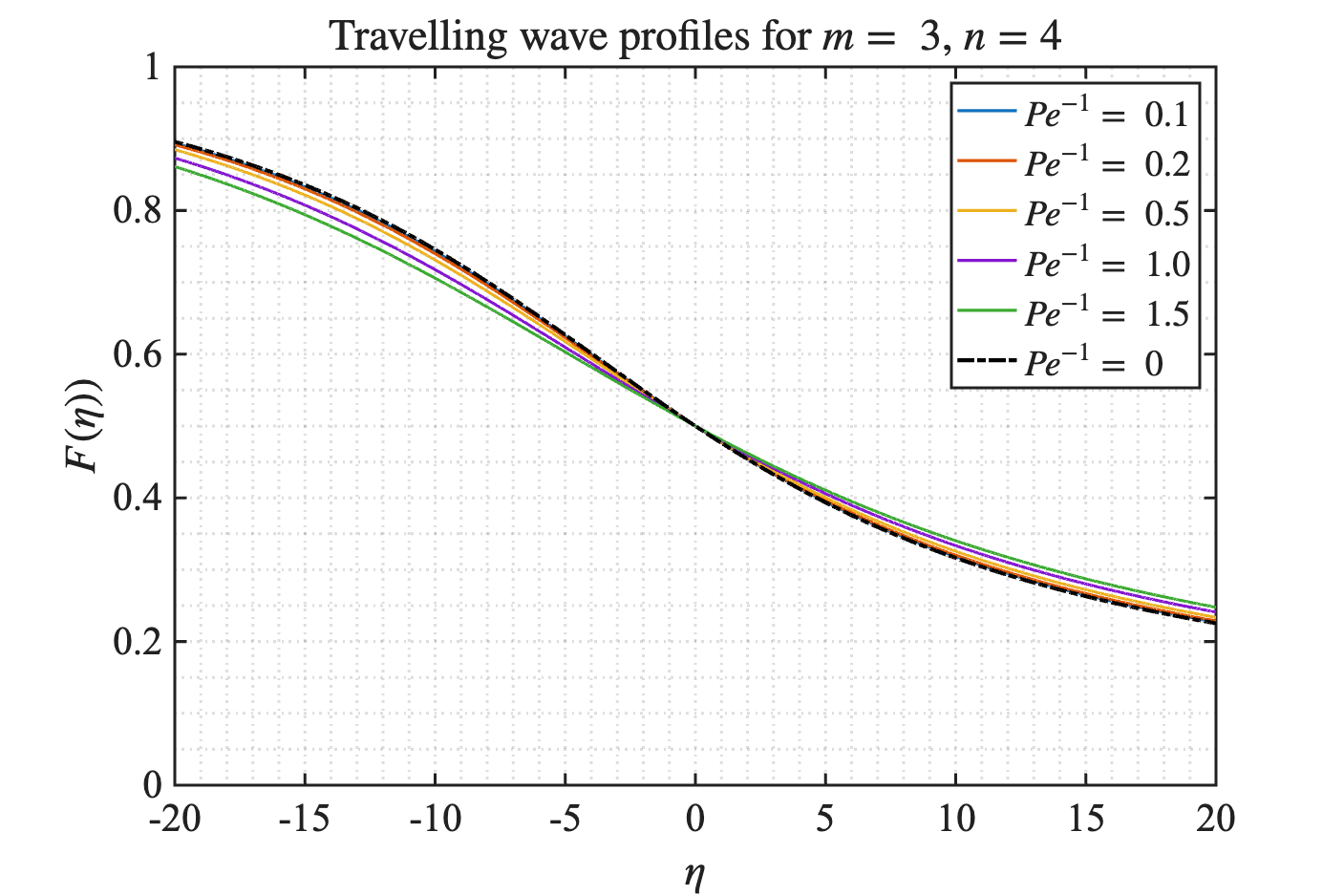}
\includegraphics[width=0.45\linewidth]{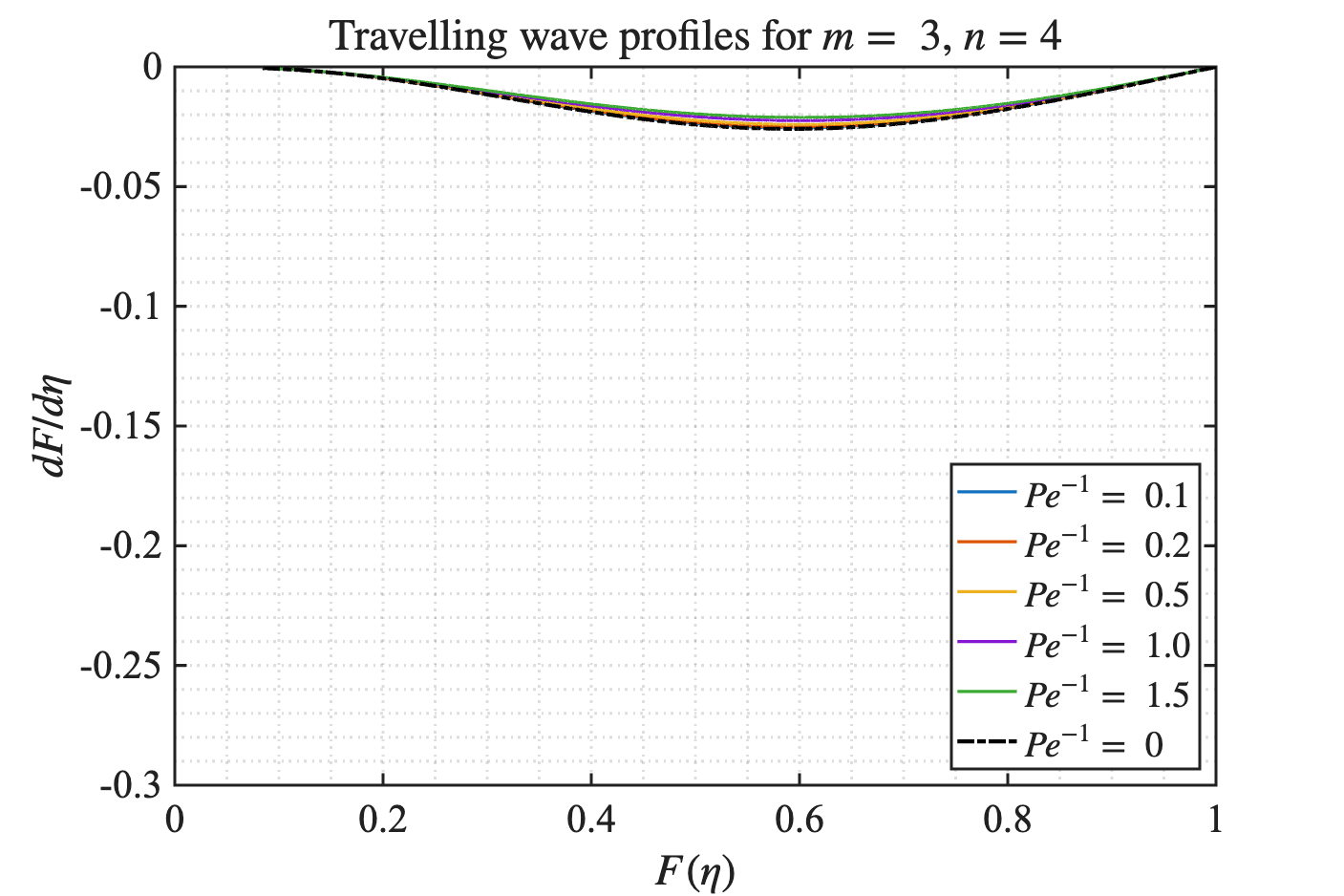}
    \caption{Travelling-wave solutions as a function of $\eta$ (left column), and heteroclinic connections of \eqref{eq:ODE2n} in the phase plane (right column) for $(m,n)=(1,3),(2,3)$ and $(3,4)$ for different values of the inverse P\'eclet number. The simulation values are $q_e=0.7$, $\Da=0.1$.}
\label{Fig:solMaria_perNM_EspaiFases_FrontOna_B}
\end{figure}

In Figure~\ref{Fig:L2norm} we compare the travelling-wave curves through the $L^2$-norm in a bounded domain, $[-\eta^*,\eta^*]$, of the difference between the leading-order solution, $F_0(\eta)$, and the solution of the full system \eqref{eq:fullODE}, $F(\eta;\Pe)$, for values of $\Pe$ ranging from 0 to 1.5,
\begin{equation}
\label{eq:ePe}
e(\Pe) = \left(\int_{-\eta^*}^{\eta^*} (F(\eta;\Pe)-F_0(\eta))^2\,\text{d}\eta\right)^{1/2}\,,	
\end{equation}
where we take $\eta^*=20$ to capture the transition zone from 1 to 0. In Figure~\ref{Fig:L2norm} we show the values of $e(\Pe)$ for different combinations of the order parameters, $(m,n)$, and for different values of the Damk\"ohler number, $\Da$. We use a log-log scale to reveal the linear dependence of the error on $\Pe$. To do this, we perform a linear fit, but only including values of $\Pe\in[0,0.4]$. The slopes of the fit are close to one, as predicted in \cite{AGUARELES2023}. In fact, if one takes larger values of $\Pe$, the linear fit provides slopes that start being much slower than one, indicating that the errors stop being linear in $\Pe$. However, when the Damk\"ohler number becomes of order one and greater, the fit seems to be worse. To achieve the same fits as with lower values of the Damk\"ohler number, the upper bound in the $\Pe$ range should be lower. We also note that the errors are larger in the physisorption case, i.e., for $m=n=1$.

\begin{figure}
\centering
\includegraphics[scale=0.35]{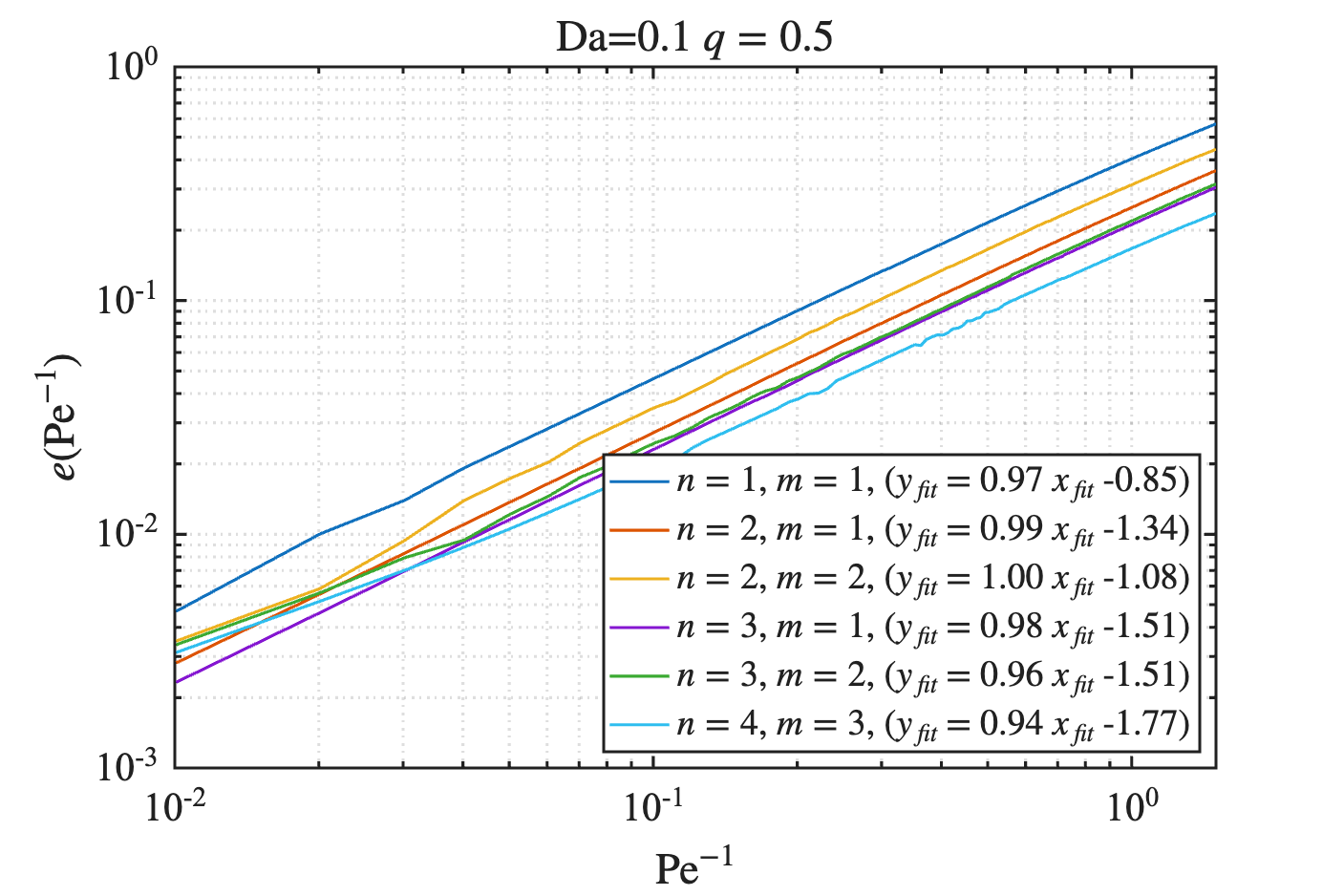}
\includegraphics[scale=0.35]{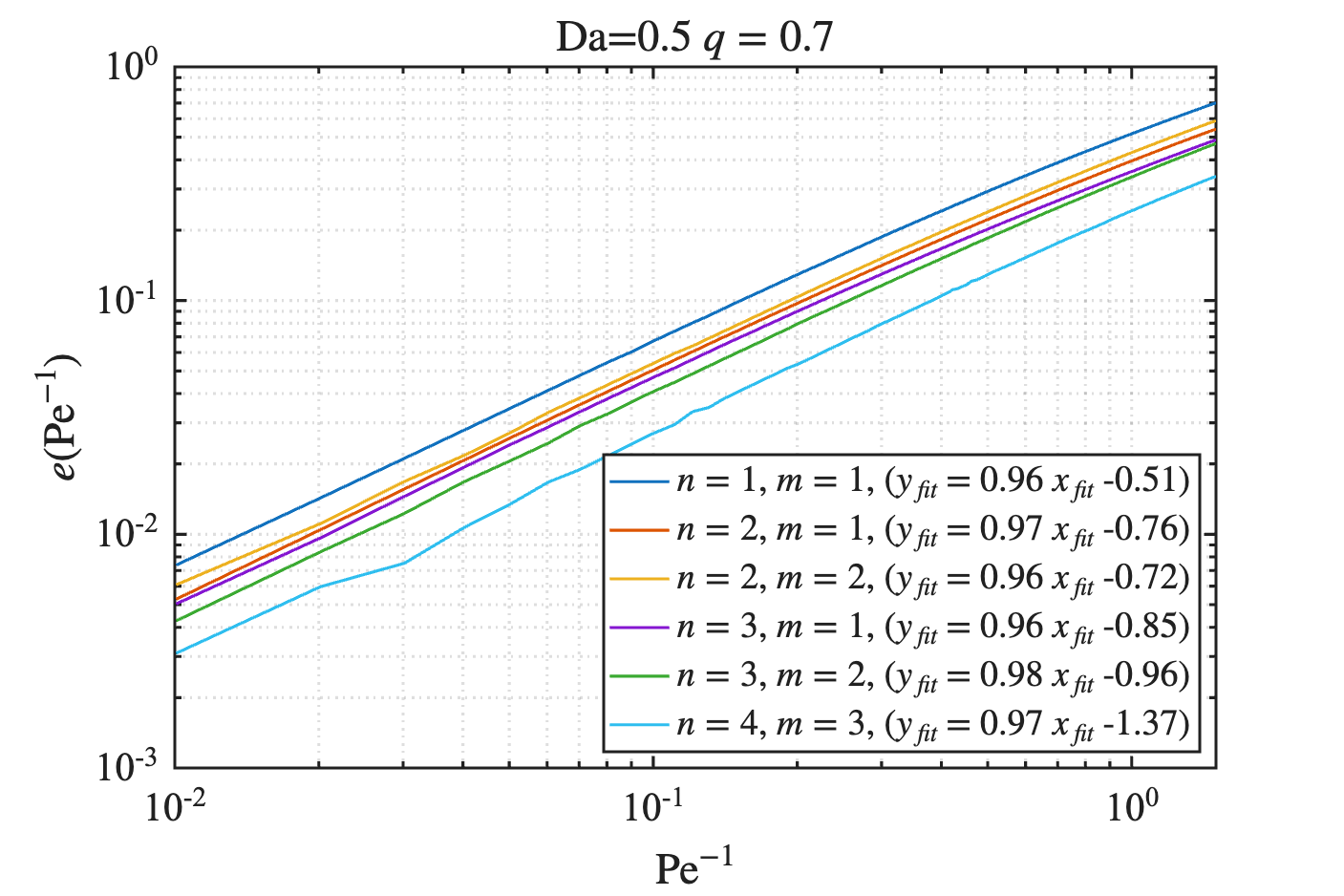}
\includegraphics[scale=0.35]{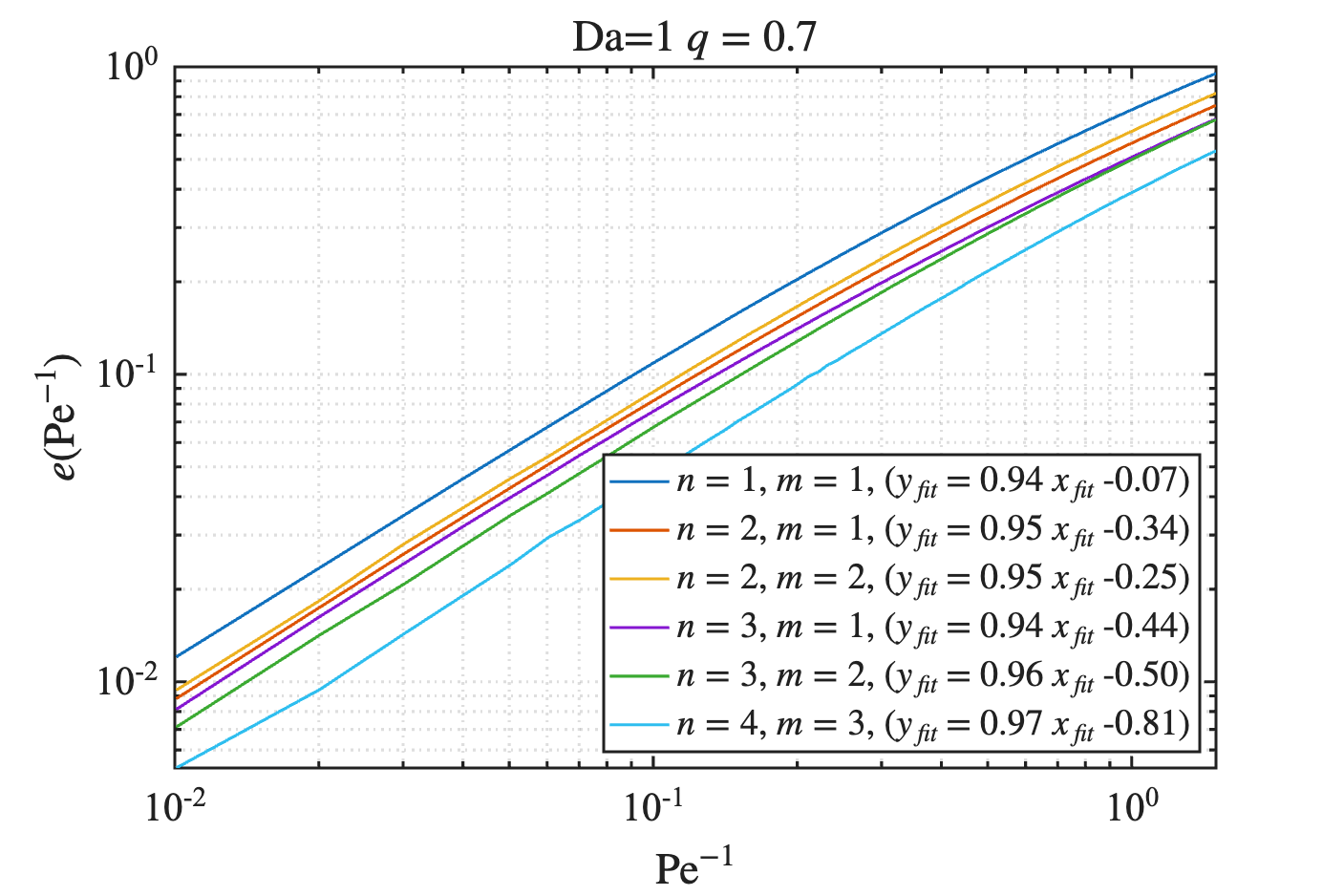}
\includegraphics[scale=0.35]{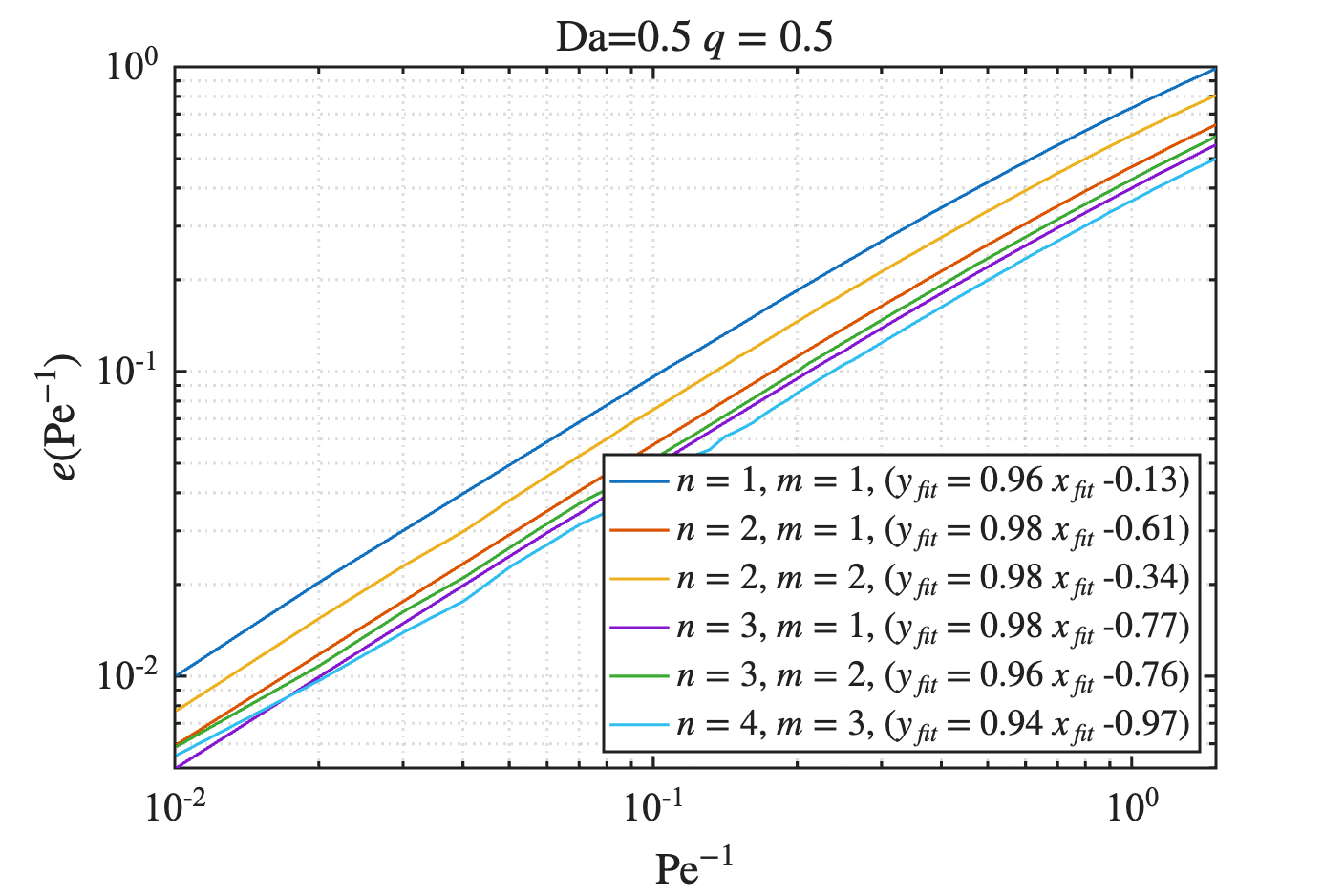}
\includegraphics[scale=0.35]{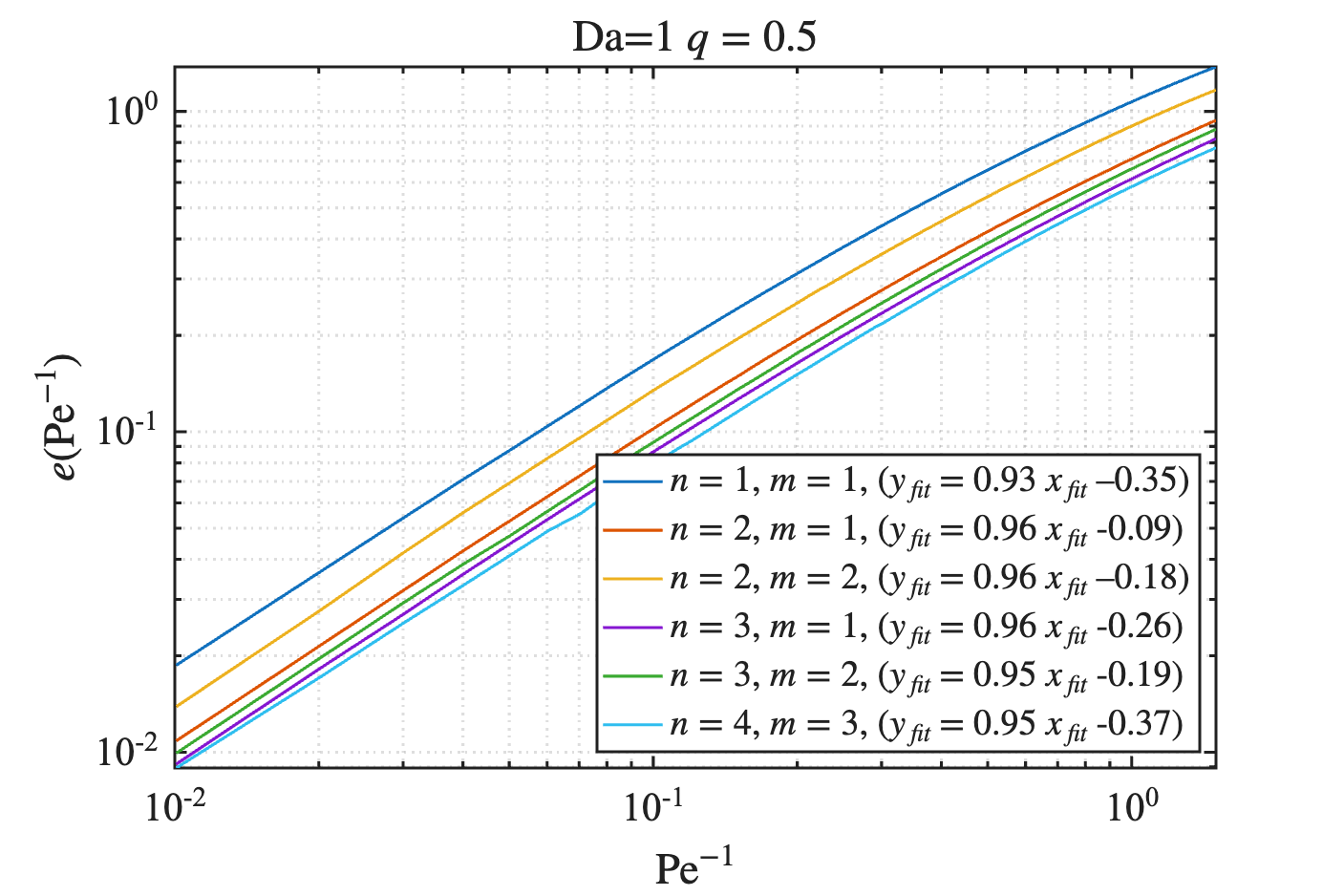}
\includegraphics[scale=0.35]{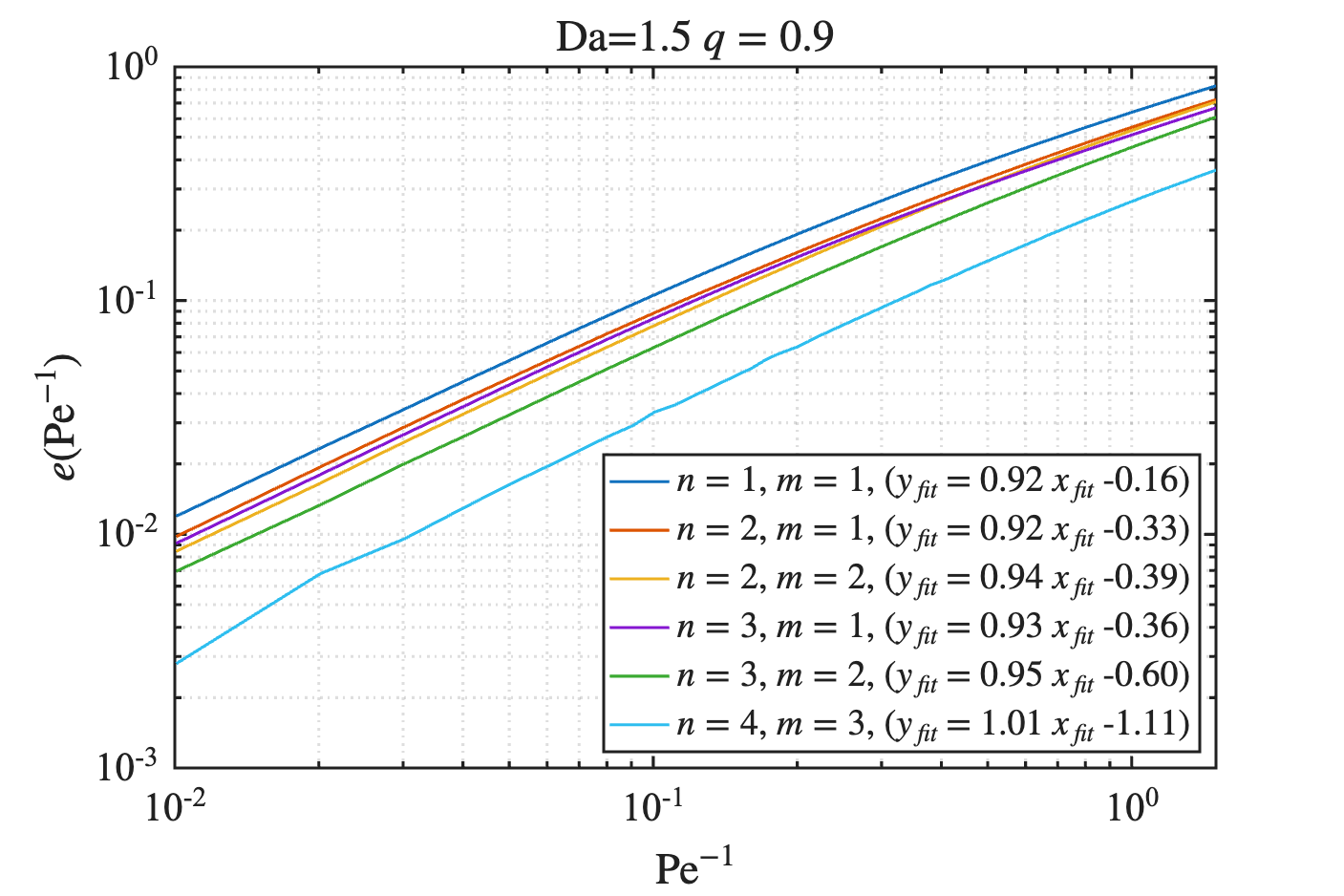
}
\caption{The $L^2$-norm of the difference between the travelling-wave profiles obtained by solving the system \eqref{eq:TW} with $F(0)=1/2$, as a function of the inverse P\'eclet number. The plots are presented in logarithmic scales, and fitting to a straight line fit for $\Pe\in[0.01,0.25]$ is provided. See Section~\ref{sec:sens} for a description of the numerical resolutions.}
	\label{Fig:L2norm}
\end{figure}

Expression \eqref{eq:ePe} provides a robust mathematical measure of the distance to the leading-order curve within a bounded domain whose limits are near the equilibrium states. Nevertheless, one of the most critical parameters in filter design is the breakthrough time, i.e., the moment at which the contaminant begins to appear in the outflow from the filter.

In practical applications, the contaminated fluid enters the filter, displacing the clean fluid while contaminant molecules are simultaneously adsorbed by the adsorbent material. Ideally, the fluid at the outlet should remain contaminant-free for as long as possible. However, once the filter becomes saturated, contaminant molecules start to break through, and the filter must be replaced. Accurately anticipating when a breakthrough will occur is of utmost importance, as the ability to predict the breakthrough time is essential for effective filter design. 

To compute the breakthrough time by means of the travelling-wave approximation, one must fix a small enough positive value for the concentration and compute the time it takes to achieve a given breakthrough concentration. This is due to the fact that the zero concentration value is an equilibrium point to which the solutions tend, but never actually reach. As a measure, we have chosen to compute the time needed to jump from a concentration of 10$^{-4}$ (0.01\% of the initial concentration) to 0.01 (1\% of the initial concentration) by solving system \eqref{eq:TW} with  $F(0)=1/2$. Denoting by $t_{b0}$ the corresponding time when the inverse P\'eclet number is neglected, and $t_{b\Pe}$ for $\Pe\neq 0$, we compute the relative error
\begin{equation}
e_{BT}(\Pe) = \frac{t_{b\Pe}-t_{b0}}{t_{b0}}\,,
	\label{eq:ePeBT}
\end{equation} 
for which we note that we do not take the absolute value since we also want to detect if $t_{b\Pe}>t_{b0}$. 

In Figure~\ref{Fig:BT1}, we present the curves of $e_{BT}(\Pe)$ over the range $\Pe \in [0,1.5]$ for various combinations of the order parameters $(m,n)$ and different values of $q$ and $\Da$. We also use a log-log scale and perform a linear fit to reveal that $e_{BR}$ linearly depends on $\Pe$. However, in this new set of simulations, we employ the whole range of $\Pe$ to obtain the linear fit. Remarkably, we observe that the errors stay mainly linear even for values of $Pe$ that are of order 1. Additionally, we observe that, in all cases, the breakthrough time with diffusion, $t_{b\Pe}$, exceeds the baseline value, $t_{b0}$. This implies that using $t_{b0}$ as a predictive estimate ensures the outlet concentration remains below 1\% of the inlet concentration, thus guaranteeing the filter is replaced before the outlet level exceeds the prescribed threshold. Therefore, the trends shown in Figure~\ref{Fig:BT1} highlight the strong robustness of the breakthrough time with respect to variations in the inverse P\'eclet number.
\begin{figure}
	\centering
	\includegraphics[scale=0.35]{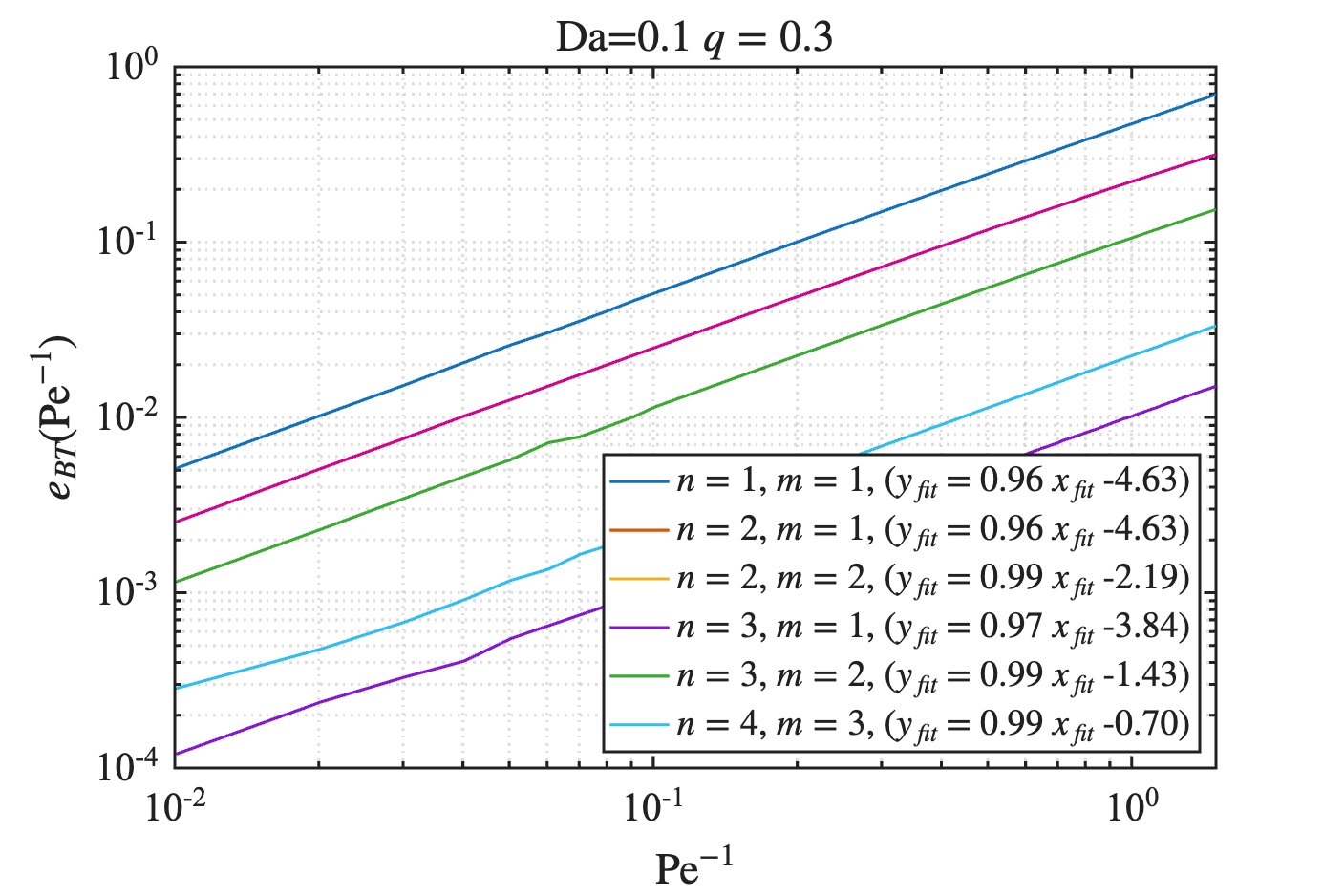}
	\includegraphics[scale=0.35]{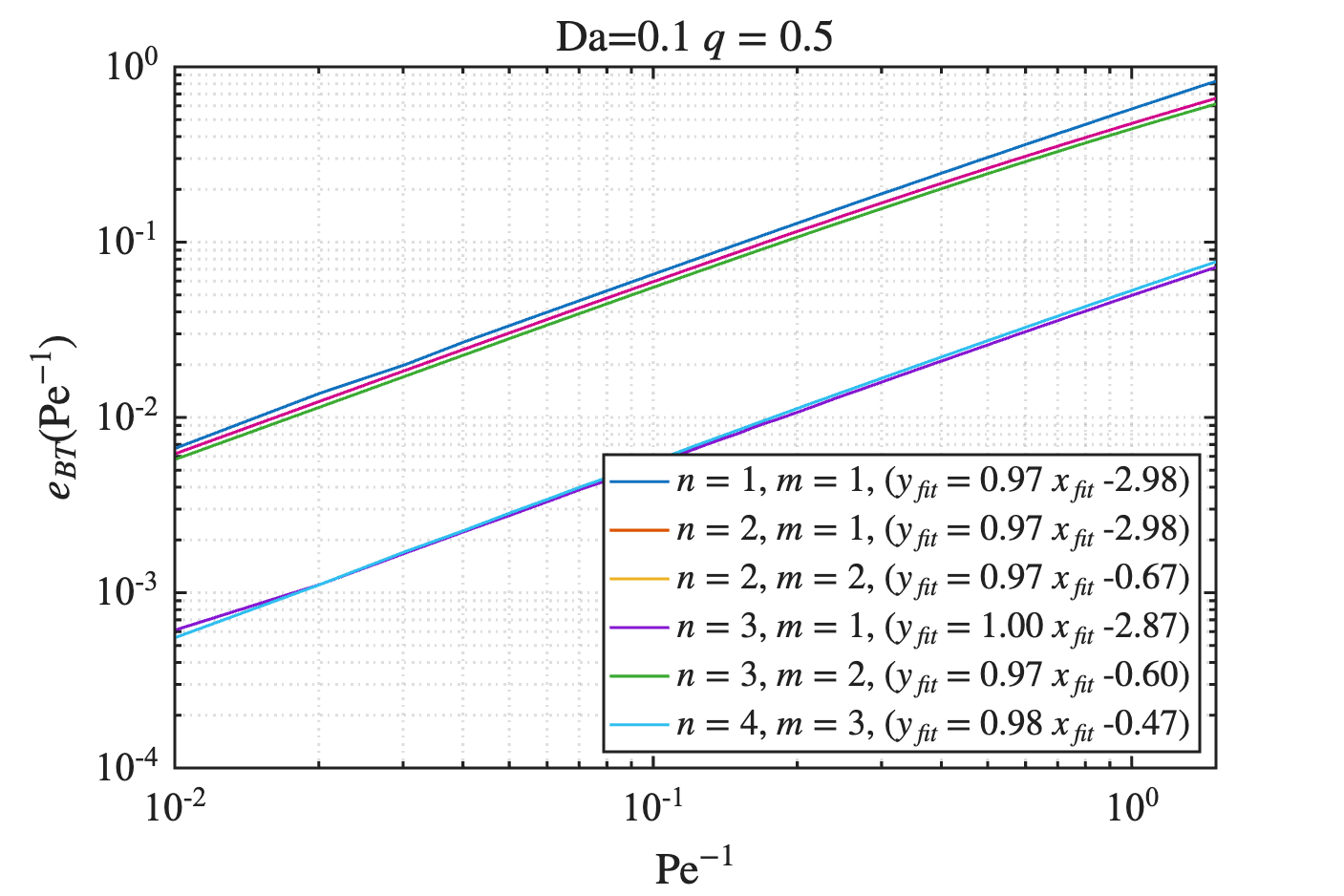}
	\includegraphics[scale=0.35]{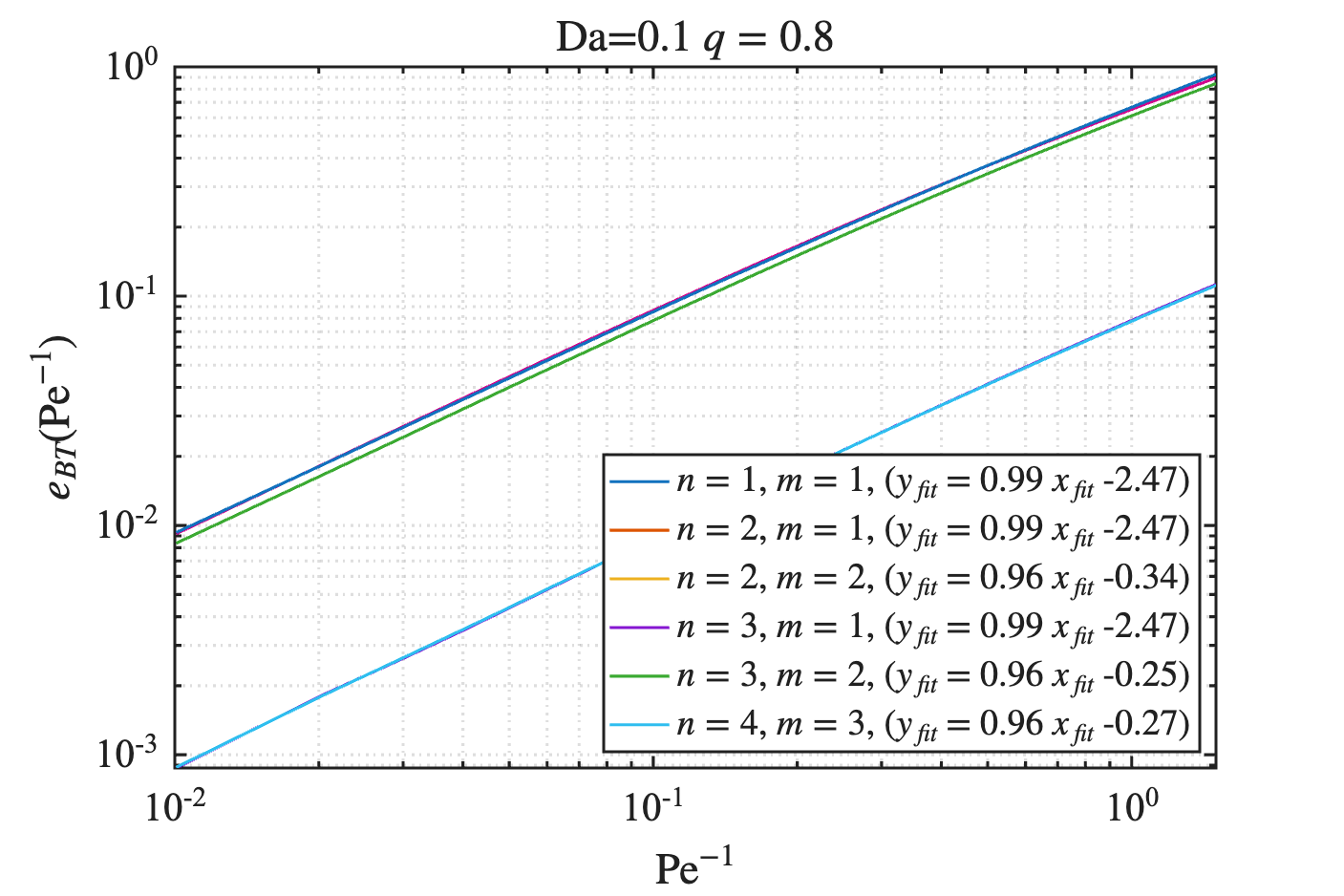}
	\includegraphics[scale=0.35]{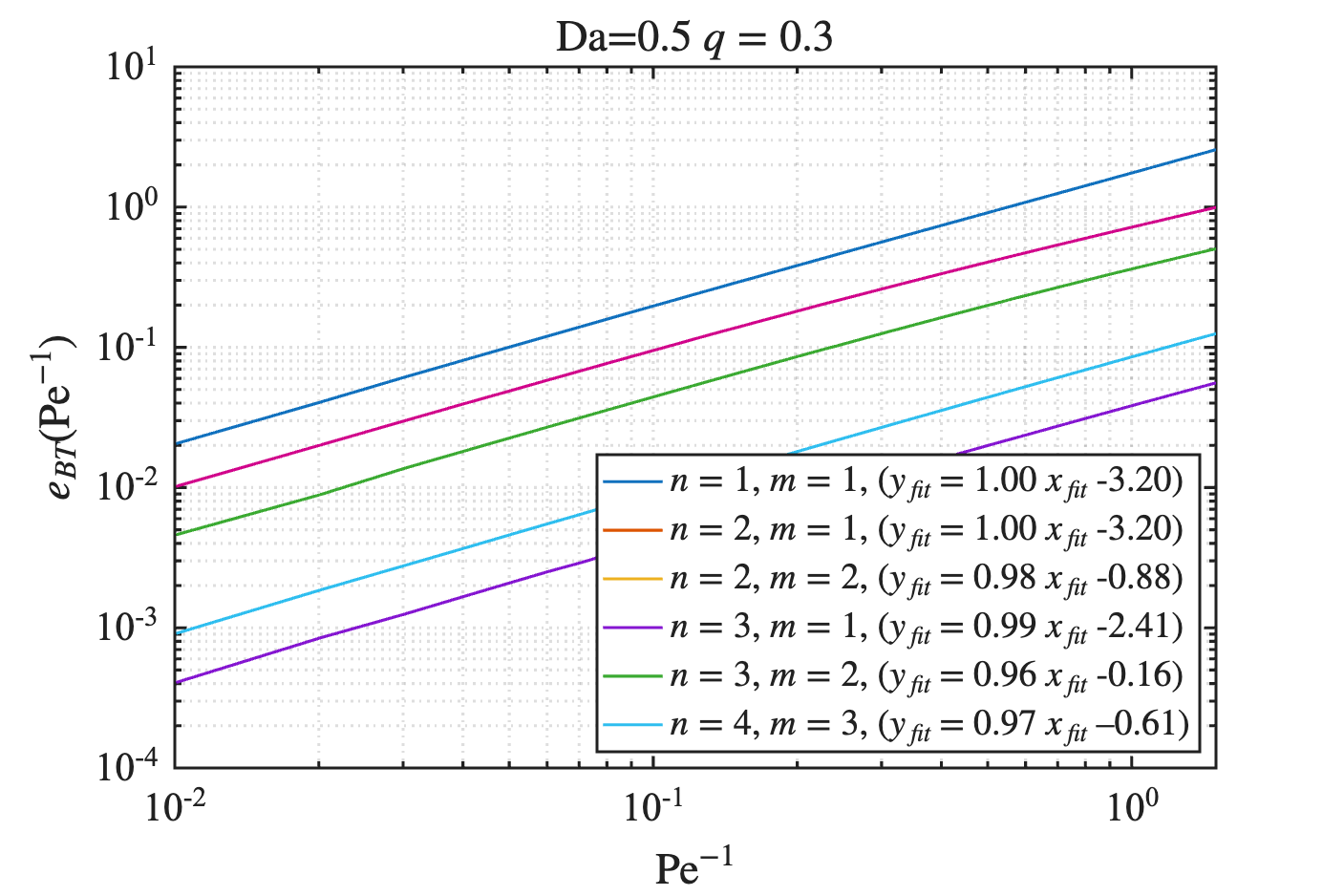}
	\includegraphics[scale=0.35]{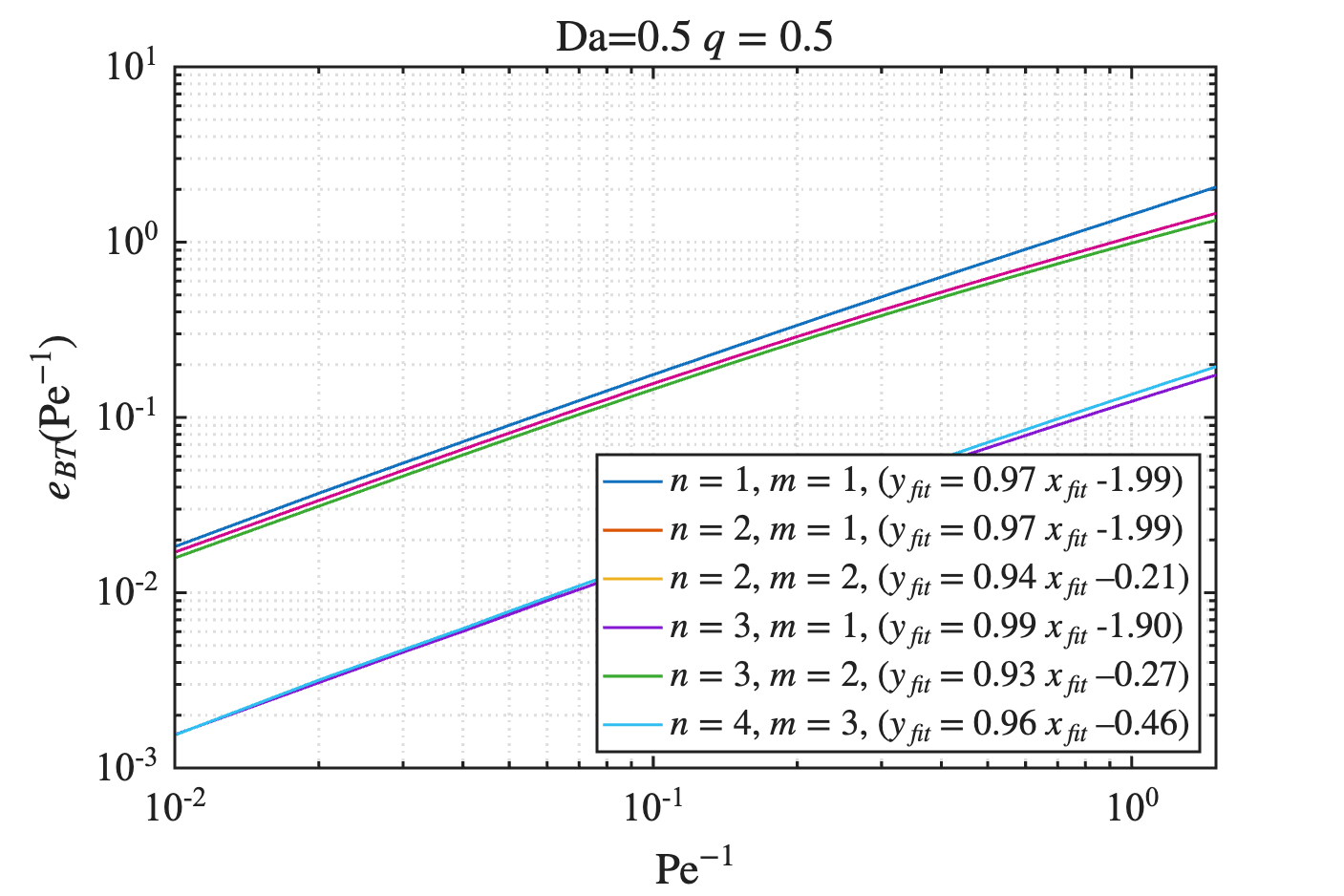}
	\includegraphics[scale=0.35]{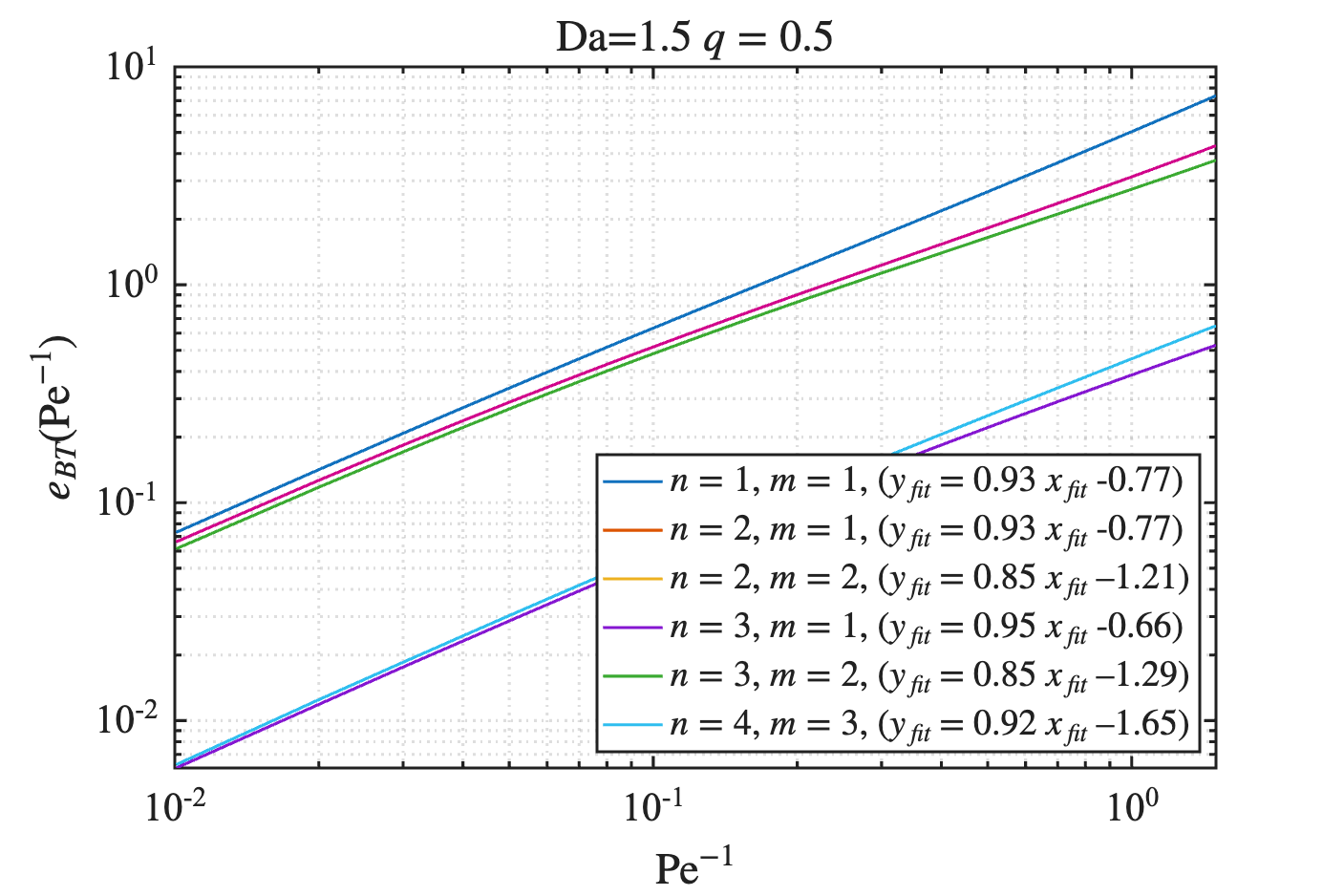}
	\caption{Relative error of the breakthrough time as a function of the inverse P\'eclet number, as defined in \eqref{eq:ePeBT} for different combinations of the order parameters, $(m,n)$, and for different values of $\Da$ and $q$. The plots are presented in logarithmic scales, and fitting to a straight line fit for $\Pe\in[0.01,0.25]$ is provided. See Section~\ref{sec:sens} for a description of the numerical resolutions. The values are obtained from the system \eqref{eq:TW} with $F(0)=1/2$. See Section~\ref{sec:sens} for a description of the numerical resolutions.}
	\label{Fig:BT1}
\end{figure}

\section{Travelling-wave exploration as solutions to the partial differential equation problem}
\label{sec:pde}

Here, we present additional evidence supporting the consistency of the travelling-wave approximation. Figures from \ref{fig:PDE1} to \ref{fig:PDE9} illustrate the evolution of the concentration, $c(x,t)$, obtained by solving system \eqref{eq:Sist_ND} using a Scharfetter-Gummel Discretization Scheme for different values of the reaction exponents, $(m,n)$, inverse P\'eclet, $\Pe$, Damk\"ohler number, $\Da$, and the saturation value of the filter, $q_e$. We choose not to plot $q(x,t)$ because it does not contribute to any other relevant information, aside from the observation that it appears to have a slightly longer transient. In what follows, all the magnitudes will be given in non-dimensional form. 

At this point, we must recall that the time and length scales depend on the constants $k_\text{ad}, k_\text{de}$ and $q_\text{max}$ which cannot be directly measured and are determined a posteriori by fitting experimental data. Therefore, the non-dimensional spatial length used in this section is arbitrary and has been selected to properly demonstrate well-formed travelling waves. One might question whether the corresponding physical column length would be sufficient to allow for the full development of these waves. In \cite{SULAYMON2014325} and \cite{Valverde24}, the authors demonstrate that, in most experiments, adsorption reactions typically occur at a much faster rate than advective transport. This implies that $\Da$ is typically very small, which suggests that the wave motion is expected to be much slower than the fluid velocity, as would be expected for a functionally effective filter:
$$v = \frac{\mathcal{L}}{\tau} \hat{v} = u_\text{in}\Da \frac{1}{q_e+\Da}\ll u_\text{in}.$$
The plots in figures~\ref{fig:PDE1} to \ref{fig:PDE9} (included in Appendix~\ref{app}) show transients that rapidly evolve toward a travelling wave profile. Furthermore, laboratory experiments are typically conducted with adsorbents designed for long-term operation, meaning the contaminant is not released  for several hours. This suggests that the waves have ample time to develop. Finally, figure~\ref{Fig:figPDE1} in Section~\ref{sec:model} demonstrates that the laboratory column length is sufficient for the waves to reach full development.

The first column in figures~\ref{fig:PDE1} to \ref{fig:PDE9} depicts the profiles $c(x,\cdot)$ at different times, $t$. In the second column, we neglect the initial transient by neglecting the curves before $0.4$ times the final time, and we shift $c(x,\cdot)$ so that all the curves intersect at the point $(x,c) = (0,1/2)$. Finally, in the third column, we represent three level sets: $c(x,t)=1/4, c(x,t)=1/2$, and $c(x,t) = 3/4$.  It is crucial to emphasise that the adsorption and desorption rates, $k_\text{ad}, k_\text{de}$, and the maximum capacity of the filter, $q_\text{max}$, are typically determined through data fitting in specific experiments. In some cases, the adsorption reaction exponents are known, while in others, the actual adsorption reaction is not well-characterised, necessitating the determination of the global orders, $m,n$, through fitting as well. Consequently, we have selected a set of values that may represent a physical experimental setting but do not specifically correspond to any particular one.    

In all these figures, one can observe that, after an initial transient, the solutions maintain their shape as they approach the outlet. As explained in Section~\ref{sec:model}, the fact that the level sets in the third column appear as parallel straight lines is consistent with the concentration profiles maintaining their shape as they move towards the outlet at a constant velocity. In fact, the velocity of the front, $v$, whose expression is given in \eqref{eq:v}, should correspond to the slopes of the level sets, which we denote by $v_s$. In table~\ref{tab:vel}, we compare $v$ with $v_s$ by computing the relative error given by
\begin{equation}
	v^*_\text{err}=\left|\frac{v-v_*}{v} \right|\cdot 100,
	\label{eq:v_err}
\end{equation}
where $v_\text{err}^*$ represents the set of velocity values obtained by fitting the level sets $c_{1/4} = c(x,t) - 1/4$, $c_{1/2} = c(x,t) - 1/2$, and $c_{3/4} = c(x,t) - 3/4$ to a straight line for various values of $q_e$, $\Da$, and $\Pe$, and for different combinations of $(m,n)$ such that $n \geq m$. For a given simulation, the values of $v_\text{err}^{1/4}$, $v_\text{err}^{1/2}$, and $v_\text{err}^{3/4}$ should be approximately equal and close to the predicted velocity given in \eqref{eq:v}. 

Figures~\ref{fig:PDE1} to \ref{fig:PDE9} demonstrate that, in general, the profiles tend to coincide around $c=1/2$, while they exhibit distinct differences in the regions closer to $c=1$, except for certain cases with $(m,n)=(3,4)$, where the profiles exhibit greater disparities around $c=0$. Table~\ref{tab:vel} indicates that the relative errors appear to be diminished for smaller values of the reaction exponents. This observation could be attributed to extended transients. Nevertheless, the errors, in overall, remain remarkably small, which further reinforces the validity of the travelling-wave approximation.

\begin{table}
\begin{center}
\begin{tabular}{ |c|c| l l l l l l l| }
\hline 
\multirow{2}{4em}{$(q_e,\Da)$} & \multirow{2}{4em}{$\Pe$} & & \multicolumn{6}{c|}{$v^*_\text{err}$ for $(m,n):$} \\
\cline{3-9}		
& &  & $(1,1)$ & $(1,2)$ & $(1,3)$ & $(2,2)$ & $(2,3)$ & $(3,4)$\\
\hline
\multirow{12}{4em}{$(0.7, 1)$} & \multirow{3}{4em}{0.1} & $v_\text{err}^{1/4}$ & 0.69  &  0.19 & 0.56 & \textbf{1.47} &  \textbf{0.92} & \textbf{2.79} \\
& &  $v_\text{err}^{1/2}$ & \textbf{0.37} &   \textbf{0.05} &   \textbf{0.14}  & 3.93 &    3.67  &  7.26\\
& &  $v_\text{err}^{3/4}$ & 1.65 & 0.30 & 0.77 & 6.12 &  6.59&   14.80\\
\cline{2-9} 
& \multirow{3}{4em}{0.5} & $v_\text{err}^{1/4}$ &1.56  & 0.63& 0.85 & \textbf{0.85}   & \textbf{0.65} & \textbf{3.79} \\ 
& & $v_\text{err}^{1/2}$ &\textbf{0.51} & \textbf{0.14} &\textbf{0.16}& 4.25& 4.16 &7.99\\
& & $v_\text{err}^{3/4}$ &3.14 &0.84& 1.32 &7.63 &7.99& 16.43
 \\
 \cline{2-9}
& \multirow{3}{4em}{1} & $v_\text{err}^{1/4}$ & 3.11  & 1.59  & 1.63 & \textbf{0.33} & \textbf{0.00} & \textbf{5.30}  \\ 
& & $v_\text{err}^{1/2}$ & \textbf{0.64} & \textbf{0.39} &\textbf{0.22} &4.51 &4.70   &8.83\\
& & $v_\text{err}^{3/4}$ & 5.69 & 2.25 &2.68 &9.97 &10.20 &18.91\\
\cline{2-9}
& \multirow{3}{4em}{1.5} & $v_\text{err}^{1/4}$ & 4.70 & 2.79& 2.56 &\textbf{1.62} &\textbf{0.78} &\textbf{6.64}\\
& &  $v_\text{err}^{1/2}$ & \textbf{0.72}  & \textbf{0.62}  & \textbf{0.33}  &4.62 &5.04 & 9.38\\
& & $v_\text{err}^{3/4}$ &8.10 & 6 4.04  & 4.30 & 12.13 & 12.28 &21.08\\
\hline
\multirow{12}{4em}{$(0.9, 0.1)$} & \multirow{3}{4em}{0.1} & $v_\text{err}^{1/4}$  & \textbf{0.01} & 0.77 & 9.00 & \textbf{3.13} & 7.32 & 19.51 \\
& & $v_\text{err}^{1/2}$ & 0.03&   \textbf{0.37}& 7.11& 6.73&  \textbf{0.76}& \textbf{3.16} \\
& & $v_\text{err}^{3/4}$ & 0.07& 0.96& \textbf{0.33}& 11.45& 10.78& 14.54
 \\
\cline{2-9}
& \multirow{3}{4em}{0.5} & $v_\text{err}^{1/4}$  & \textbf{0} & 0.76&    8.95 & 3.45&   7.03 &19.17  \\
& & $v_\text{err}^{1/2}$ & 0.07& \textbf{0.29}& 6.97& \textbf{7.28} &\textbf{0.21} &\textbf{2.58} \\ 
& & $v_\text{err}^{3/4}$ & 0.14 & 1.20 & \textbf{0.54} & 12.31 & 11.53 & 15.34 \\
\cline{2-9}
& \multirow{3}{4em}{1} & $v_\text{err}^{1/4}$  & \textbf{0.06} &0.72 &8.90 &\textbf{3.80}& 6.69& 18.83  \\ 
& & $v_\text{err}^{1/2}$ & 0.18 &\textbf{0.16} &6.65 &8.10& \textbf{0.69}& \textbf{1.60} \\ 
& & $v_\text{err}^{3/4}$ & 0.38 &1.61 &\textbf{1.16} &13.88 &12.95 &16.77 \\
\cline{2-9}
& \multirow{3}{4em}{1.5} & $v_\text{err}^{1/4}$  & \textbf{0.11}& 0.76& 8.86& \textbf{4.20}& 6.39& 18.71 \\
& & $v_\text{err}^{1/2}$ & 0.35& \textbf{0.03}& 6.27& 9.00& \textbf{1.59}& \textbf{0.80}\\ 
& & $v_\text{err}^{3/4}$ & 0.80& 2.14& \textbf{2.01}& 15.46& 14.58& 18.32\\
\hline
\multirow{12}{4em}{$(0.9, 1.5)$} & \multirow{3}{4em}{0.1} & $c_{1/4}$  & 3.93& 3.78& 4.05& 3.21& 3.73& 5.18 \\ 
& & $v_\text{err}^{1/2}$ & \textbf{0.89}& \textbf{1.20}& \textbf{1.06}& \textbf{0.04}& \textbf{0.08}& \textbf{0.49} \\
& & $v_\text{err}^{3/4}$ & 3.46& 3.02& 3.39& 4.22& 4.66& 5.70\\
\cline{2-9}
& \multirow{3}{4em}{0.5} & $v_\text{err}^{1/4}$  & 6.76& 6.73& 6.79& 5.84& 6.18& 7.50 \\
& & $v_\text{err}^{1/2}$ & \textbf{1.00}& \textbf{1.46}& \textbf{1.27}& \textbf{0.18}&    \textbf{0} &   \textbf{0.55} \\ 
& & $v_\text{err}^{3/4}$ & 6.36& 5.87& 6.12& 6.94& 7.32& 8.47\\
\cline{2-9}
& \multirow{3}{4em}{1} & $v_\text{err}^{1/4}$  & 10.12& 10.24& 10.13& 9.15& 9.27& 10.56 \\ 
& & $c_{1/2}$ & \textbf{1.11}& \textbf{1.67}& \textbf{1.48}& \textbf{0.41}& \textbf{0.16}& \textbf{0.49}
 \\ 
& & $v_\text{err}^{3/4}$ & 9.80& 9.29& 9.45& 10.27& 10.55& 11.73  \\
\cline{2-9}
& \multirow{3}{4em}{1.5} & $v_\text{err}^{1/4}$  & 12.86& 13.05& 12.87& 11.90& 11.88& 13.14 \\
& & $v_\text{err}^{1/2}$ & \textbf{1.27}& \textbf{1.86}& \textbf{1.69}& \textbf{0.64}& \textbf{0.37}& \textbf{0.34}\\ 
& & $v_\text{err}^{3/4}$ & 12.43& 11.92& 12.05& 12.86& 13.08& 14.26 \\
\hline
\end{tabular}
\end{center}
\caption{Comparison between the expression \eqref{eq:v} and the numerically computed front velocity (see Section~\ref{sec:pde}). The errors are quantified using the expression \eqref{eq:v_err}. For each simulation, the corresponding minimum error value is indicated in bold. This table should be interpreted in conjunction with the figures from~\ref{fig:PDE1} to~\ref{fig:PDE9}.}\label{tab:vel}	
\end{table}

\subsection{Solutions with non-zero initial conditions}

To derive the travelling-wave approximation, we assumed that the filter was initially fresh, i.e., $q(x,0)=0$ for all $x$, and the air within the filter was clean, i.e., $c(x,0)=0$ for all $x$. However, it is relevant to consider the scenario where neither of these conditions is satisfied. For instance, if one intends to recycle a filter that has been utilised but not yet exhausted, or if, for some reason, the air within the filter is initially not completely devoid of contaminant molecules. 

To gain insights into how concentration and adsorbed fraction would evolve in these two scenarios, we conducted two new sets of simulations with non-zero initial conditions. In the first set (first two rows in Figure~\ref{Fig:NonZIC}), we assume the air inside the filter has a uniform concentration of half the initial one, i.e., $c(x,0)=0.5$, while the adsorbent material is fresh, i.e., $q(x,0)=0$. In contrast, in the second set (third and fourth rows in Figure~\ref{Fig:NonZIC}), we assume the air inside the filter is clean, i.e., $c(x,0)=0$, but the adsorbent has a uniform adsorbed fraction of 0.4, i.e., $q(x,0) = 0.4$. 

We observe that when the filter is clean, but the inside concentration is initially non-zero (first and second rows in Figure~\ref{Fig:NonZIC}), the adsorbed fraction very rapidly jumps to a constant value. This happens at a very fast timescale because we are dealing with a partial differential equation of parabolic type, which are well known to have an infinite propagation speed. As for the concentration, contrary to what one might expect, it appears to evolve like a travelling wave, but the limit as $x\to\infty$ is neither zero nor 0.5, but another value smaller than the initial concentration inside the filter. This indicates that the filter will not be able to remove the whole initial amount of contaminant, but still the fluid at the breakthrough will have a lower concentration than the initial one for some time. The two constant values for $c$ and $q$ as $x\to\infty$ are related by the isotherm. However, further analysis is required in order to predict what these values will be.

When the air inside the filter is clean, but the adsorbent is not purely fresh (as shown in the third and fourth rows of Figure~\ref{Fig:NonZIC}), a similar phenomenon occurs. The solutions appear to converge to a constant non-zero value at infinity, except for the concentration when $m=1$ and $n=2$, which seems to vanish (or approach a very small value). 

These simulations highlight the importance of initial conditions, but a comprehensive analysis is beyond the scope of the current paper. Therefore, we leave this for future work. Additionally, the role of reaction exponents appears to be crucial in determining breakthrough values.

\section{Conclusions}
\label{sec:conc}

In this paper, the travelling-wave approximation is shown to accurately describe filter performance when the inverse P\'eclet number is small (see, for example, \cite{AGUARELES2023}, \cite{AUTON2024827}). We rigorously demonstrate that the solutions obtained by neglecting the inverse P\'eclet number (i.e., the leading-order approximation) are, in fact, the limiting solutions of the original system of equations given in \eqref{eq:TW} or \eqref{eq:fullODE}. Moreover, our numerical analysis reveals excellent agreement between the leading-order approximation and the solutions of \eqref{eq:TW}, even for moderate values of $\Pe$. 

The existence of travelling-wave solutions in the full system is proved by extending beyond the leading-order approximation. The governing equation is reformulated as a slow-fast system, where variables evolve at different time scales. The key point is the identification of heteroclinic connections, which represent transitions between equilibrium points representing the states $F=1$ and $F=0$. The study first confirms the presence of such connections and provides an implicit solution for the leading-order approximation where rapid transitions take place.  Through analytical continuation methods, it is demonstrated that the solutions of the full system remain close to those of the leading-order approximation, provided $\Pe$ is sufficiently small. This confirms that, even when diffusion effects are present, the leading-order approximation remains a valid predictive tool and the travelling wave persists for small but non-zero values of the inverse P\'eclet number. Figures \ref{Fig:solMaria_perNM_EspaiFases_FrontOna_A} and \ref{Fig:solMaria_perNM_EspaiFases_FrontOna_B} show that even when the inverse P\'eclet number is of order one, the similarity between the curves remains remarkably high. This is particularly important because the exact value of $\Pe$ is not always well-defined in experimental settings.

The breakthrough time is a critical parameter that must be determined with precision. Replacing adsorption materials incurs high economic and environmental costs, and industries cannot afford the risks associated with unintended contaminant emissions. As shown in Figure~\ref{Fig:BT1}, the leading-order approximation provides a reasonably accurate prediction of breakthrough times. Although the estimated values are lower than those obtained when $\Pe > 0$, they remain sufficiently close to serve as practical estimates.

Also, a key aspect to consider when developing a model that fits experimental data is the description of the adsorption reaction. In many cases, the physicochemical processes involved are not entirely understood, and the reaction orders cannot be determined \emph{a priori}. In such situations, one could consider fitting the order parameters $(m,n)$ to the available data using the leading-order approximation. The numerical and theoretical analyses presented in this work support and validate such a fitting approach. In fact, section~\ref{sec:pde} highlights the importance of the reaction exponents, along with the initial filter conditions on the final filter's performance.

Overall, the analysis performed highlights the importance of considering slow-fast dynamics in the study of travelling waves in this model. The persistence of these solutions reinforces the relevance of the approximation techniques used, while also providing a framework to extend the analysis to more complex scenarios. This study completes the work initiated in \cite{AGUARELES2023} and establishes a framework for the applicability of the results presented there.

\section*{Author contributions}
MA, JA-L, EB: Conceptualization, Methodology, Investigation, Formal analysis, Writing – original draft,  Writing – review \& editing. JA-L: Software. MA, EB: Supervision.

\section*{Acknowledgements}

This publication is part of the research projects PID2023-146332OB-C22 (funding the three authors), financed by the Spanish \emph{Agencia Estatal de Investigación} of \emph{Ministerio de Ciencia, Innovación y Universidades} and TED2021-131455A-I00 of the Agencia Estatal de Investigación (Spain) (funding E. Barrabés and M. Aguareles). M. Aguareles acknowledges the support of the “consolidated research group” (Ref 2021 SGR01352) of the Catalan Ministry of Research and Universities.  E. Barrab\'es is supported by the Spanish grant PID2021-123968NB-I00 (AEI/FEDER/UE). 
\section*{Conflict of interest}
All authors declare no conflicts of interest in this paper.

\appendix
\section{List of figures corresponding to Section~\ref{sec:pde}}
\label{app}
In this Appendix we include the set of figures discussed throughout Section~\ref{sec:pde}.

\begin{figure}
\centering
\includegraphics[width=\textwidth]{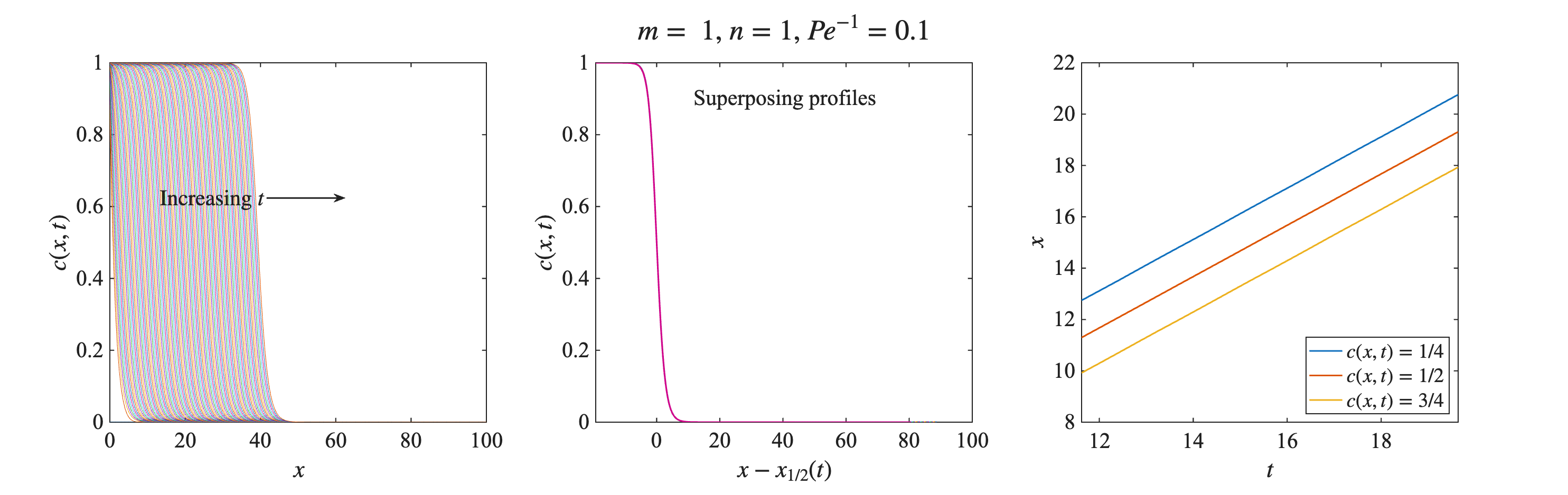}
\includegraphics[width=\textwidth]{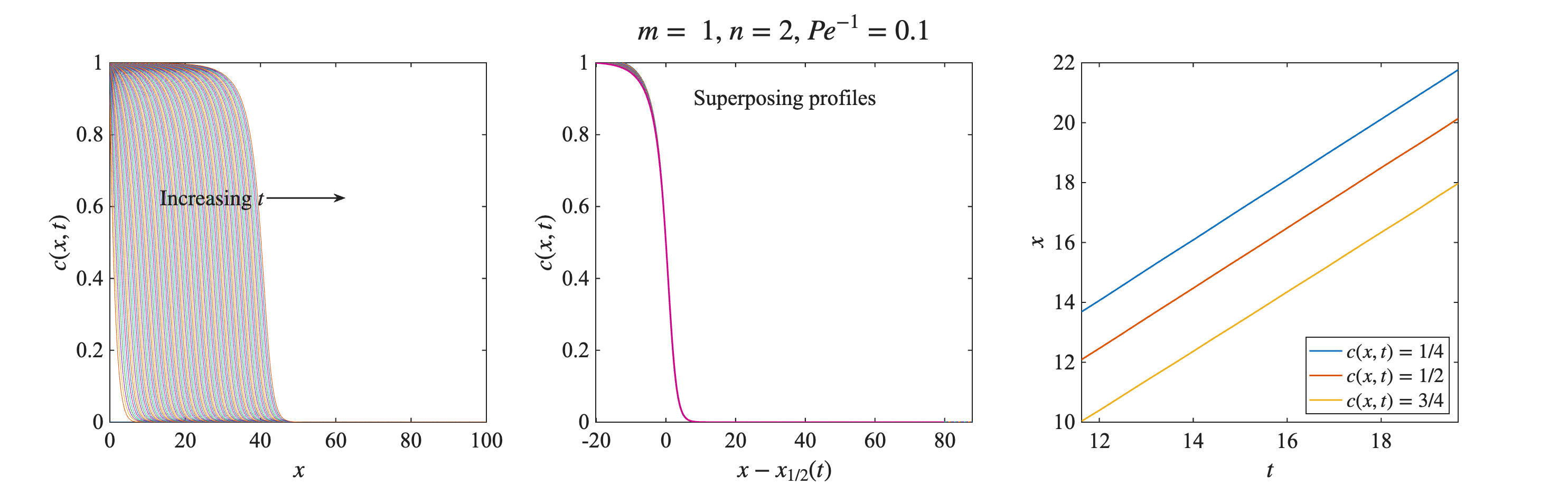}
\includegraphics[width=\textwidth]{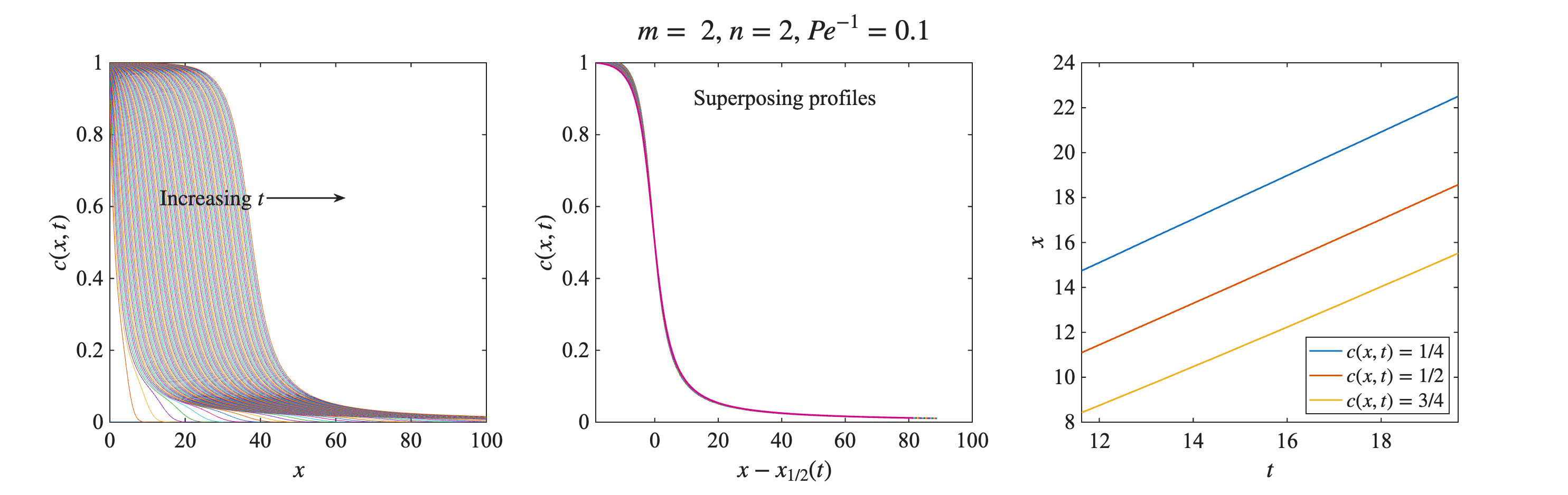}
\includegraphics[width=\textwidth]{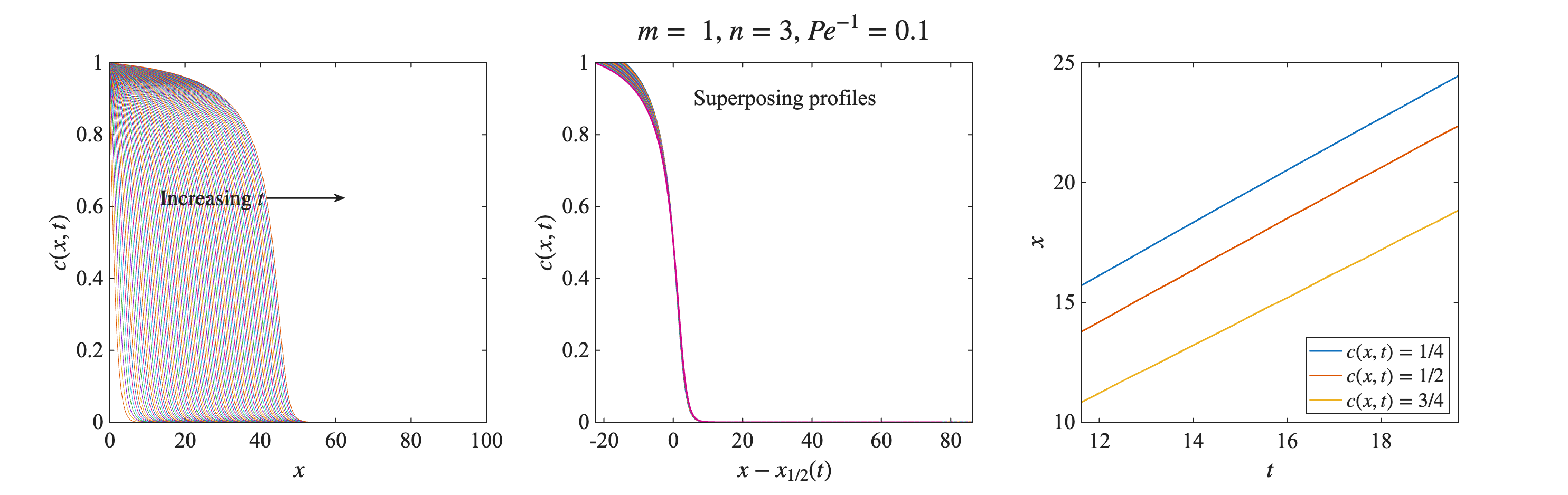}
\caption{Solution of system \eqref{eq:Sist_ND} with $q_e =0.9$ and $\Da = 0.1$. Left: evolution in time of $c(x,t)$. Middle: shifted curves to force the intersections at $(x-x_{1/2},c)=(0,1/2)$. Right: level sets for $c(x,t)=c^*$. Initial transient curves have been excluded from the middle and right plots.}
\label{fig:PDE1}
\end{figure}

\begin{figure}
\centering
\includegraphics[width=\textwidth]{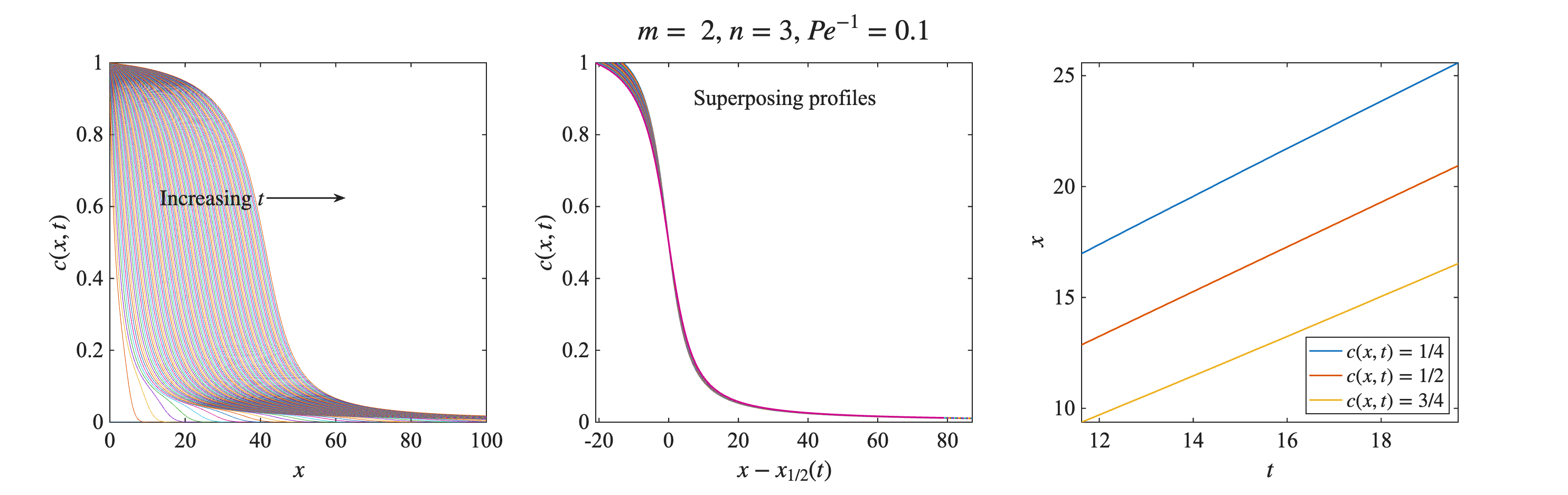}
\includegraphics[width=\textwidth]{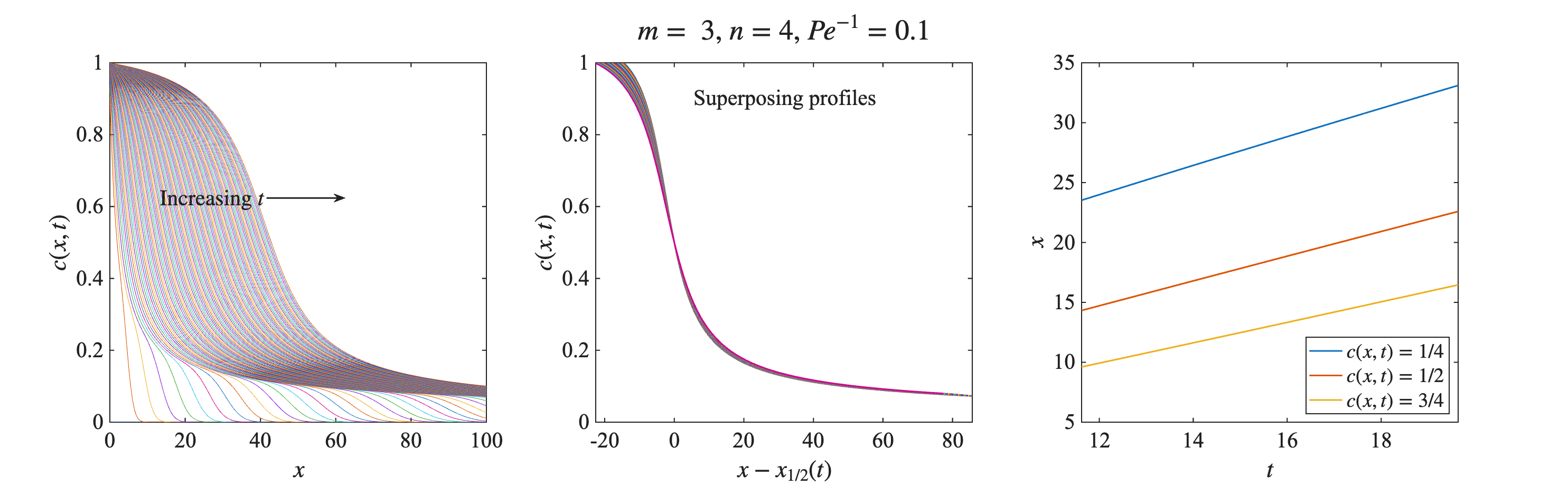}
\includegraphics[width=\textwidth]{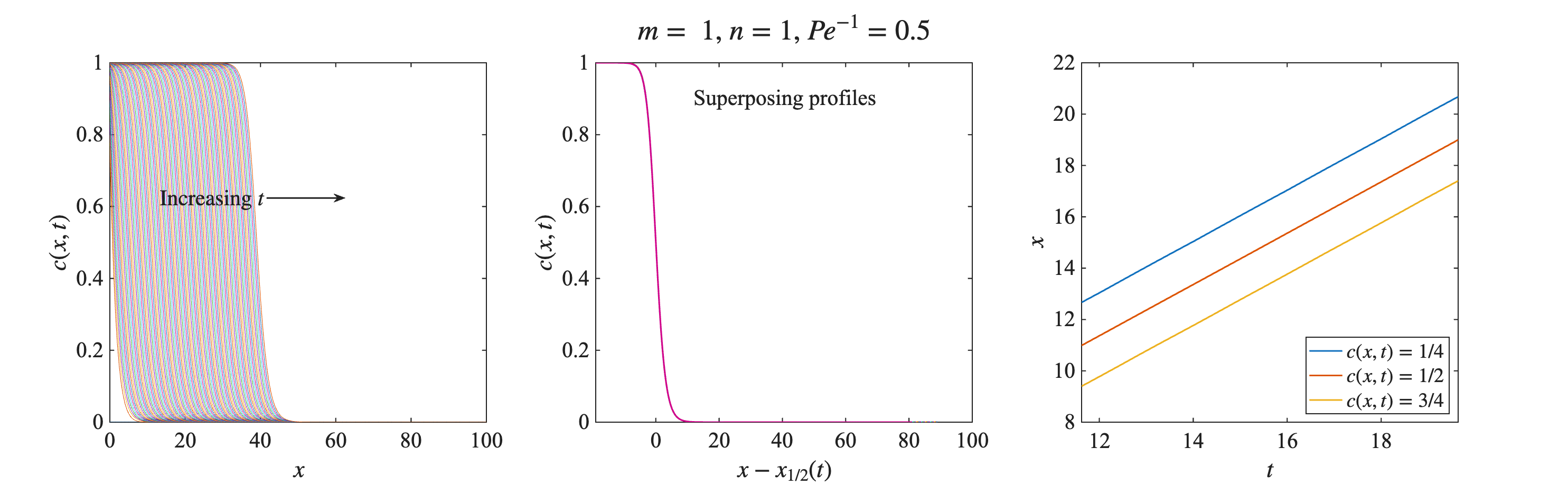}
\includegraphics[width=\textwidth]{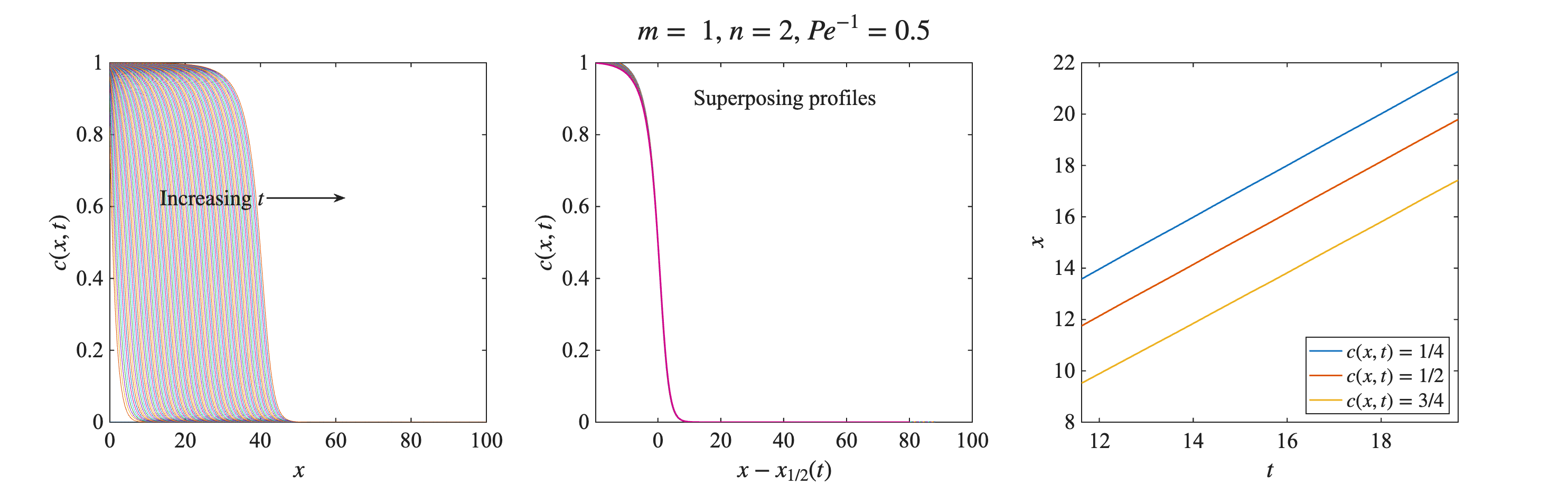}
\caption{Solution of system \eqref{eq:Sist_ND} with $q_e =0.9$ and $\Da = 0.1$. Left: evolution in time of $c(x,t)$. Middle: shifted curves to force the intersections at $(x-x_{1/2},c)=(0,1/2)$. Right: level sets for $c(x,t)=c^*$. Initial transient curves have been excluded from the middle and right plots.}
\label{fig:PDE2}
\end{figure}

\begin{figure}
\centering
	\includegraphics[width=\textwidth]{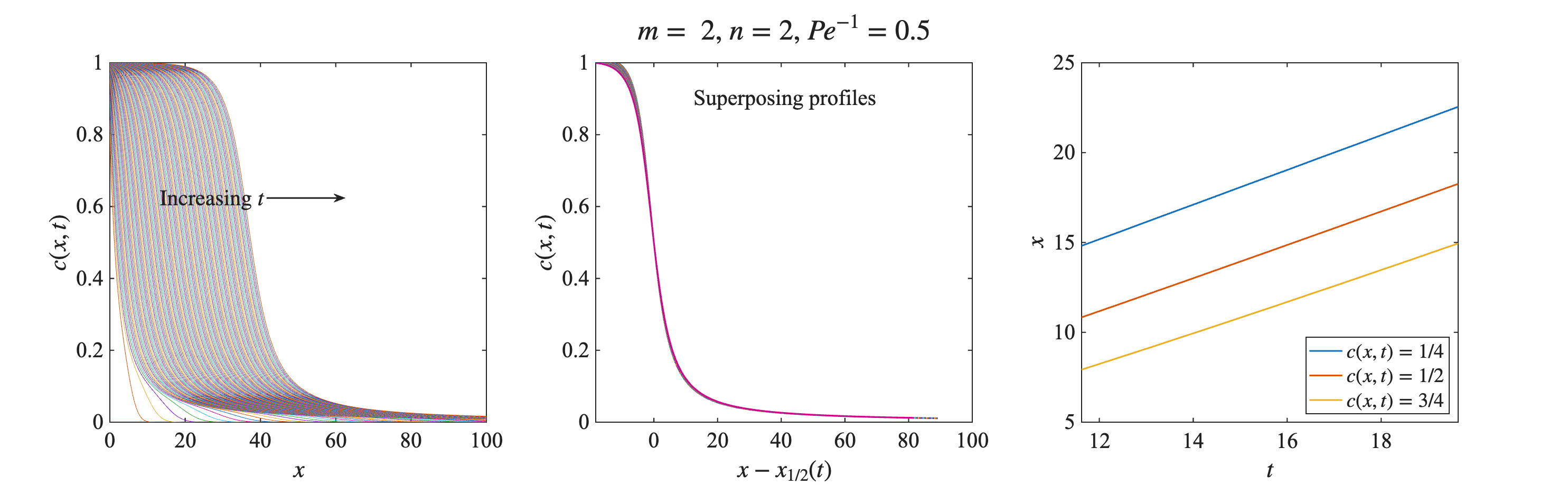}
\includegraphics[width=\textwidth]{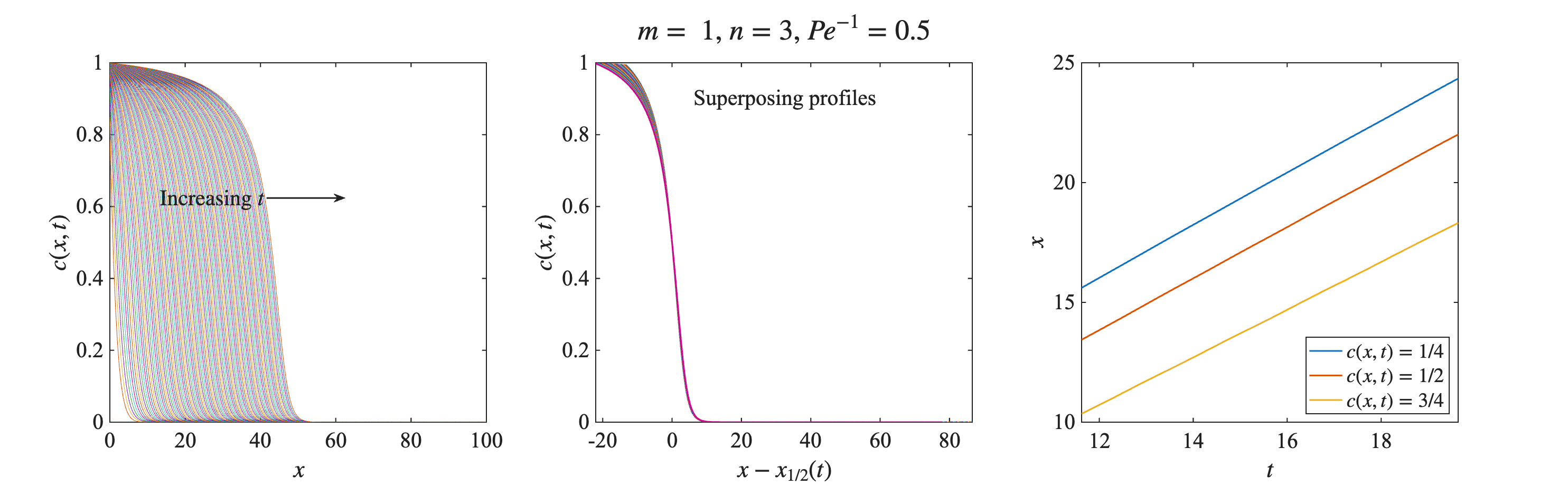}
	\includegraphics[width=\textwidth]{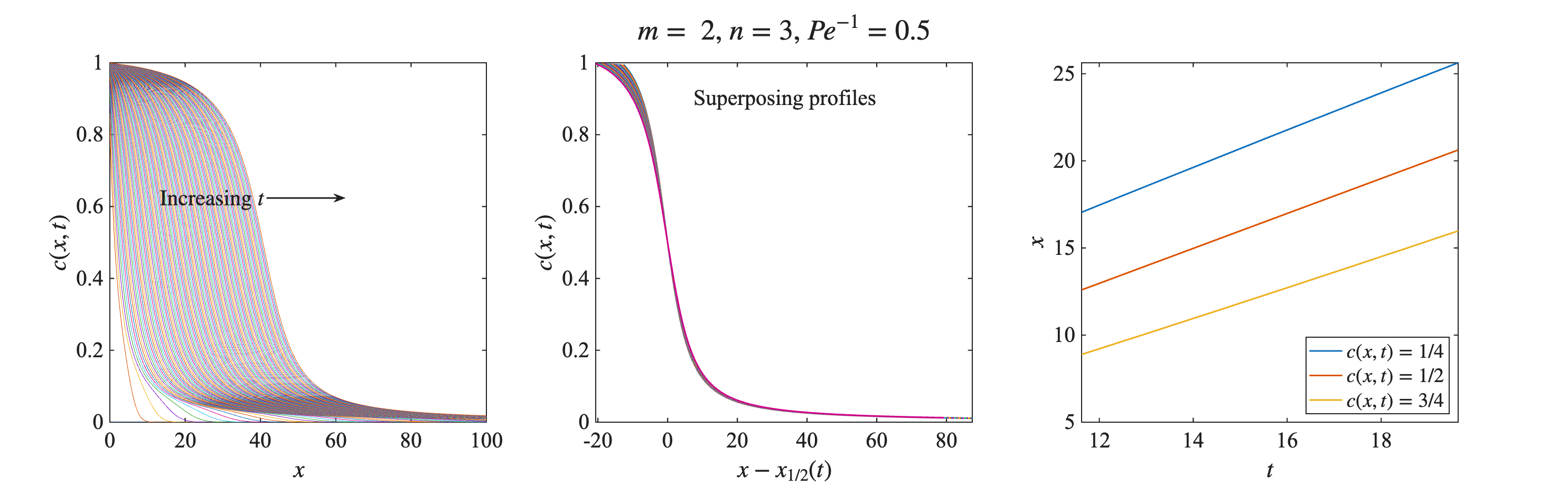}
\includegraphics[width=\textwidth]{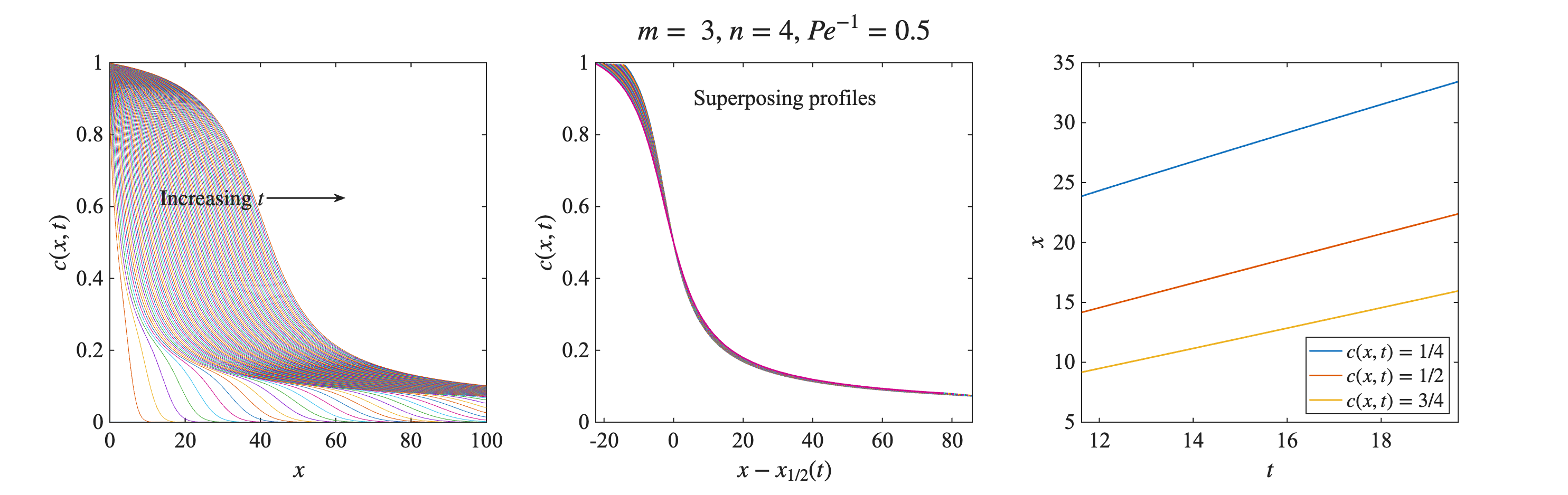}
\caption{Solution of system \eqref{eq:Sist_ND} with $q_e =0.9$ and $\Da = 0.1$. Left: evolution in time of $c(x,t)$. Middle: shifted curves to force the intersections at $(x-x_{1/2},c)=(0,1/2)$. Right: level sets for $c(x,t)=c^*$. Initial transient curves have been excluded from the middle and right plots.}
\label{fig:PDE3}
\end{figure}

\begin{figure}
\centering
\includegraphics[width=\textwidth]{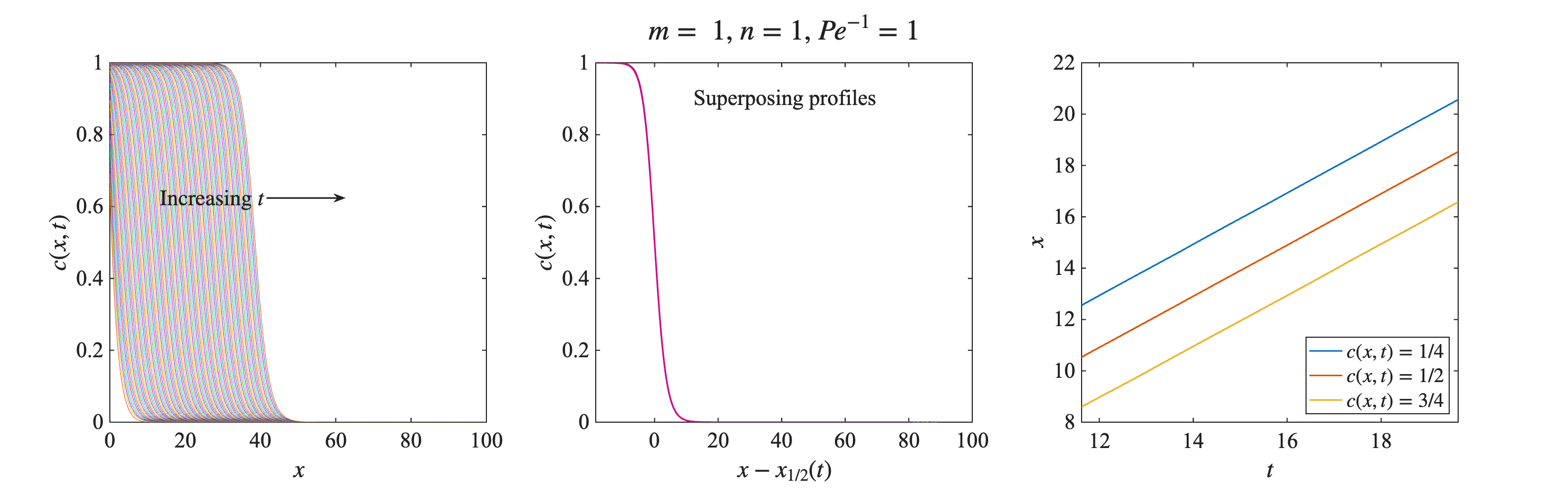}
\includegraphics[width=\textwidth]{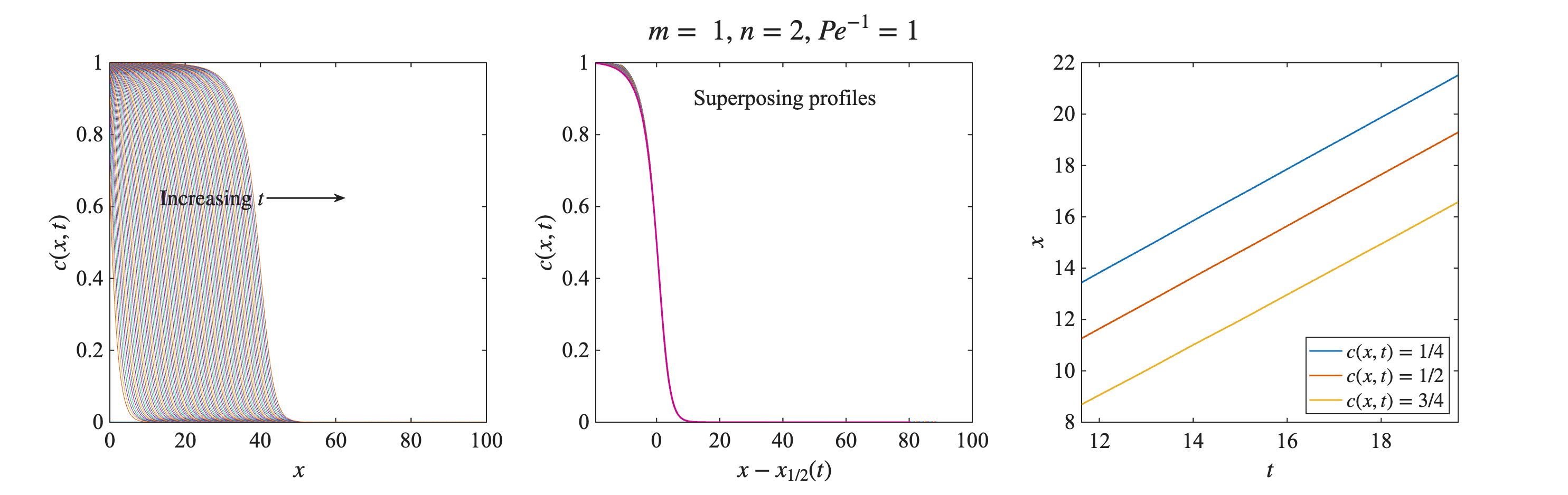}
\includegraphics[width=\textwidth]{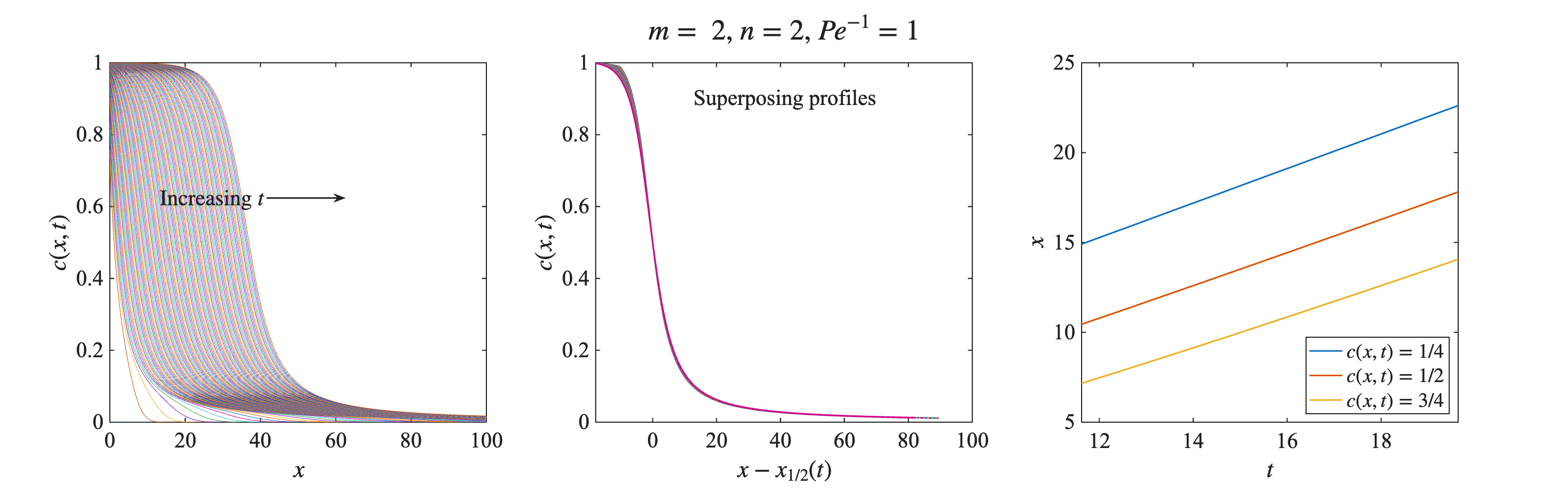}
\includegraphics[width=\textwidth]{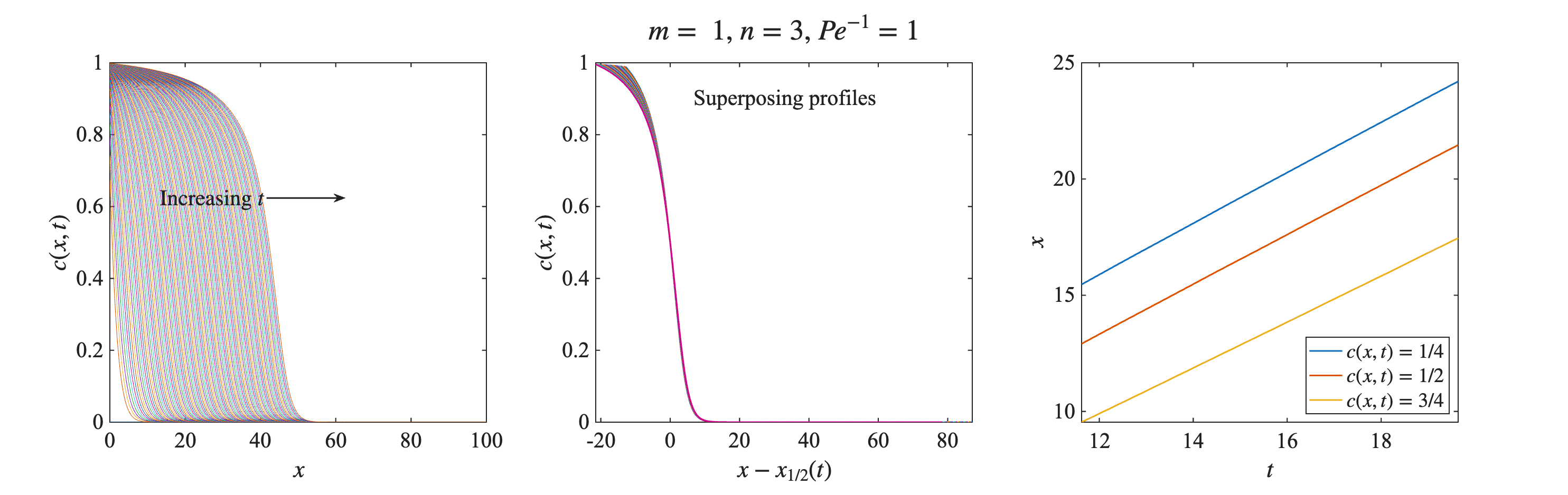}
\caption{Solution of system \eqref{eq:Sist_ND} with $q_e =0.9$ and $\Da = 0.1$. Left: evolution in time of $c(x,t)$. Middle: shifted curves to force the intersections at $(x-x_{1/2},c)=(0,1/2)$. Right: level sets for $c(x,t)=c^*$. Initial transient curves have been excluded from the middle and right plots.}
\label{fig:PDE4}
\end{figure}

\begin{figure}
\centering
\includegraphics[width=\textwidth]{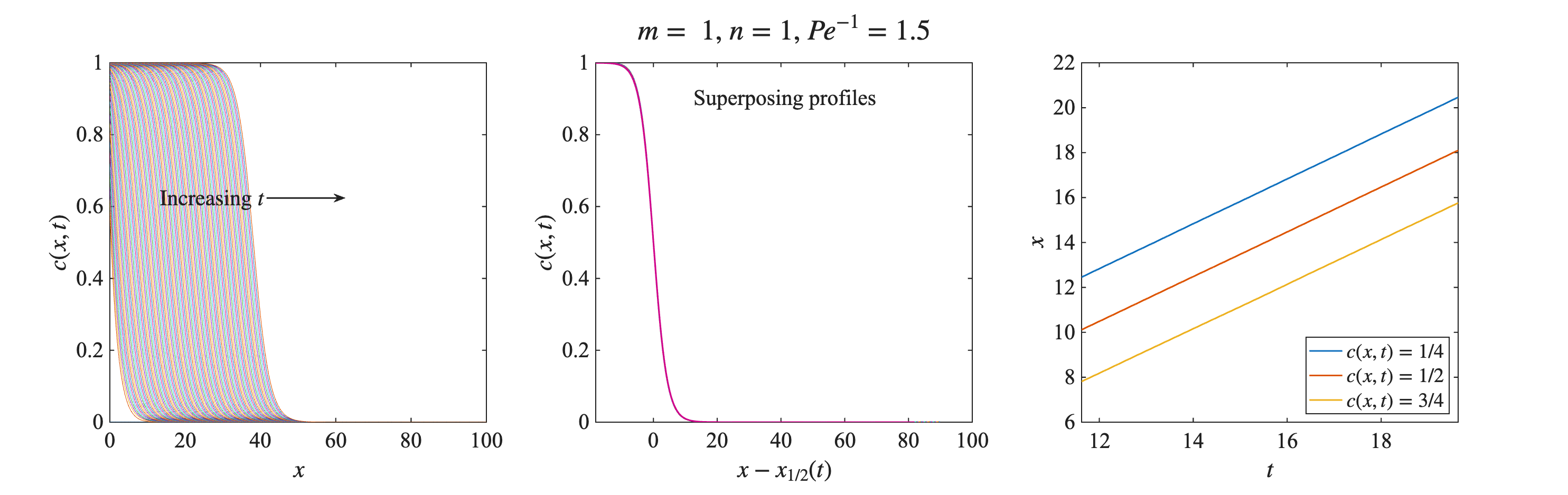}
\includegraphics[width=\textwidth]{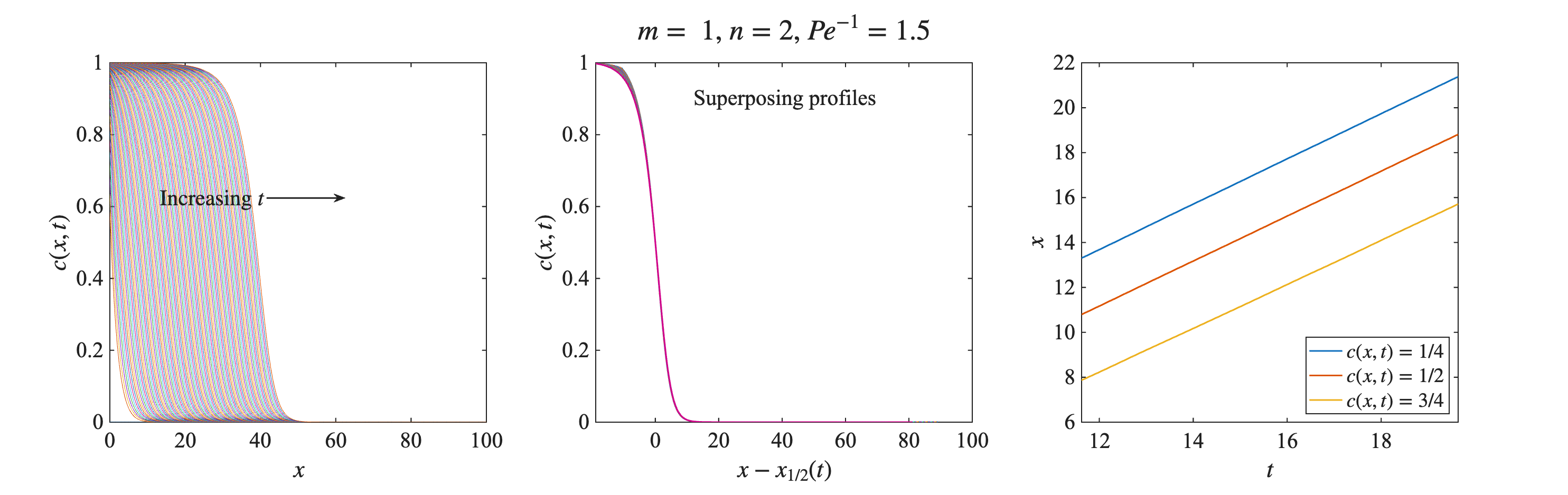}
\includegraphics[width=\textwidth]{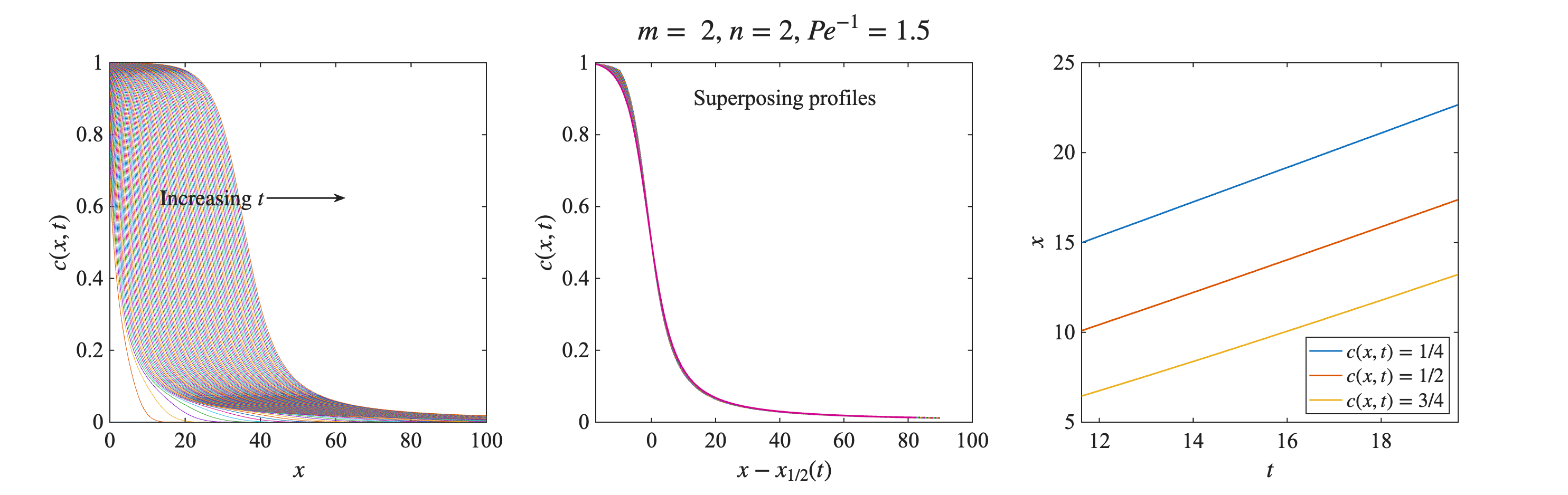}
\includegraphics[width=\textwidth]{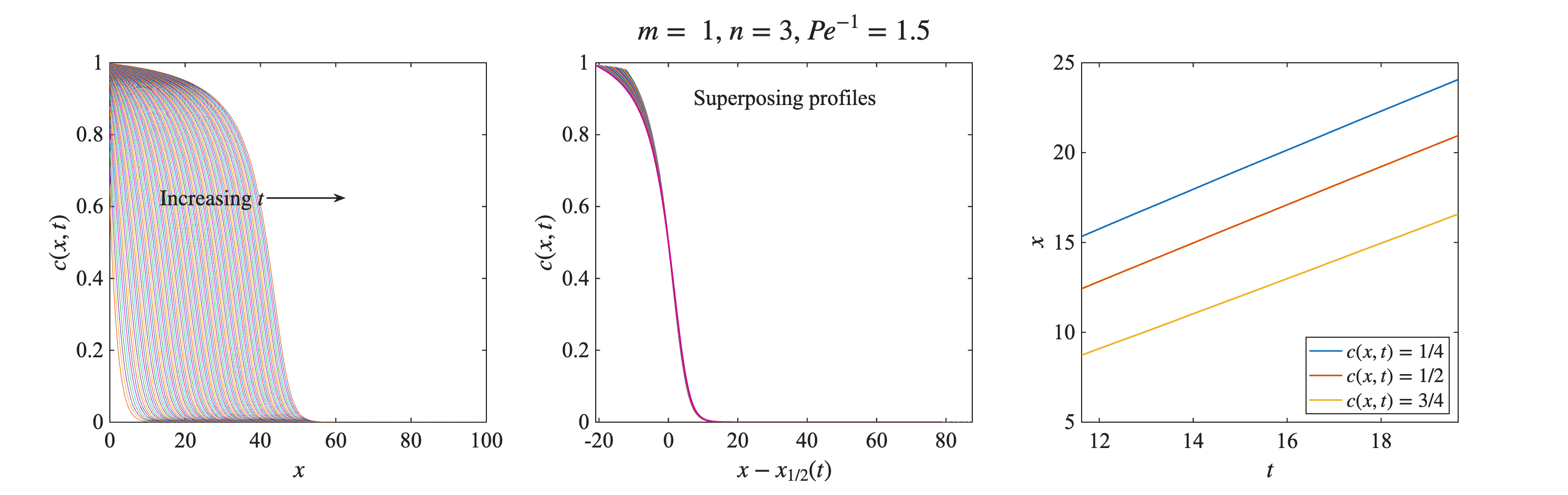}
\caption{Solution of system \eqref{eq:Sist_ND} with $q_e =0.9$ and $\Da = 0.1$. Left: evolution in time of $c(x,t)$. Middle: shifted curves to force the intersections at $(x-x_{1/2},c)=(0,1/2)$. Right: level sets for $c(x,t)=c^*$. Initial transient curves have been excluded from the middle and right plots.}
\label{fig:PDE5}
\end{figure}

\begin{figure}
\centering
\includegraphics[width=\textwidth]{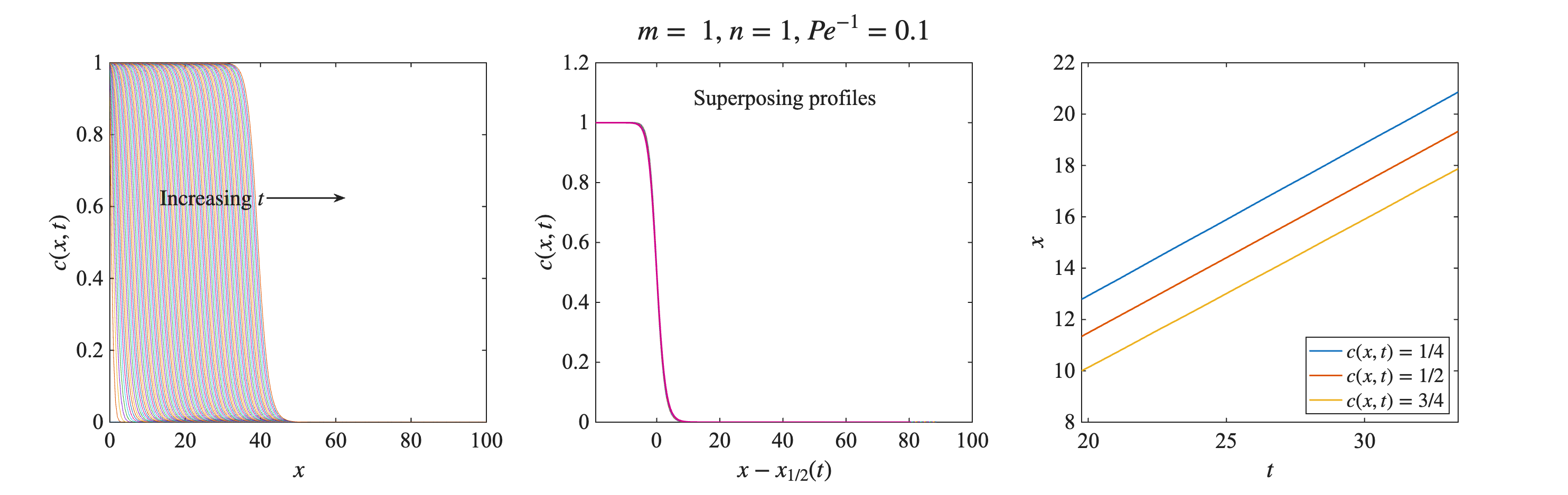}
\includegraphics[width=\textwidth]{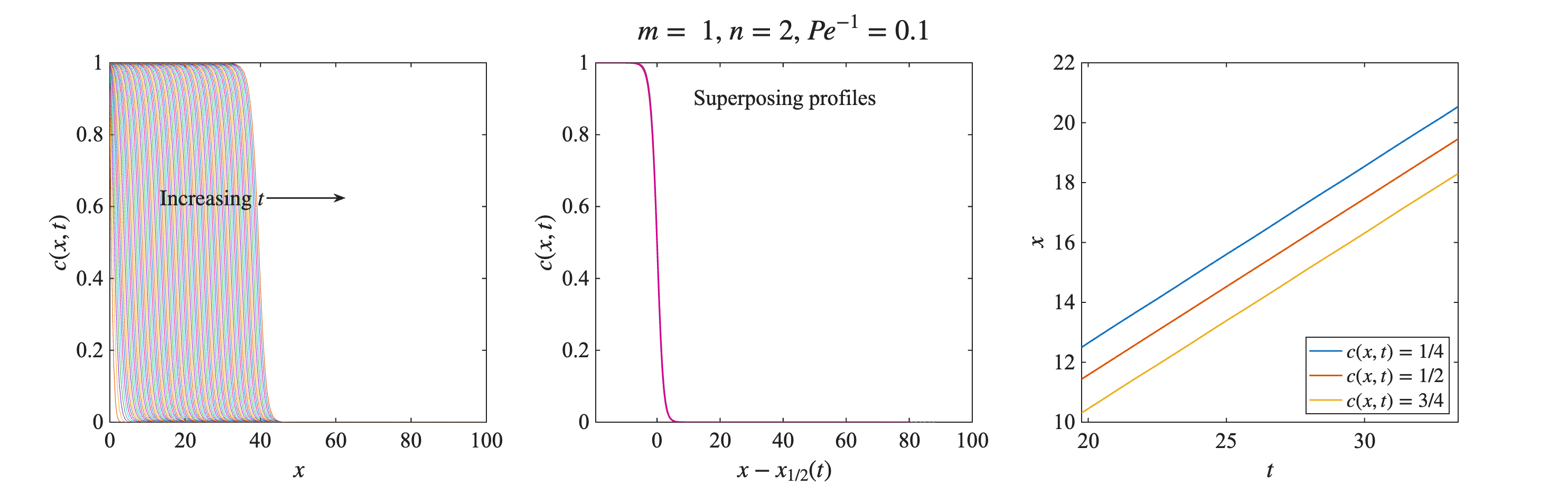}
\includegraphics[width=\textwidth]{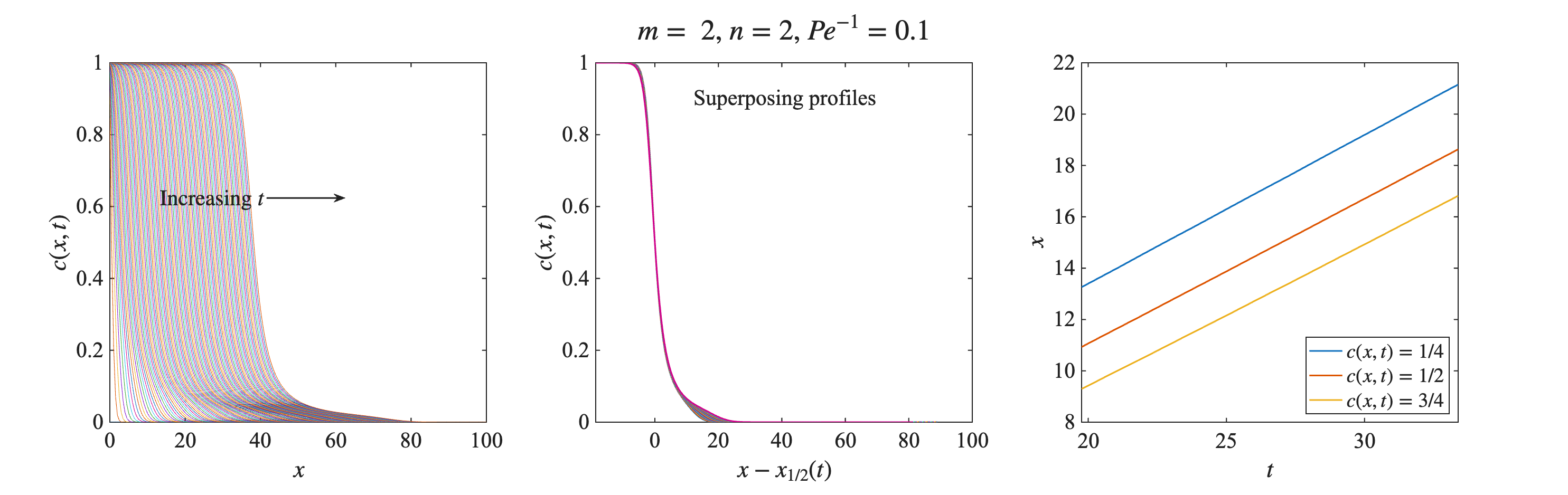}
\includegraphics[width=\textwidth]{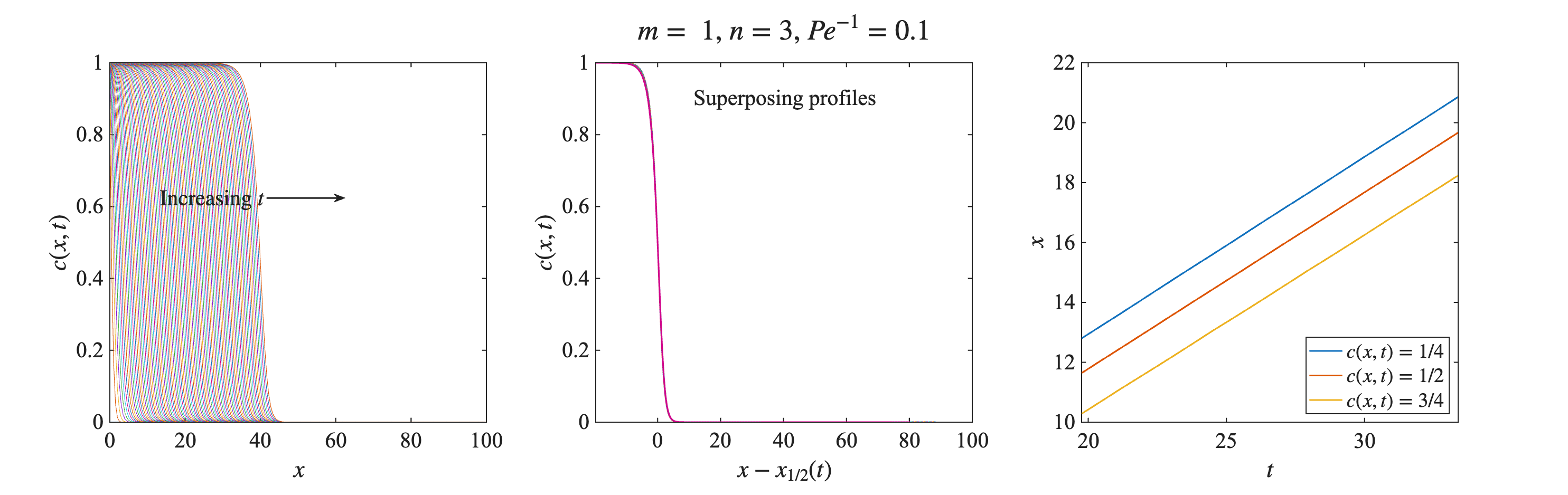}
\caption{Solution of system \eqref{eq:Sist_ND} with $q_e =0.7$ and $\Da = 1$. Initial transient curves have been excluded from the second row.}
\label{fig:PDE6}
\end{figure}

\begin{figure}
\centering
\includegraphics[width=\textwidth]{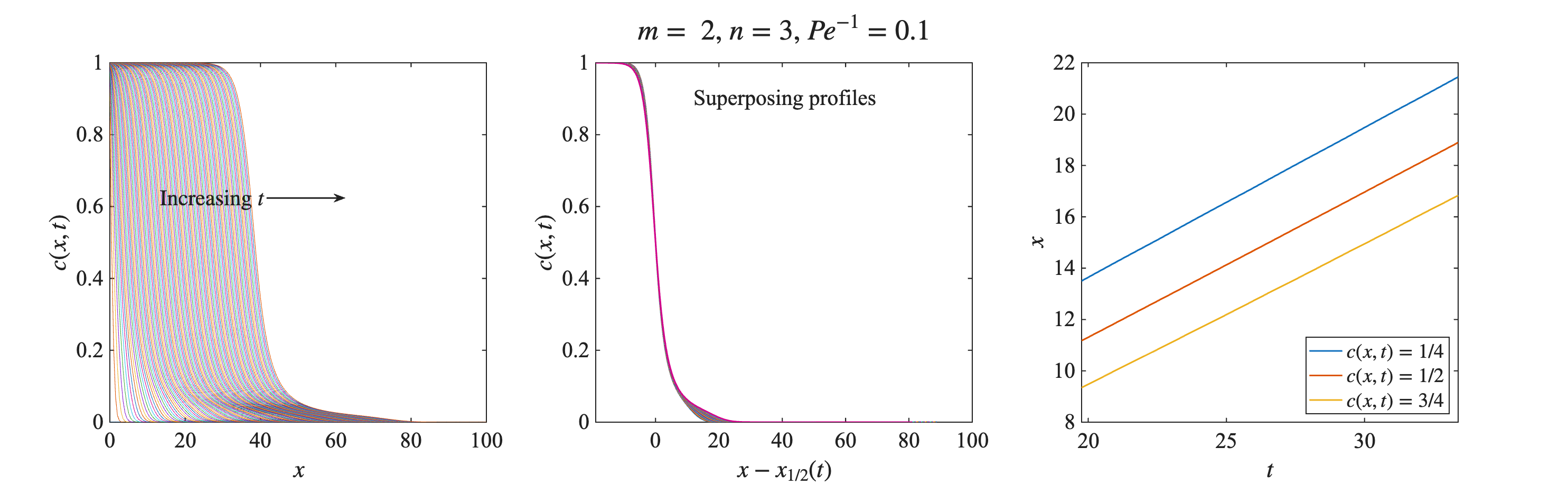}
\includegraphics[width=\textwidth]{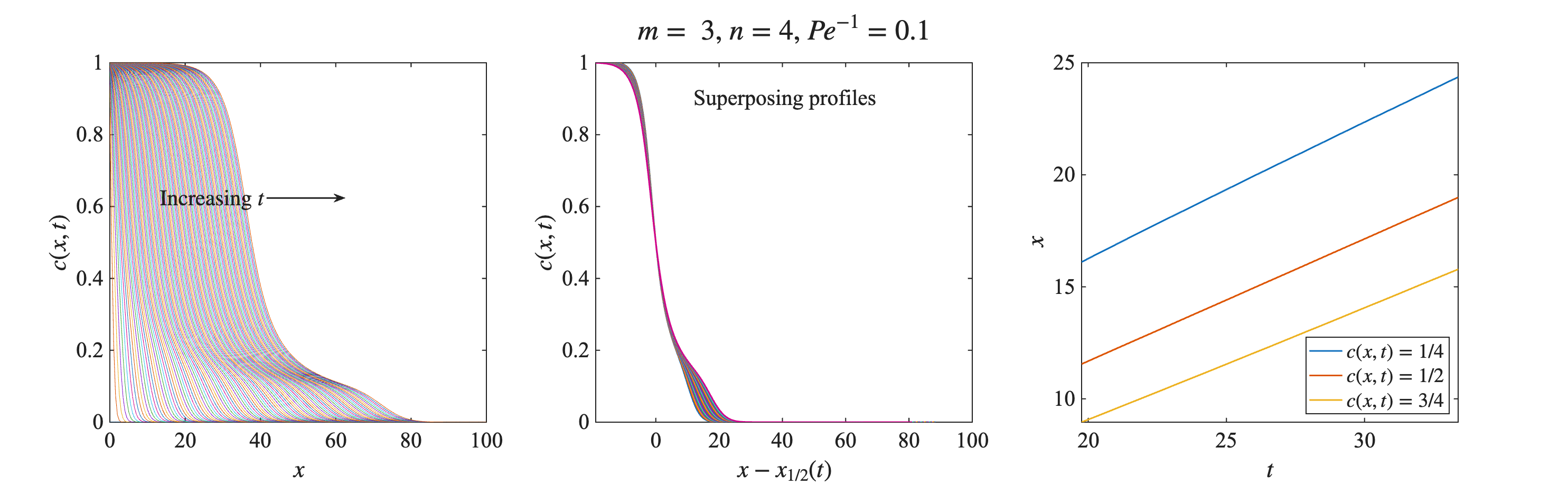}
\includegraphics[width=\textwidth]{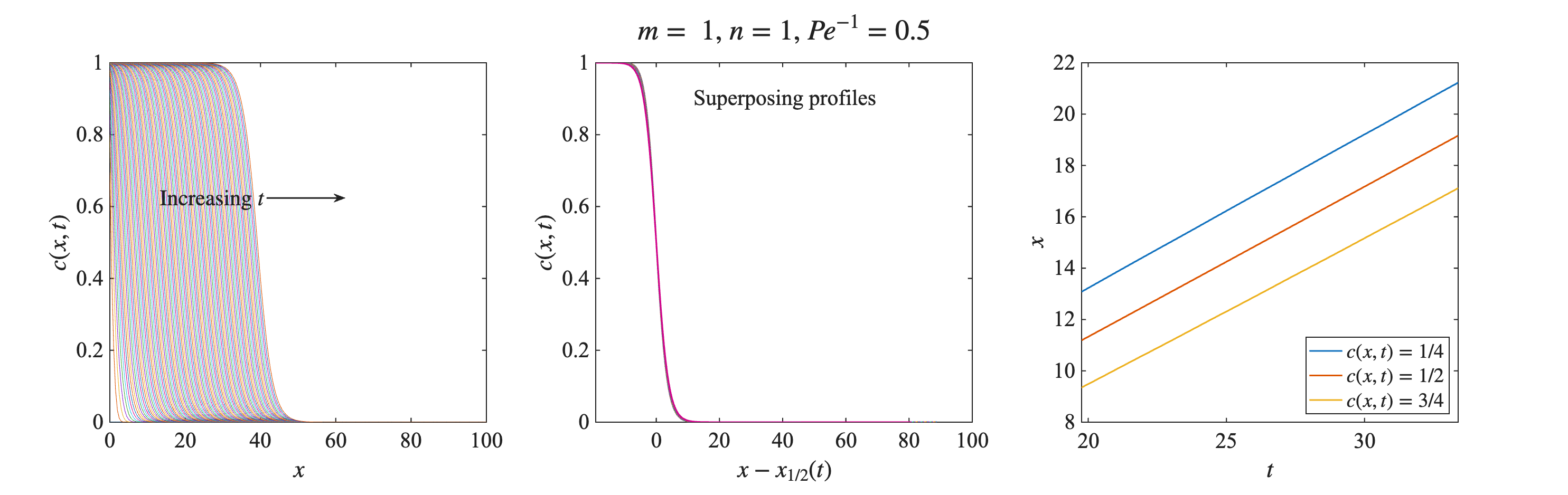}
\includegraphics[width=\textwidth]{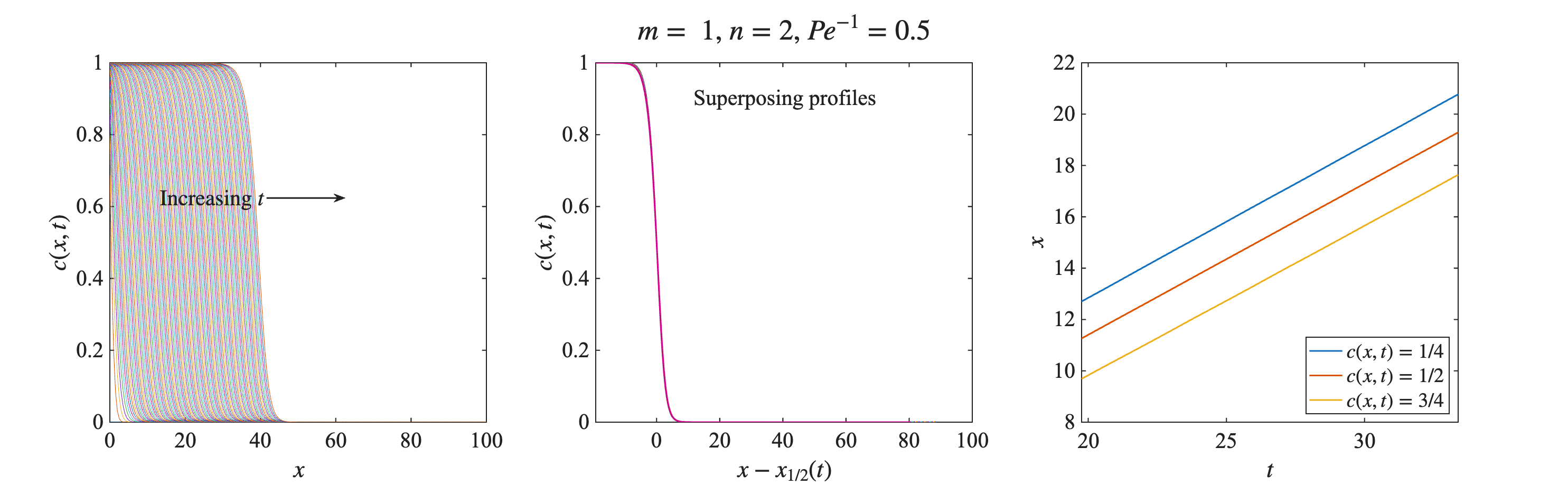}
\caption{Solution of system \eqref{eq:Sist_ND} with $q_e =0.7$ and $\Da = 1$. Left: evolution in time of $c(x,t)$. Middle: shifted curves to force the intersections at $(x-x_{1/2},c)=(0,1/2)$. Right: level sets for $c(x,t)=c^*$. Initial transient curves have been excluded from the middle and right plots.}
\label{fig:PDE7}
\end{figure}

\begin{figure}
\centering
	\includegraphics[width=\textwidth]{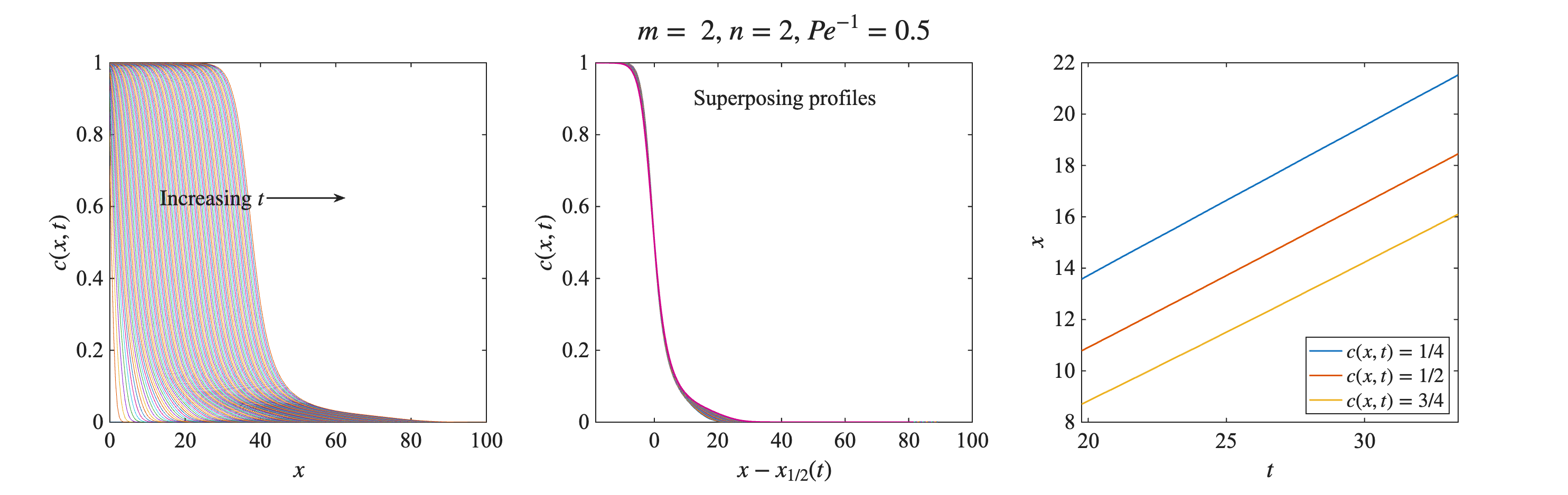}
\includegraphics[width=\textwidth]{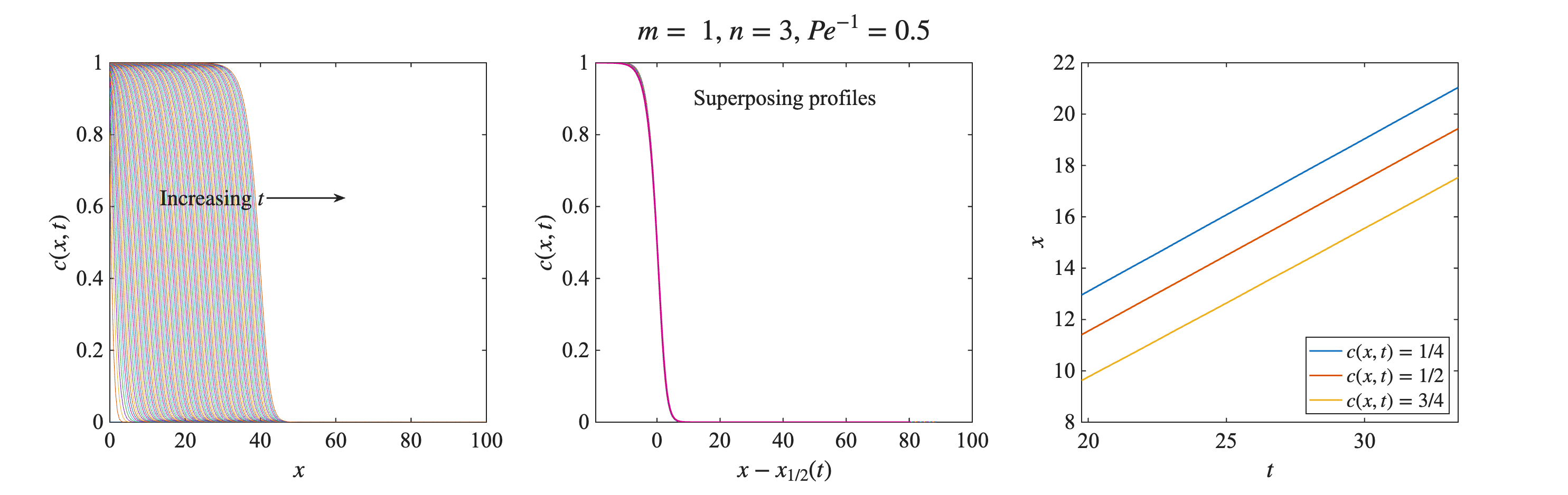}
	\includegraphics[width=\textwidth]{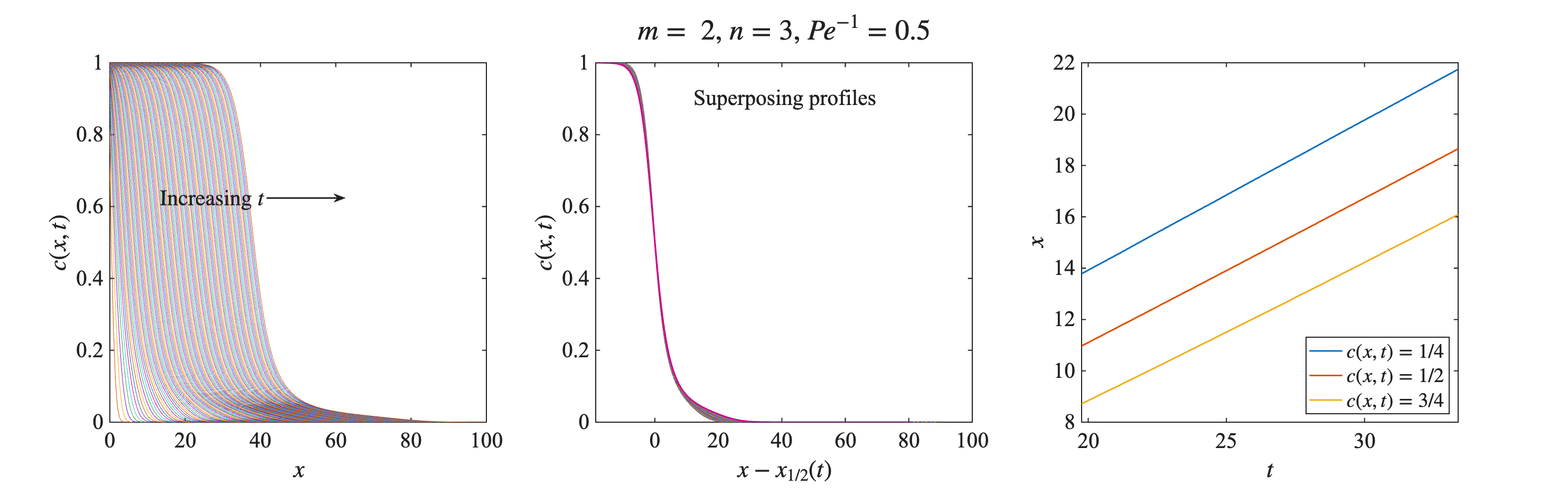}
\includegraphics[width=\textwidth]{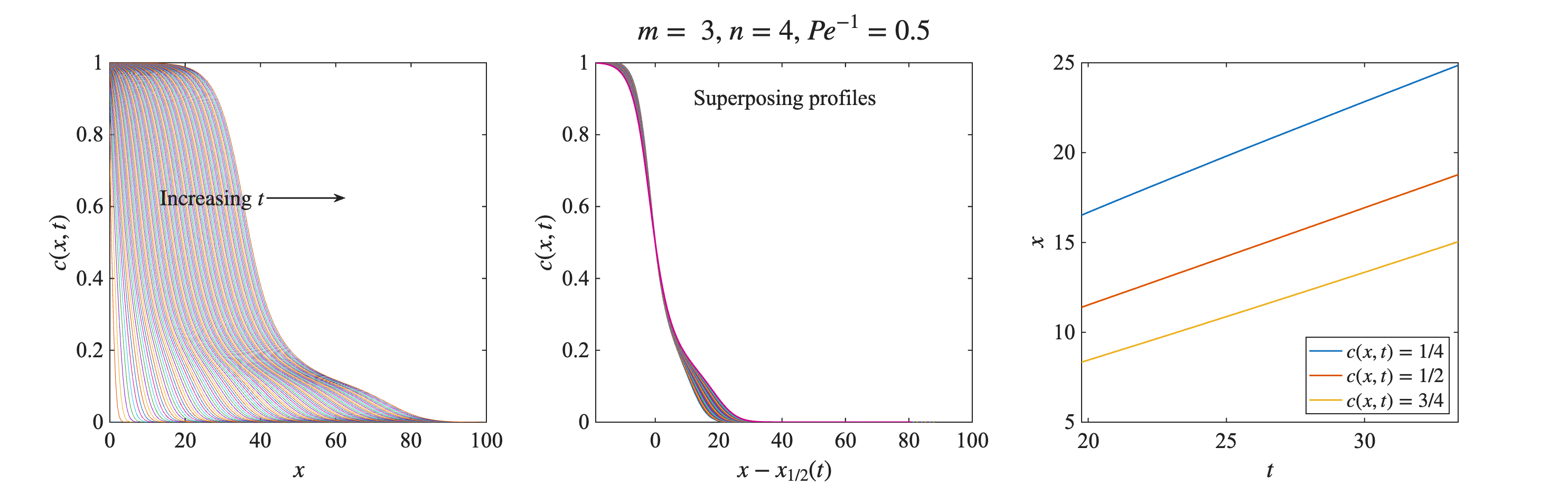}
\caption{Solution of system \eqref{eq:Sist_ND} with $q_e =0.7$ and $\Da = 1$. Left: evolution in time of $c(x,t)$. Middle: shifted curves to force the intersections at $(x-x_{1/2},c)=(0,1/2)$. Right: level sets for $c(x,t)=c^*$. Initial transient curves have been excluded from the middle and right plots.}
\label{fig:PDE8}
\end{figure}

\begin{figure}
\centering
\includegraphics[width=\textwidth]{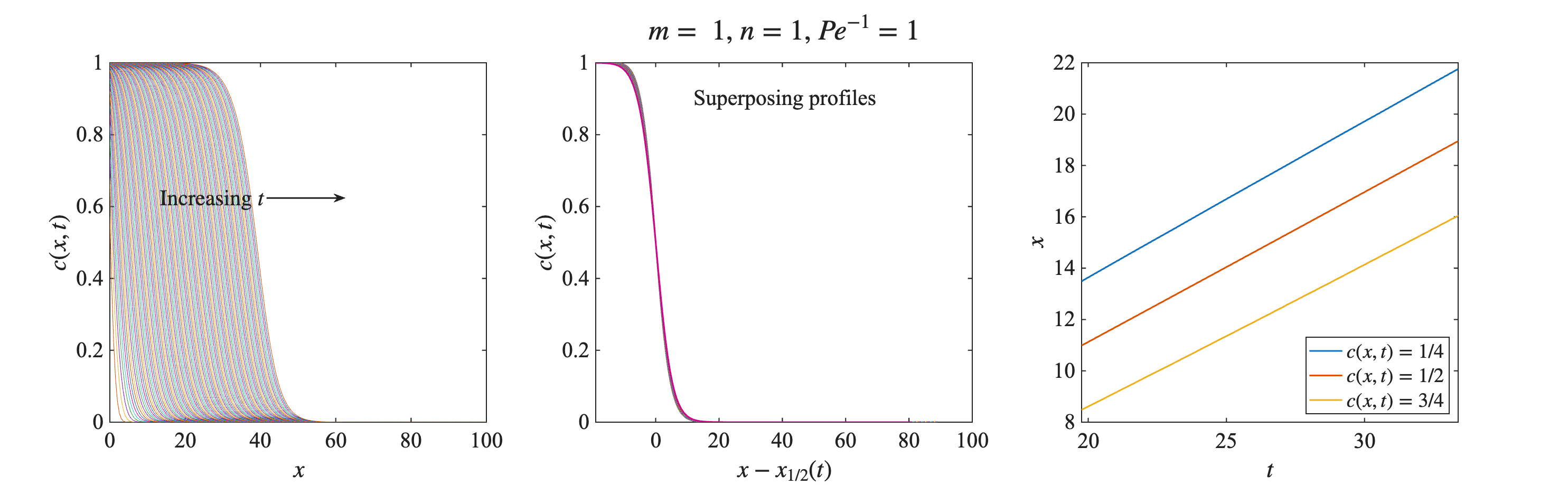}
\includegraphics[width=\textwidth]{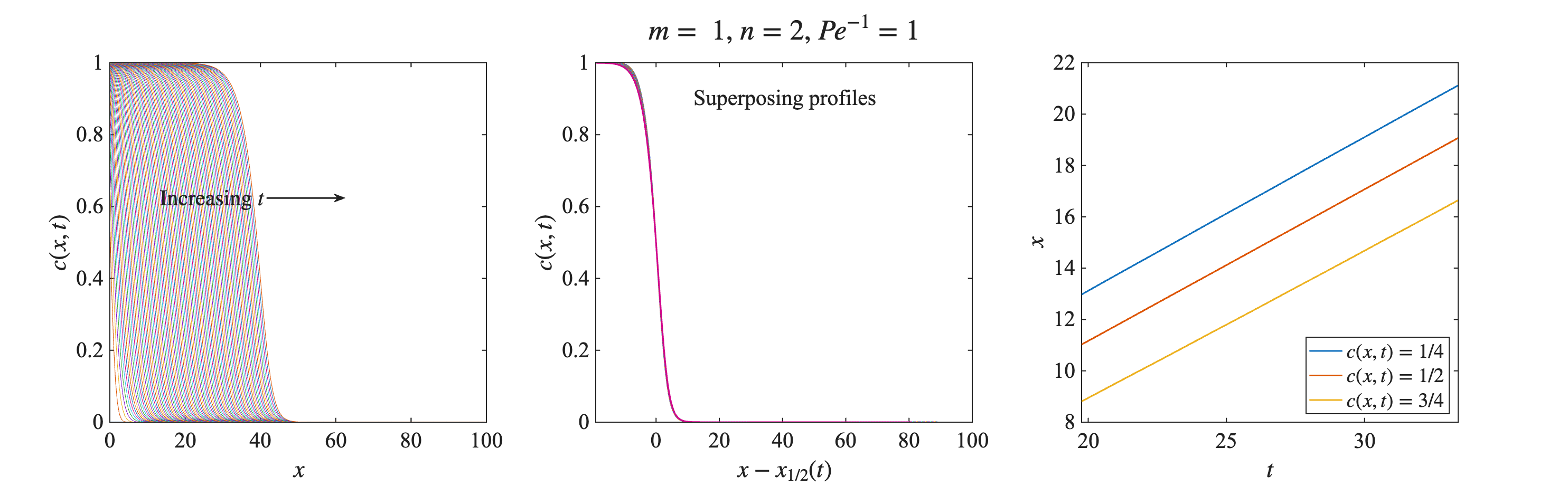}
\includegraphics[width=\textwidth]{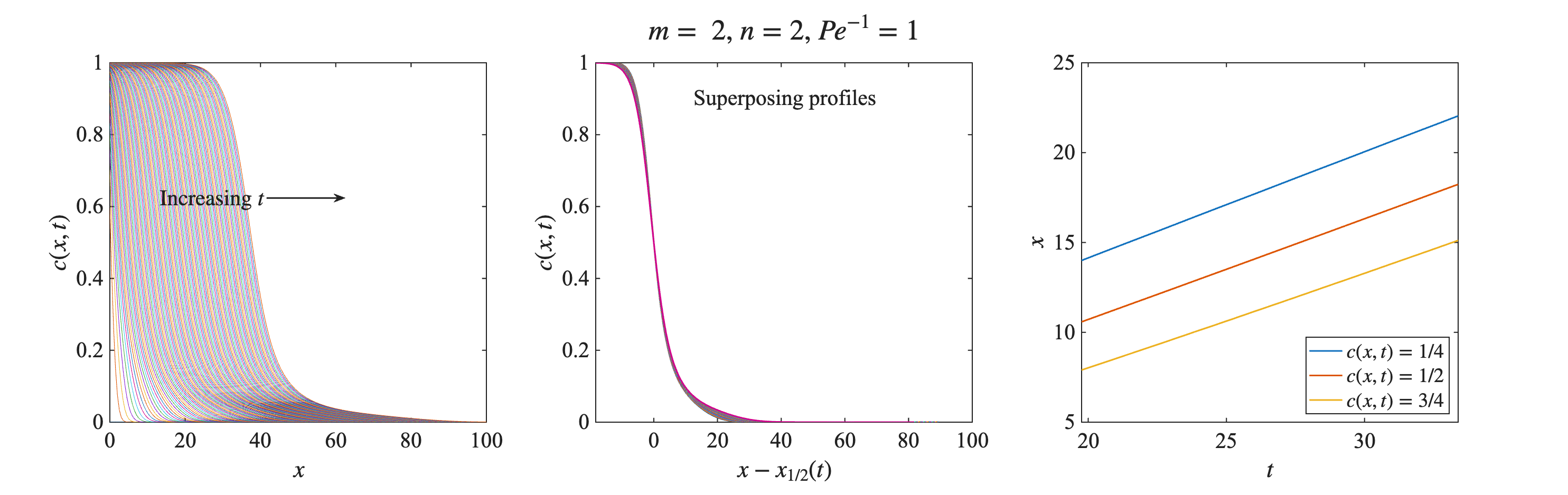}
\includegraphics[width=\textwidth]{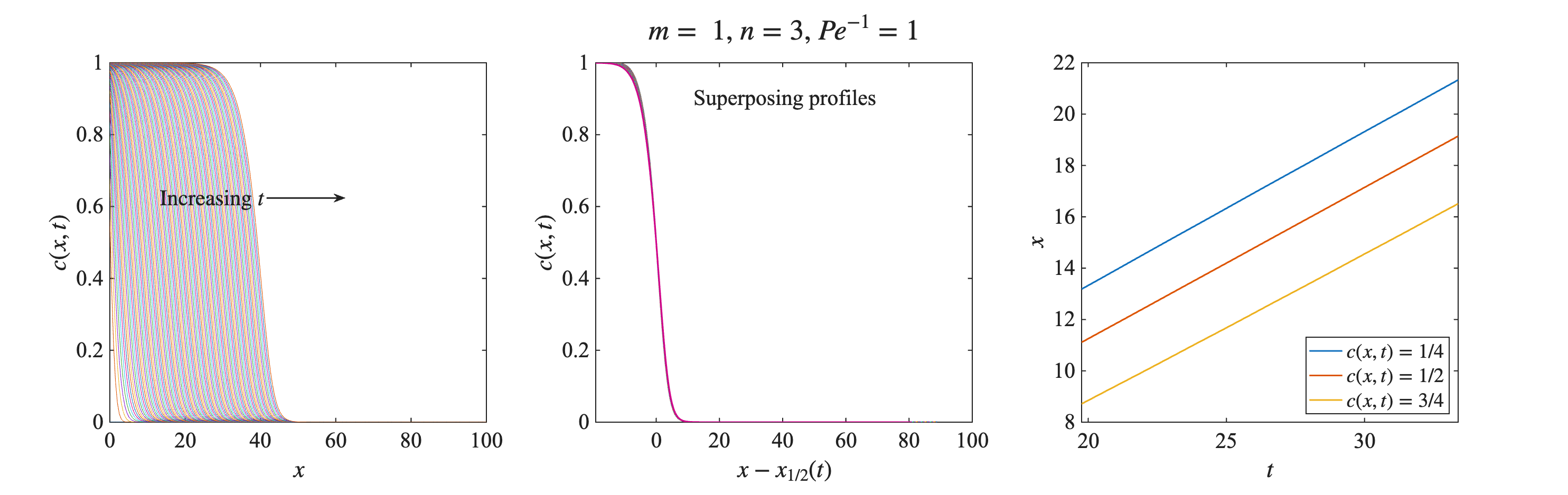}
\caption{Solution of system \eqref{eq:Sist_ND} with $q_e =0.7$ and $\Da = 1$. Left: evolution in time of $c(x,t)$. Middle: shifted curves to force the intersections at $(x-x_{1/2},c)=(0,1/2)$. Right: level sets for $c(x,t)=c^*$. Initial transient curves have been excluded from the middle and right plots.}
\label{fig:PDE9}
\end{figure}

\begin{figure}
\centering
\includegraphics[width=\textwidth]{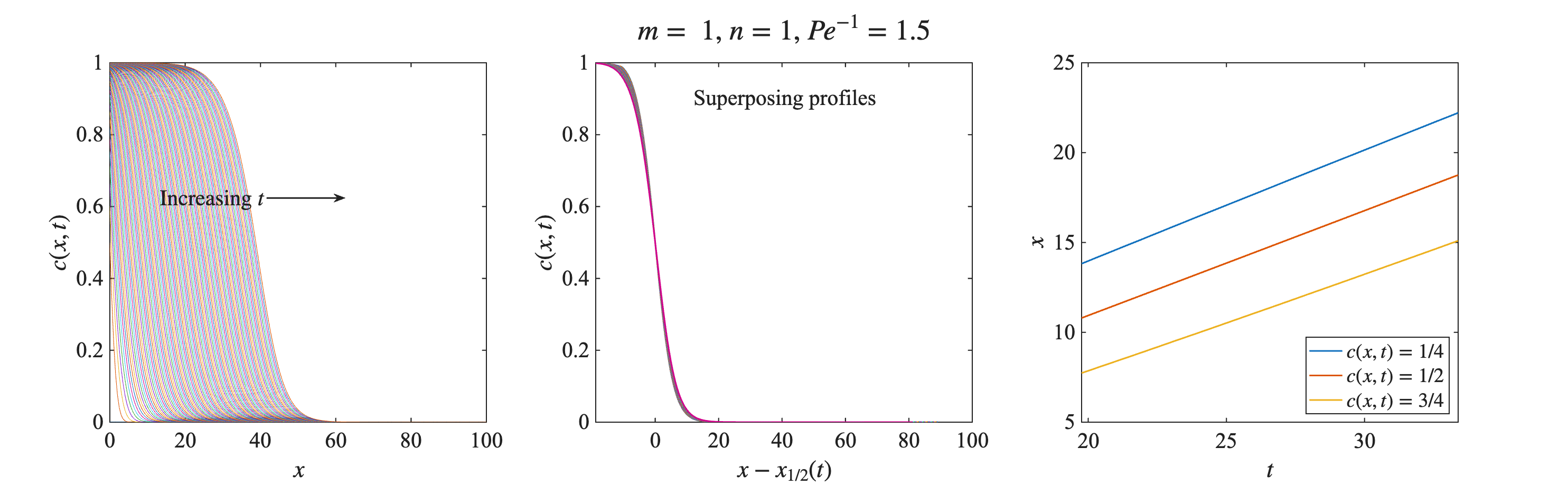}
\includegraphics[width=\textwidth]{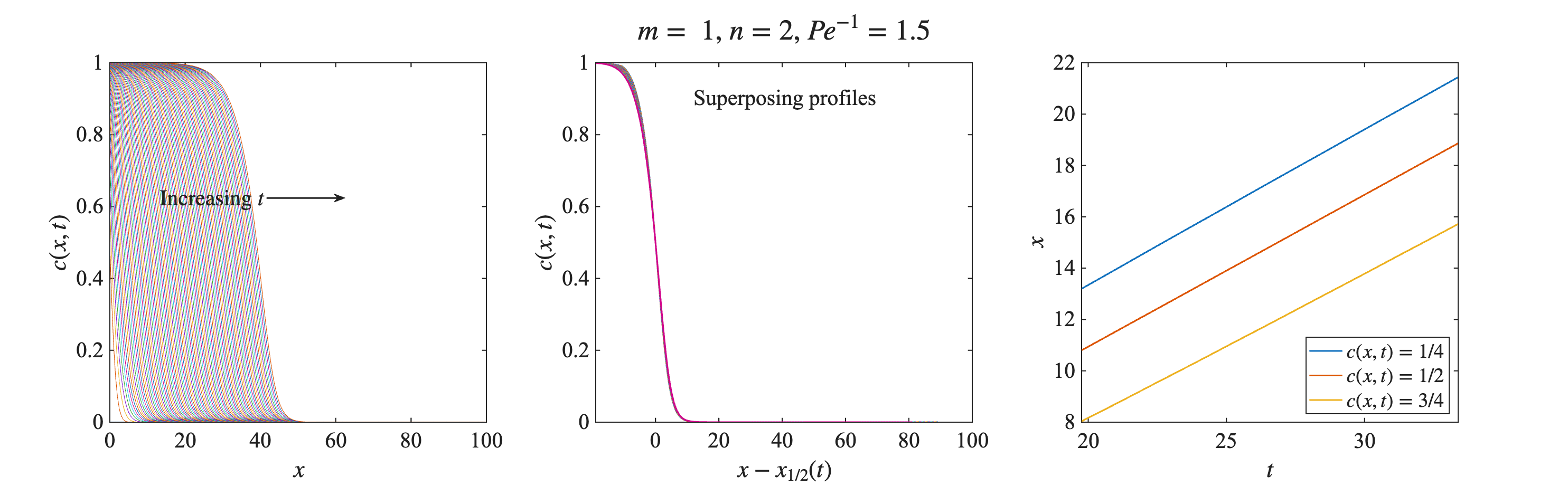}
\includegraphics[width=\textwidth]{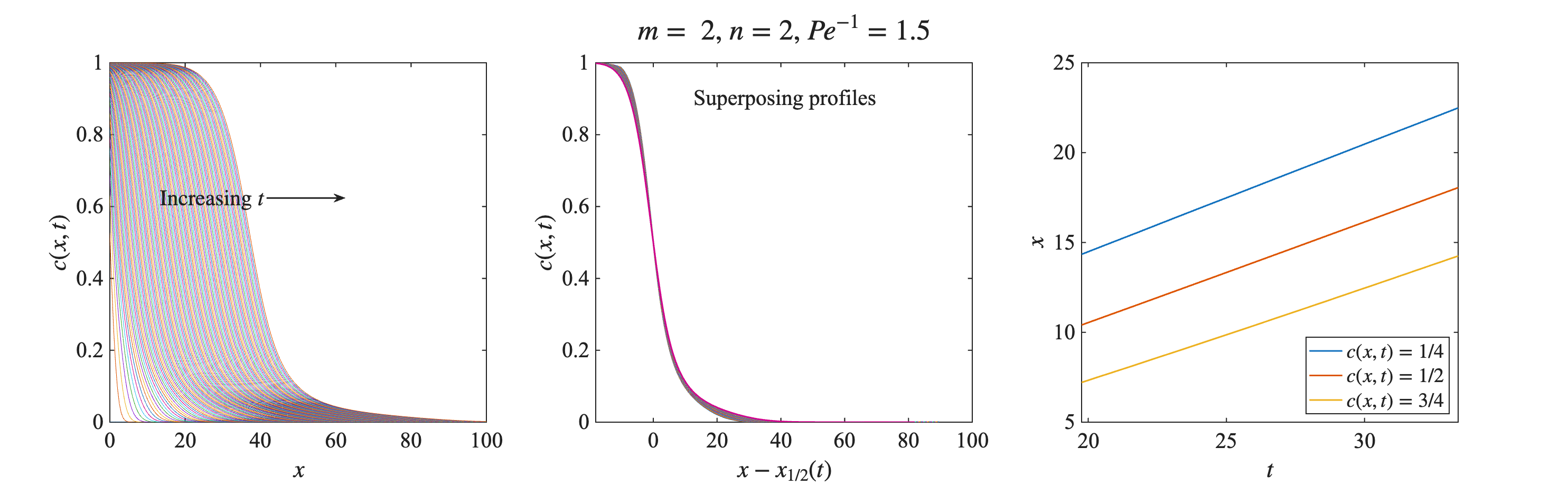}
\includegraphics[width=\textwidth]{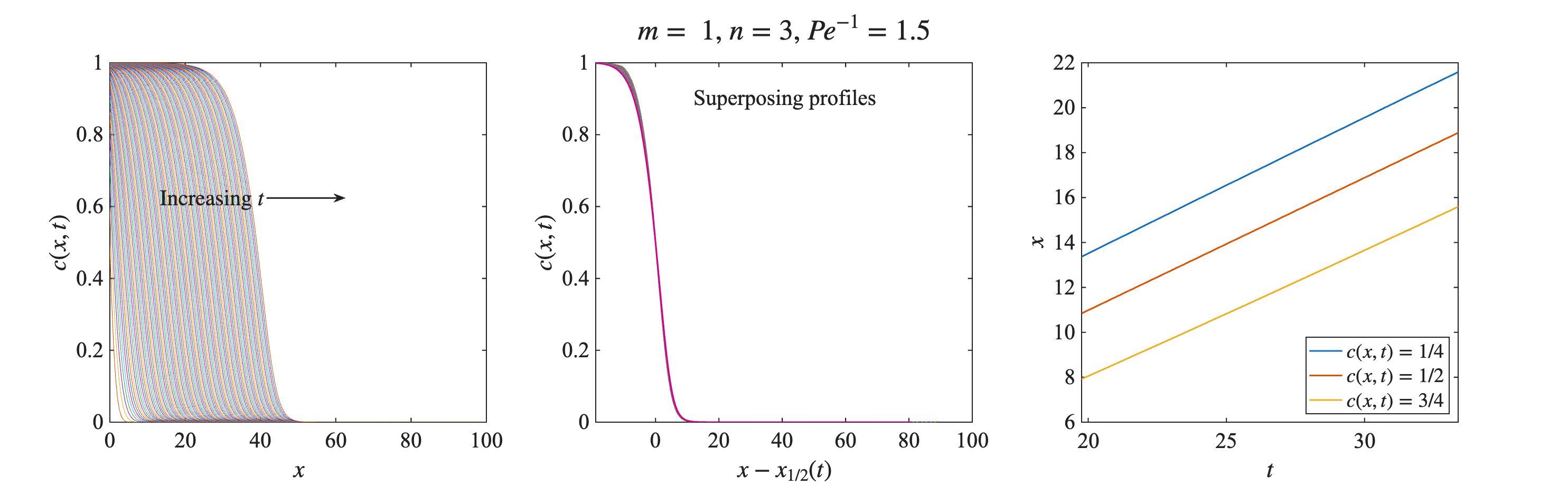}
\caption{Solution of system \eqref{eq:Sist_ND} with $q_e =0.7$ and $\Da = 1$. Left: evolution in time of $c(x,t)$. Middle: shifted curves to force the intersections at $(x-x_{1/2},c)=(0,1/2)$. Right: level sets for $c(x,t)=c^*$. Initial transient curves have been excluded from the middle and right plots.}
\label{fig:PDE10}
\end{figure}

\begin{figure}
	\centering
	\includegraphics[scale=0.3]{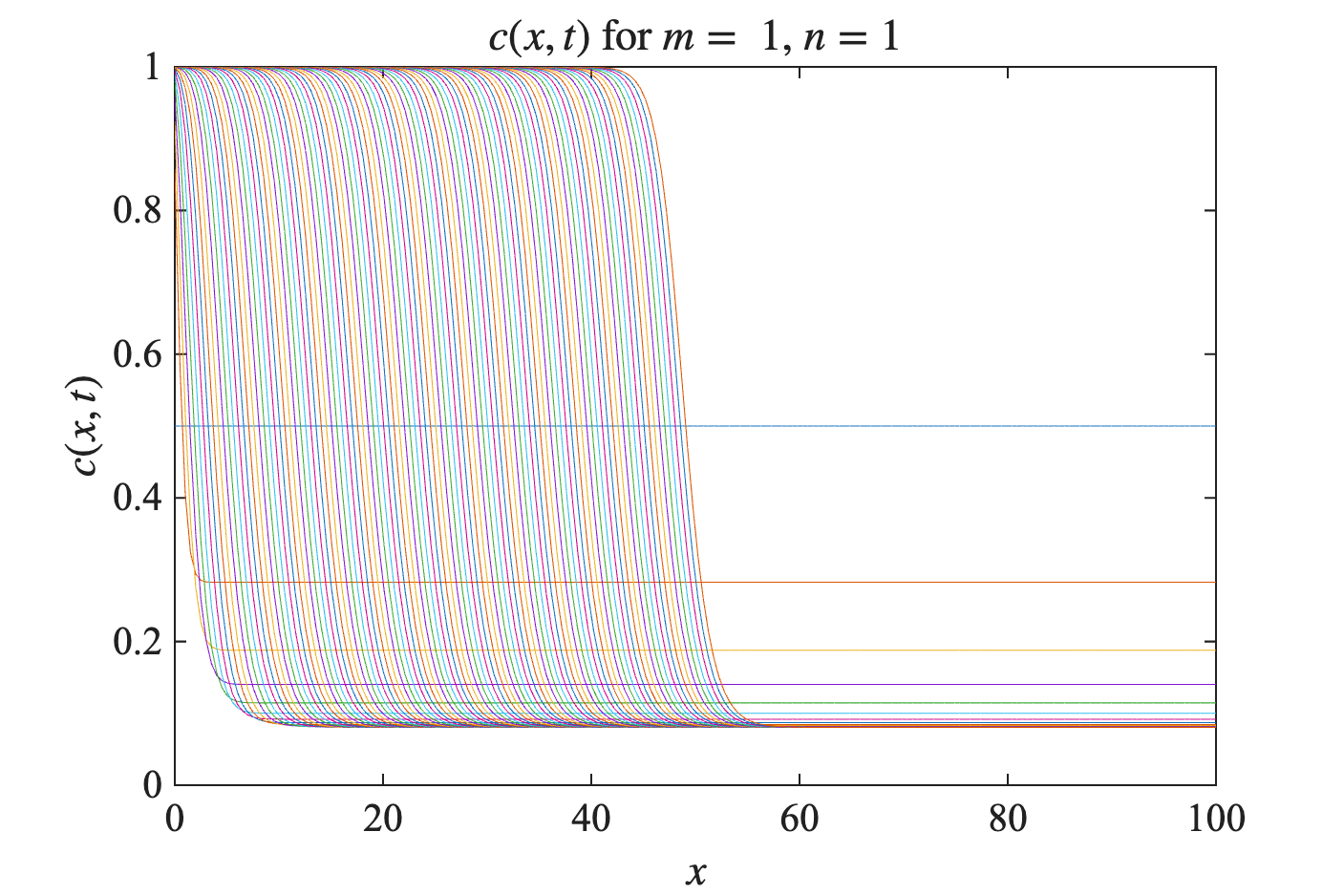}
	\includegraphics[scale=0.3]{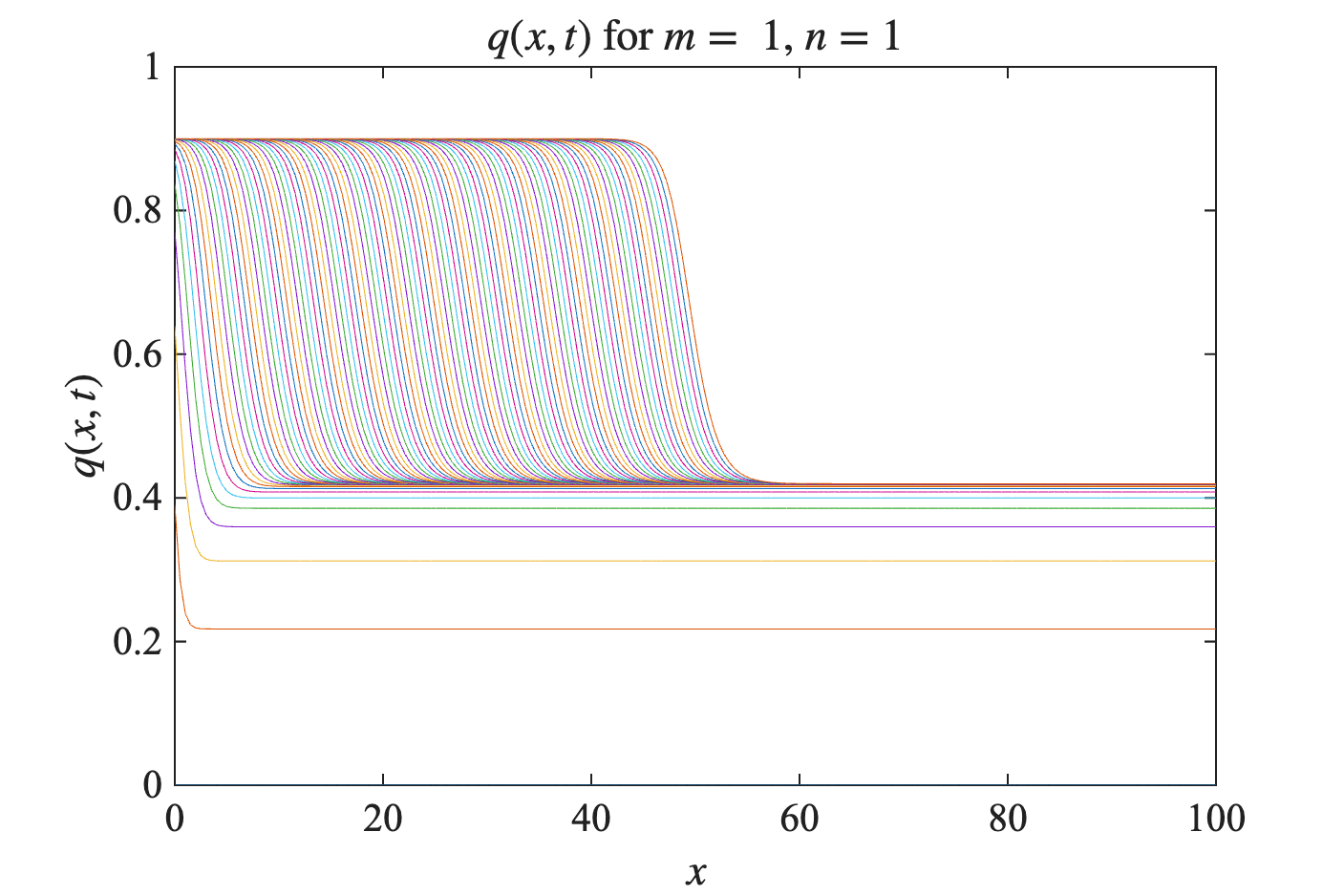}
	\includegraphics[scale=0.3]{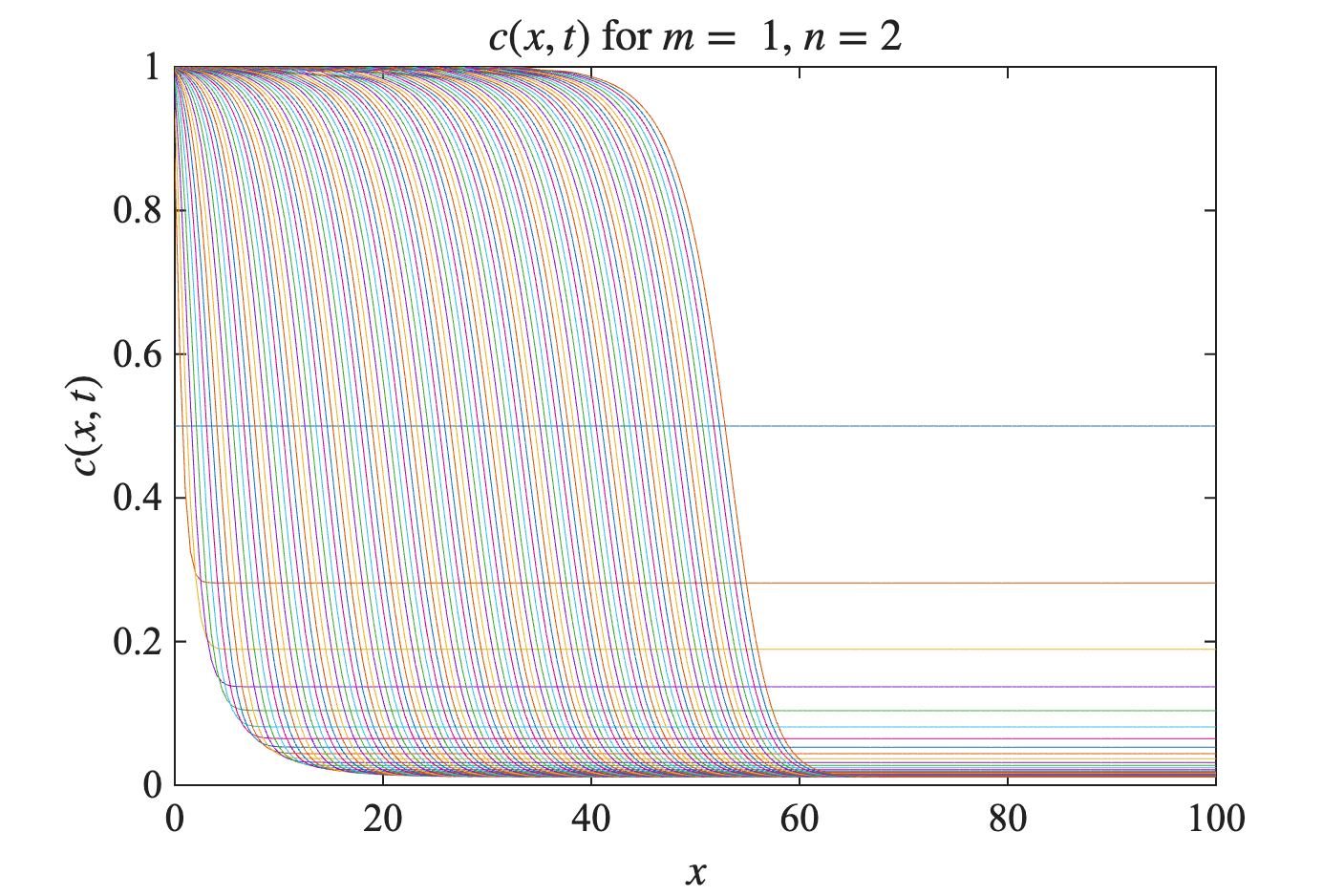}
	\includegraphics[scale=0.3]{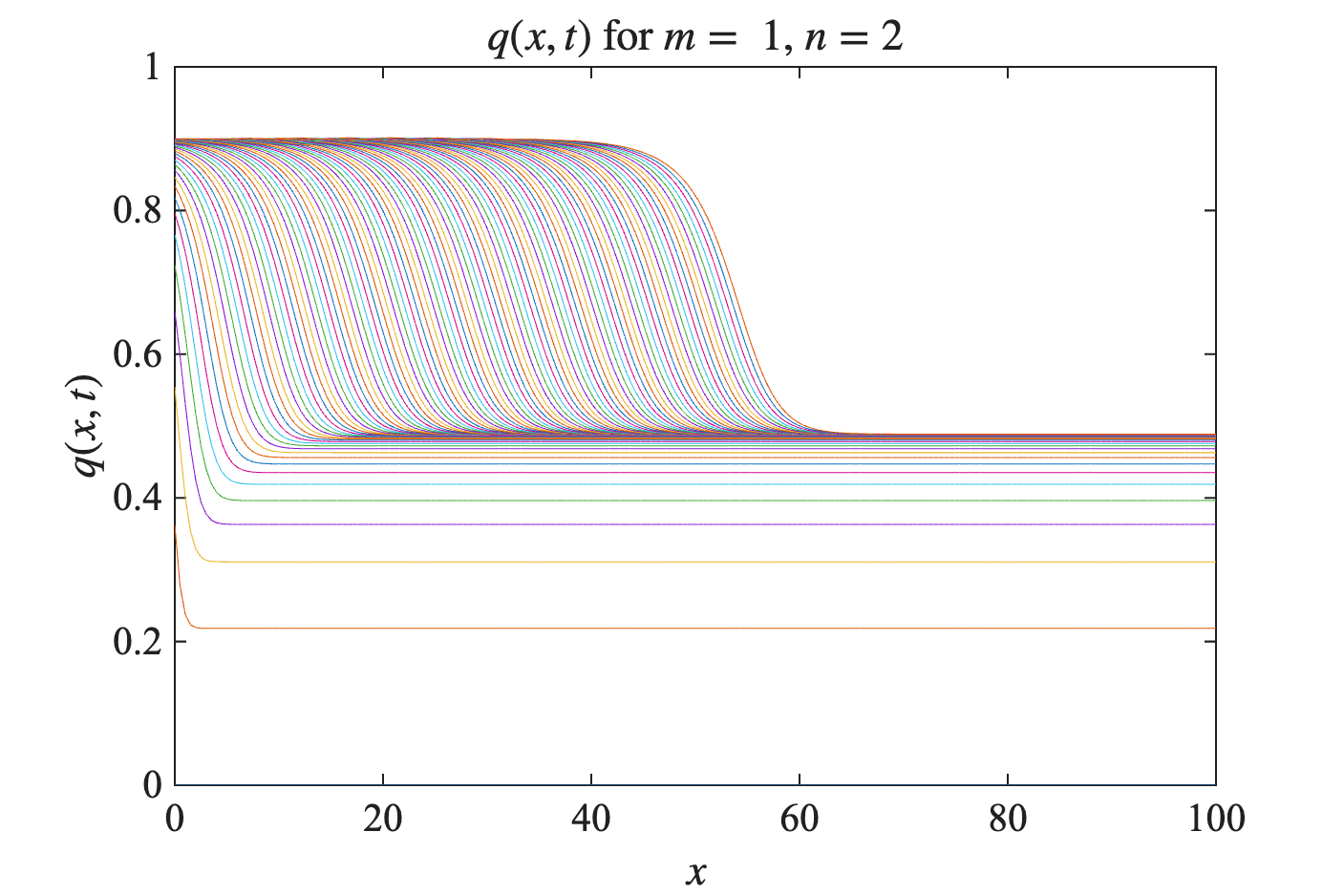}
	\includegraphics[scale=0.3]{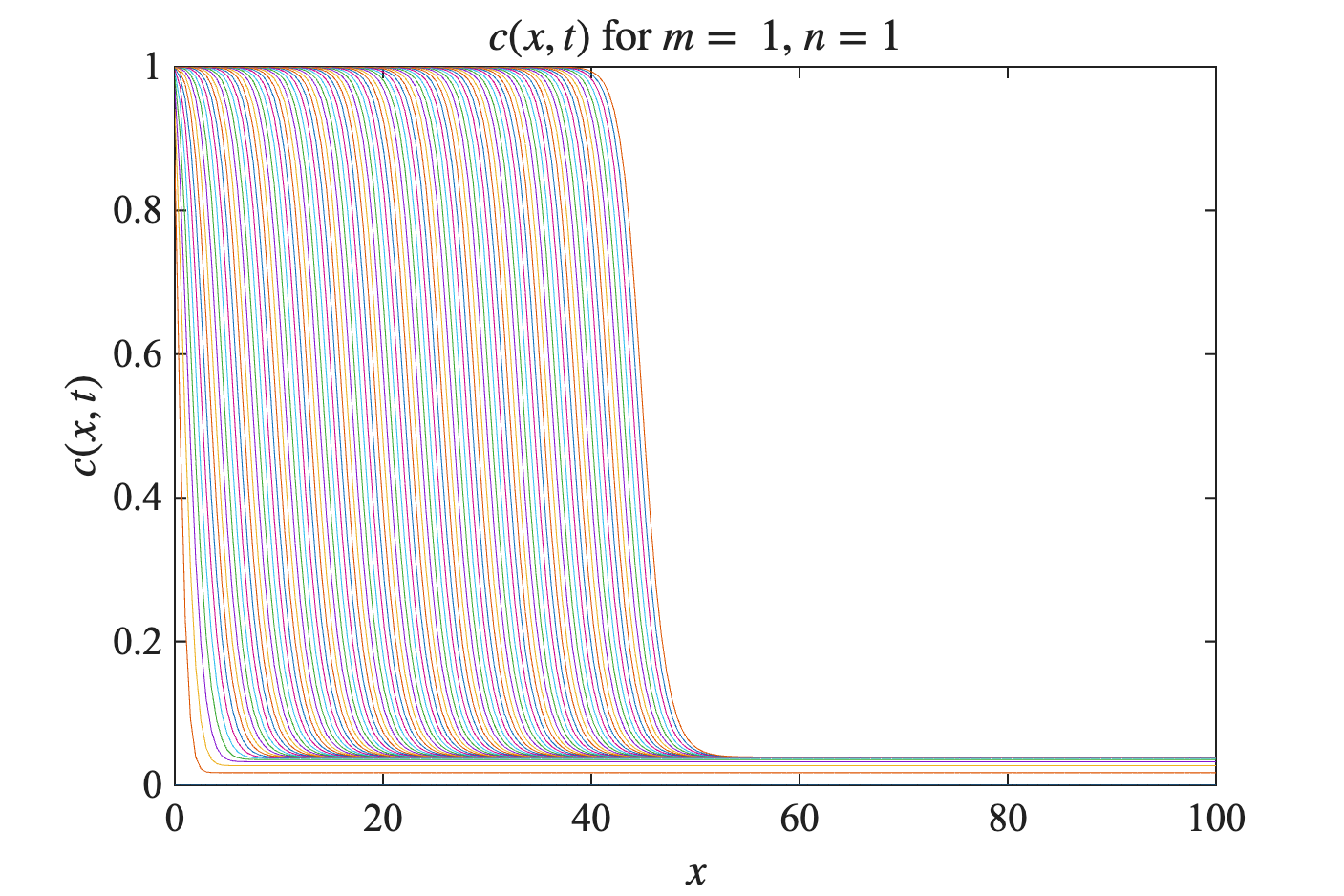}
	\includegraphics[scale=0.3]{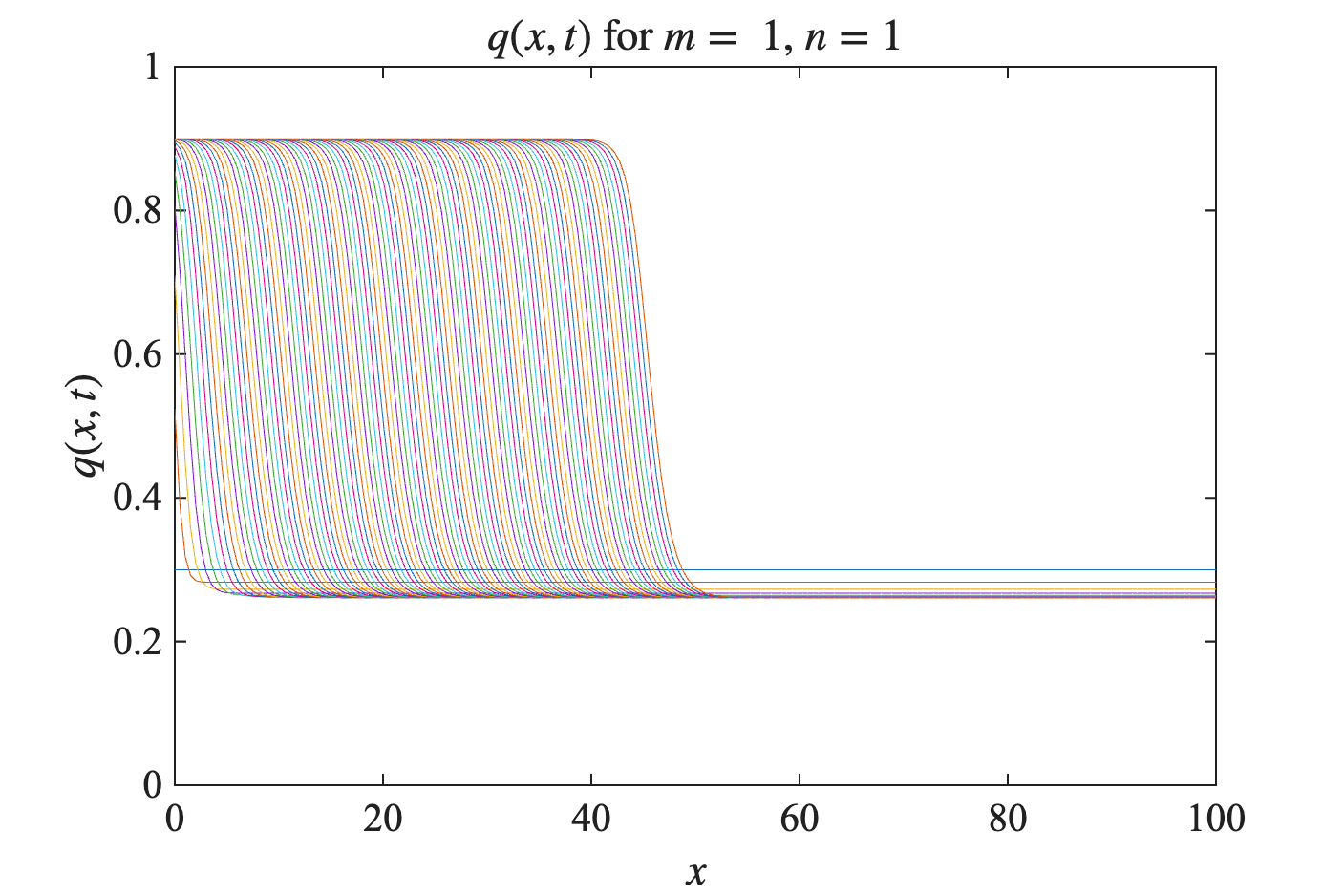}
	\includegraphics[scale=0.3]{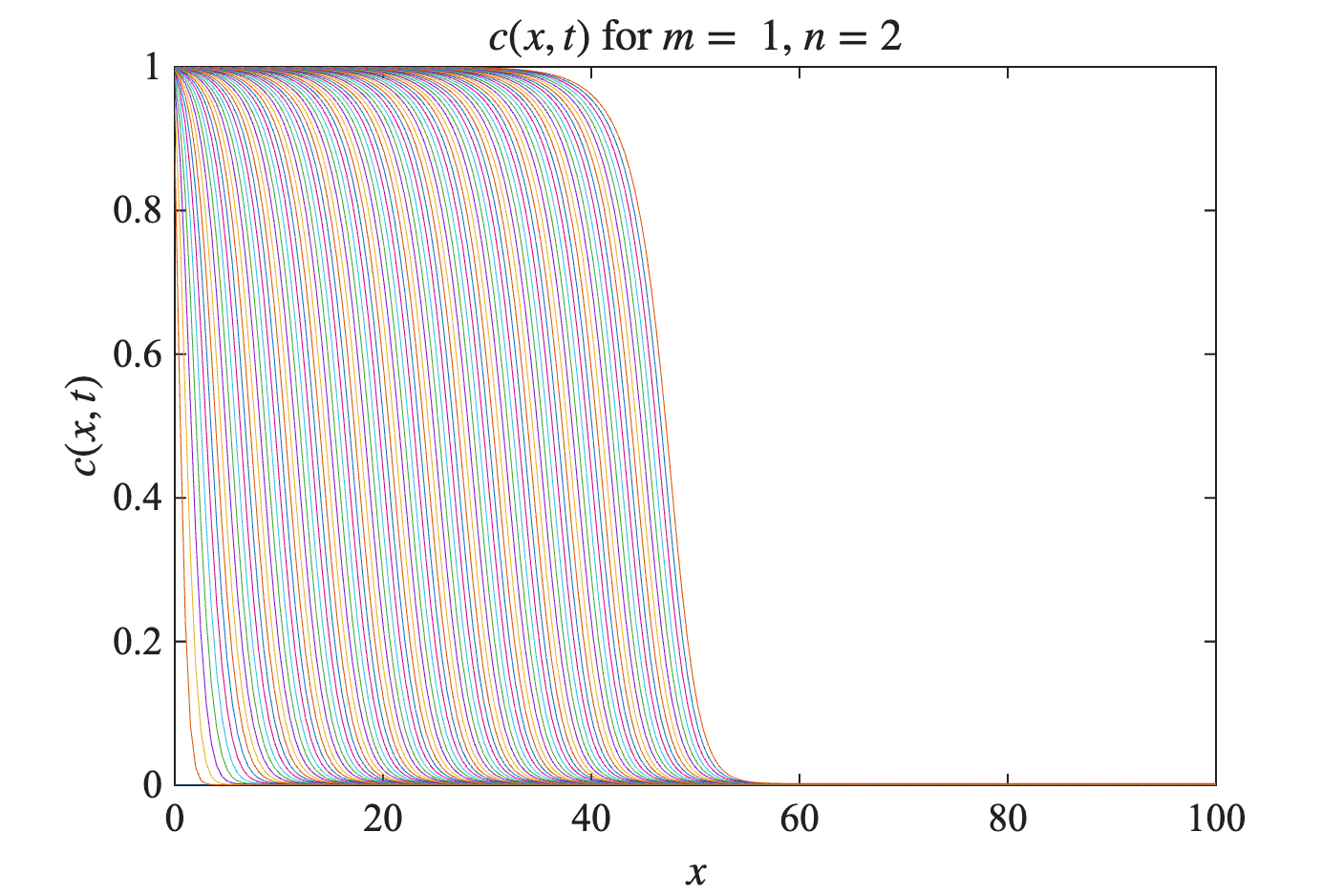}
	\includegraphics[scale=0.3]{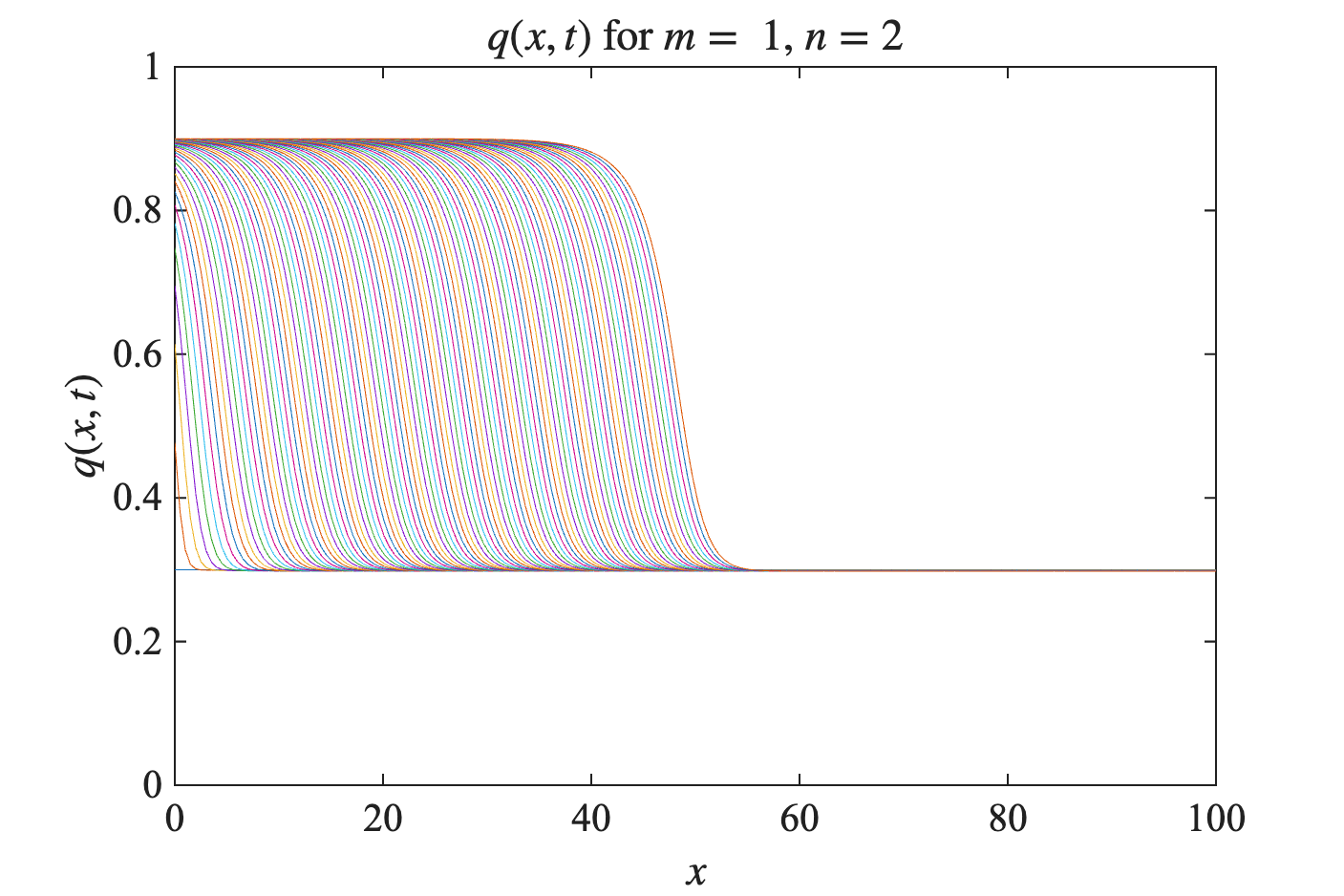}
	\caption{Solution of system \eqref{eq:Sist_ND} with $q_e =0.9$, $\Pe = 0.1$, and $\Da = 1$ for non-zero initial conditions. First and second rows: $c(x,0)=0.5$, $q(x,0)=0$.  Third and fourth rows: $c(x,0)=0$, $q(x,0)=0.3$.}
	\label{Fig:NonZIC}
\end{figure}

\end{document}